\documentclass[11pt, draftclsnofoot, journal, onecolumn]{IEEEtran}

\usepackage{amsmath,stmaryrd,acronym,amssymb,amsthm,dsfont,graphicx, mathtools}

\graphicspath{{graphics/}} 

\acrodef{AEP}{Asymptotic Equipartition Property}
\acrodef{AoA}{Angle of Arrival}
\acrodef{AWGN}{Additive White Gaussian Noise}
\acrodef{BER}{Bit-Error-Rate}
\acrodef{BEC}{Binary Erasure Channel}
\acrodef{BPSK}{Binary Phase-Shift Keying}
\acrodef{BSC}{Binary Symmetric Channel}
\acrodef{CDF}[CDF]{Cumulative Distribution Function}
\acrodef{CLT}[CLT]{Central Limit Theorem}
\acrodef{CSI}[CSI]{Channel State Information}
\acrodef{DMC}[DMC]{Discrete Memoryless Channel}
\acrodef{DMS}[DMS]{Discrete Memoryless Source} 
\acrodef{iid}[i.i.d.]{independent and identically distributed}
\acrodef{LPD}[LPD]{Low Probability of Detection}
\acrodef{LDPC}[LDPC]{Low-Density Parity-Check}
\acrodef{MAC}[MAC]{multiple-access channel}
\acrodef{MIMO}[MIMO]{Multiple-Input Multiple-Output}
\acrodef{MISO}{Multiple-Input Single-Output}
\acrodef{PDF}[PDF]{Probability Distribution Function}
\acrodef{PMF}[PMF]{Probability Mass Function}
\acrodef{PPM}[PPM]{Pulse Position Modulation}
\acrodef{PSD}{Power Spectral Density}
\acrodef{QPSK}{Quadrature Phase-Shift Keying}
\acrodef{SIMO}{Single-Input Multiple-Output}
\acrodef{SNR}{Signal-to-Noise Ratio}
\acrodef{wrt}[w.r.t.]{with respect to}
\acrodef{WSS}{Wide Sense Stationary}

\DeclareMathAlphabet{\eurm}{U}{eur}{m}{n}
\DeclareMathAlphabet{\mathbsf}{OT1}{cmss}{bx}{n}
\DeclareMathAlphabet{\mathssf}{OT1}{cmss}{m}{sl}
\DeclareMathAlphabet{\mathcsf}{OT1}{cmss}{sbc}{n}

\DeclareSymbolFont{bsfletters}{OT1}{cmss}{bx}{n}  
\DeclareSymbolFont{ssfletters}{OT1}{cmss}{m}{n}
\DeclareMathSymbol{\bsfGamma}{0}{bsfletters}{'000}
\DeclareMathSymbol{\ssfGamma}{0}{ssfletters}{'000}
\DeclareMathSymbol{\bsfDelta}{0}{bsfletters}{'001}
\DeclareMathSymbol{\ssfDelta}{0}{ssfletters}{'001}
\DeclareMathSymbol{\bsfTheta}{0}{bsfletters}{'002}
\DeclareMathSymbol{\ssfTheta}{0}{ssfletters}{'002}
\DeclareMathSymbol{\bsfLambda}{0}{bsfletters}{'003}
\DeclareMathSymbol{\ssfLambda}{0}{ssfletters}{'003}
\DeclareMathSymbol{\bsfXi}{0}{bsfletters}{'004}
\DeclareMathSymbol{\ssfXi}{0}{ssfletters}{'004}
\DeclareMathSymbol{\bsfPi}{0}{bsfletters}{'005}
\DeclareMathSymbol{\ssfPi}{0}{ssfletters}{'005}
\DeclareMathSymbol{\bsfSigma}{0}{bsfletters}{'006}
\DeclareMathSymbol{\ssfSigma}{0}{ssfletters}{'006}
\DeclareMathSymbol{\bsfUpsilon}{0}{bsfletters}{'007}
\DeclareMathSymbol{\ssfUpsilon}{0}{ssfletters}{'007}
\DeclareMathSymbol{\bsfPhi}{0}{bsfletters}{'010}
\DeclareMathSymbol{\ssfPhi}{0}{ssfletters}{'010}
\DeclareMathSymbol{\bsfPsi}{0}{bsfletters}{'011}
\DeclareMathSymbol{\ssfPsi}{0}{ssfletters}{'011}
\DeclareMathSymbol{\bsfOmega}{0}{bsfletters}{'012}
\DeclareMathSymbol{\ssfOmega}{0}{ssfletters}{'012}

\newcommand{\calE}{{\mathcal{E}}}

\newcommand{\calK}{{\mathcal{K}}}

\newcommand{\calN}{{\mathcal{N}}}

\newcommand{\calT}{{\mathcal{T}}}
\newcommand{\calS}{{\mathcal{S}}}
\newcommand{\calU}{{\mathcal{U}}}

\newcommand{\calX}{{\mathcal{X}}}
\newcommand{\calY}{{\mathcal{Y}}}

\newcommand{\calZ}{{\mathcal{Z}}}

\newcommand{\E}[2][]{{\mathbb{E}_{#1}}{\left(#2\right)}}       
\renewcommand{\P}[2][]{{\mathbb{P}_{#1}}{\left(#2\right)}}
\newcommand{\Var}[1]{{\text{\textnormal{Var}}{\left(#1\right)}}}       
\newcommand{\D}[2]{{{\mathbb{D}}\!\left({#1\Vert#2}\right)}}
\newcommand{\avgD}[2]{{{\mathbb{D}}\!\left({#1\Vert#2}\right)}}
\newcommand{\V}[1]{{{\mathbb{V}}\!\left(#1\right)}}

\newcommand{\avgI}[1]{{{\mathbb{I}}\!\left(#1\right)}}
\newcommand{\avgH}[1]{{\mathbb{H}}\!\left(#1\right)}

\newcommand{\Hb}[1]{{\mathbb{H}_b}\left(#1\right)}

\newcommand{\card}[1]{\ensuremath{\left|{#1}\right|}}           
\newcommand{\abs}[1]{\ensuremath{\left|#1\right|}}              
\newcommand{\eqdef}{\ensuremath{\triangleq}}                    
\newcommand{\intseq}[2]{\ensuremath{\llbracket{#1},{#2}\rrbracket}}  
\newcommand{\indic}[1]{\ensuremath{\mathds{1}\!\left\{#1\right\}}}

\renewcommand{\leq}{\leqslant}
\renewcommand{\geq}{\geqslant}

\newcommand{\proddist}{%
  \mathchoice{\raisebox{1pt}{$\displaystyle\otimes$}}
             {\raisebox{1pt}{$\otimes$}}
             {\raisebox{0.5pt}{\scalebox{0.7}{$\scriptstyle\otimes$}}}
             {\raisebox{0.4pt}{\scalebox{0.6}{$\scriptscriptstyle\otimes$}}}}
\newcommand{\pn}{{\proddist n}}

\usepackage[textwidth=0.68in,textsize=footnotesize,disable]{todonotes}
\usepackage{xr,prettyref}
\presetkeys{todonotes}{color=redArcom, linecolor=redArcom,bordercolor=redArcom}{}

\newtheorem{theorem}{Theorem}

\newtheorem{definition}{Definition}
\newtheorem{lemma}{Lemma}
\newtheorem{corollary}{Corollary}
\newtheorem{proposition}{Proposition}
 
\newcommand{\bxn}{\textbf{X}}
\newcommand{\xn}{\textbf{x}}
\newcommand{\byn}{\textbf{Y}}
\newcommand{\yn}{\textbf{y}}
\newcommand{\bzn}{\textbf{Z}}
\newcommand{\zn}{\textbf{z}}

\newcommand{\pz}{P_\emptyset}

\newcommand{\pal}{P_{\alpha_n}}
\newcommand{\paln}{P_{\alpha_n}^{\pn}}

\newcommand{\qz}{Q_\emptyset}
\newcommand{\qzn}{Q^{\pn}_\emptyset}
\newcommand{\qhat}{\widehat{Q}}
\newcommand{\qhatn}{\widehat{Q}^n}
\newcommand{\qal}{Q_{\alpha_n}}
\newcommand{\qaln}{Q_{\alpha_n}^{\pn}}

\newcommand{\mun}[1]{\mu_{#1}^{(n)}}
\newcommand{\Psin}[1]{\Psi_{#1}^{(n)}}
\newcommand{\xijn}{\xi_{j}^{(n)}}
\newcommand{\xin}{\xi^{(n)}}
\newcommand{\agamman}[1]{\mathcal{A}_{\gamma{#1}}^n}
\newcommand{\bseteta}[1]{\mathcal{B}_{\eta{#1}}^n}
\newcommand{\rhovector}{\boldsymbol{\rho}}
\newcommand{\betavector}{\boldsymbol{\beta}}

\newcommand{\brackcurl}[1]{\left\{#1\right\}}
\newcommand{\bracknorm}[1]{\left(#1\right)}
\newcommand{\bracksq}[1]{\left[#1\right]}
\newcommand{\bigO}[1]{\mathcal{O} \left( #1 \right)}
\newcommand{\nex}{\nonumber \\}
\newcommand{\limn}{\lim_{n \to \infty}}
\newcommand{\dupspace}{\phantom{==}} 

\newcommand{\RNum}[1]{\uppercase\expandafter{\romannumeral #1\relax}}

\begin{document}
	
\title{Covert Communication over a $K$-User Multiple-Access Channel}
\author{Keerthi Suria Kumar Arumugam, \textit{Student Member, IEEE}, and Matthieu R. Bloch, \textit{Senior Member, IEEE}\thanks{Parts of this manuscript were presented at the 2016 IEEE International Symposium on Information Theory~\cite{arumugam2016keyless}. This work was supported by the National Science Foundation under Award 1527387.}}
\maketitle

\begin{abstract}
	We consider a scenario in which $K$ transmitters attempt to communicate covert messages reliably to a legitimate receiver over a discrete memoryless \ac{MAC} while simultaneously escaping detection from an adversary who observes their communication through another discrete memoryless \ac{MAC}. We assume that each transmitter may use a secret key that is shared only between itself and the legitimate receiver. We show that each of the $K$ transmitters can transmit on the order of $\sqrt{n}$ reliable and covert bits per $n$ channel uses, exceeding which, the warden will be able to detect the communication. We identify the optimal pre-constants of the scaling, which leads to a complete characterization of the covert capacity region of the $K$-user binary-input \ac{MAC}. We show that, asymptotically, all sum-rate constraints are inactive unlike the traditional \ac{MAC} capacity region. We also characterize the channel conditions that have to be satisfied for  the transmitters to operate without a secret key.
\end{abstract}  
\section{Introduction} \label{sec:Introduction}
Recently, there has been a renewed interest in the characterization of the information-theoretic limits of covert communication. 
Following the early work~\cite{bash2013limits} identifying the existence of a \emph{square-root law} similar to that of steganography~\cite{Ker2007},~{\cite{Fridrich2009}}, several follow-up works have refined the characterization of the square-root law for point-to-point classical channels~\cite{Wang2016b,bloch2016covert,Tahmasbi2017,Tahmasbi2017b} and classical-quantum channels~\cite{Bash2015a,Sheikholeslami2016,Wang2016c}; in particular, the \emph{covert capacity} has been defined and precisely computed for \acp{DMC} and \ac{AWGN} channels.
Other extensions have included developing an understanding of when the square-root law does not apply, such as in the presence of channel uncertainty~\cite{che2014reliable,lee2015achieving,Sobers2017,Goeckel2016},~{\cite{lee2018covert}} or timing uncertainty~\cite{Bash2016,arumugam2016asynch}. There have also been several investigations of multi-user models~\cite{ArumugamBloch2018,Tan2017},~{\cite{arumugam2018covertrelay, HuYanZhouEtAl2017a}}, timing channels~\cite{Soltani2015,Mukherjee2016a}, artificial-noise assisted transmissions~\cite{Soltani2017}, and covert key generation~\cite{Tahmasbi2017c}, as well as code designs~\cite{Freche2017, Zhang2016},~{\cite{KadampotTahmasbiBloch2018}}.
 Finally, we note that related works have explored the concept of \emph{stealth}~\cite{Hou2014}, which is tightly tied to the notion of approximation of output statistics~\cite{Han1993,Cuff2013} and may be viewed as \emph{low probability of interception}, whereas covertness focuses on \emph{low probability of detection}.

The main result developed in this paper is the characterization of the {covert capacity region} of the $K$-user binary-input \ac{MAC}. The tools used are natural extensions of the techniques developed for point-to-point covert and stealth channels in~\cite{Wang2016b,bloch2016covert,Hou2014} and for \ac{MAC} resolvability~\cite{Steinberg1998,Yassae2010,Frey2017}, but the converse proof requires special care beyond the approach used in~\cite{bloch2016covert}. 
We extend our previous work~\cite{arumugam2016keyless} by analyzing $K$-user \acp{MAC} for any $K\geq 2$ and characterizing  the optimal key throughput required for covert communication.
We show that, asymptotically, there exist no sum-rate constraints unlike the traditional \ac{MAC} rate region;\footnote{We correct an embarrassing oversight in the converse argument outlined in~\cite{arumugam2016keyless}, in which we assumed that the channels perceived by the users were independent without formal proof. The corrected argument is now given in Section~\ref{sec:converse-proof}.} intuitively, this happens because covertness is such a stringent constraint that {the covert users} never transmit enough bits to saturate the capacity of the channel. The system behaves as if a covert communication \emph{budget} were merely allocated to the different users.\footnote{This intuitive interpretation is attributed to Sidharth Jaggi, during discussions at ISIT 2016.} 
{A similar behavior was observed~\cite[Theorem 6]{verdu1990channel} in the calculation of the channel capacity per unit cost of a two-user \ac{MAC} when both users consist of a \emph{free} input symbol.}

The remainder of the paper is organized as follows. 
In Section~\ref{sec:notation}, we set the notation used in the paper, and in Section~\ref{sec:channelmodel}, we formally introduce our channel model and define the covert capacity region. 
In Section~\ref{sec:covertcommprocess}, we develop a preliminary result that captures the essence of our approach to covertness and extends~\cite[Lemma 1]{bloch2016covert}. 
We establish the covert capacity region of the $K$-user binary-input \ac{MAC} in Section~\ref{sec:mainresult} { and conclude our work with a brief discussion of extensions and open problems in Section~\ref{sec:conclusion}. 
The proofs of all lemmas are relegated to the appendix. }

\section{Notation} \label{sec:notation}
We denote random variables and their realizations in upper and lower case, respectively. All sequences in boldface are $n$-length sequences, where $n \in \mathbb{N}^*$, unless specified otherwise. We define the weight of a sequence as the number of non-zero symbols in that sequence. Throughout the paper, $\log$ and $\exp$ are understood to the base $e$; the results can be interpreted in bits by converting $\log$ to the base $2$. Adhering to standard information-theoretic notation, $\avgH{X}$ and $\avgI{X; Y}$ represent the {entropy} of $X$ and the {mutual information} between $X$ and $Y$, respectively. For $x \in \bracksq{0,1}$, let $\Hb{x}$ denote the {binary entropy} of $x$. For two distributions, $P$ and $Q$, defined on the same finite alphabet $\calX$, the Kullback-Leibler (KL) divergence is $\avgD{P}{Q} \eqdef \sum_x P(x) \log \frac{P(x)}{Q(x)}$ and the variational distance is $\V{P,Q} \eqdef \frac{1}{2} \sum_x |P(x) - Q(x)|$. They are related by Pinsker's inequality~{\cite{CoverThomas2006}} as $\V{P,Q}^2 \leq \frac{1}{2}\avgD{P}{Q}$. If $P$ is absolutely continuous \ac{wrt} $Q$, we write $P \ll Q$. 
For $x \in \mathbb{R}$, we define $[x]^+ \eqdef \max\bracknorm{{x,0}}$. 
We denote the cardinality of a set $\calT$ by $\card{\calT}$, and we represent the vector $\brackcurl{X_k : k \in \calT}$ by $X{{\bracksq{\calT}}}$. 
We denote the cartesian product $\times_{k \in \calT} \calX_k$ by $\calX \bracksq{\calT}$.
Furthermore, $\emptyset$ denotes an empty set, and $\calS \setminus \calT$ denotes the set difference of two sets $\calS$ and $\calT$. 

\section{Channel model} \label{sec:channelmodel}
We define the set $\calK \eqdef \intseq{1}{K}$, where $K \in \mathbb{N}^*$ and $K \geq 2$. We analyze the channel model illustrated in Figure~\ref{fig:channelmodel}, in which $K$ transmitters simultaneously communicate with a legitimate receiver over a discrete memoryless \ac{MAC} $\left(\calX[\calK],W_{Y|X[\calK]},\calY\right)$ in the presence of a warden monitoring the communication over another discrete memoryless \ac{MAC} $\left(\calX[\calK],W_{Z|X[\calK]},\calZ\right)$. As both channels are memoryless, we denote the transition probabilities corresponding to $n$ uses of the channel by $W_{Y|X{\bracksq{\calK}}}^{\pn} \eqdef \prod_{i=1}^n W_{Y|X{\bracksq{\calK}}}$ and $W_{Z|X{\bracksq{\calK}}}^{\pn} \eqdef \prod_{i=1}^n W_{Z|X{\bracksq{\calK}}}$. In addition, we assume for simplicity of exposition that each user $k\in\calK$ uses the same binary input alphabet $\calX_k\eqdef \calX\eqdef\{0,1\}$ and that the output alphabets $\calY$ and $\calZ$ are finite. We let $0 \in \calX$ be the innocent symbol corresponding to the channel input when no communication takes place. We assume that all terminals are synchronized and possess complete knowledge of the coding scheme used. 

The user indexed by $k \in \calK$, encodes a uniformly-distributed message $W_k \in\intseq{1}{M_k}$ and a uniformly-distributed secret key $S_k\in\intseq{1}{L_k}$, which is shared only with the receiver, into a codeword $\bxn_k(W_k,S_k)\in\calX^n$ of length $n$. 
We denote the collection of the $K$ codewords $\brackcurl{\bxn_k\bracknorm{W_k,S_k}}_{k \in \calK}$ by $\bxn_{\calK}\bracknorm{W{\bracksq{\calK}}, S{\bracksq{\calK}}}$. 
When the context is clear, we drop the message and key indices, $W_k$ and $S_k$, and denote $\bxn_k(W_k,S_k)$ by $\bxn_k$ instead for conciseness.  
It is convenient to think about the $K$ inputs to the channel over $n$ uses as a matrix $\bxn\bracksq{\calK}$ of size $K \times n$ obtained by vertically stacking the $K$ codewords, each of which is a row vector. 
The inputs corresponding to all users indexed by the elements of a non-empty set $\calU\subset\calK$ is a sub-matrix of $\bxn\bracksq{\calK}$ obtained by selecting the rows whose indices belong to $\calU$ and is denoted by $\bxn\bracksq{\calU}$.
The $K$ users then transmit codewords $\bxn\bracksq{\calK}$ over the channel in $n$ channel uses. 
At the end of transmission, the receiver observes $\byn$ while the warden observes $\bzn$, both of which are of length $n$.
\begin{figure}[b] 
  \centering
  \includegraphics[width=0.6\linewidth]{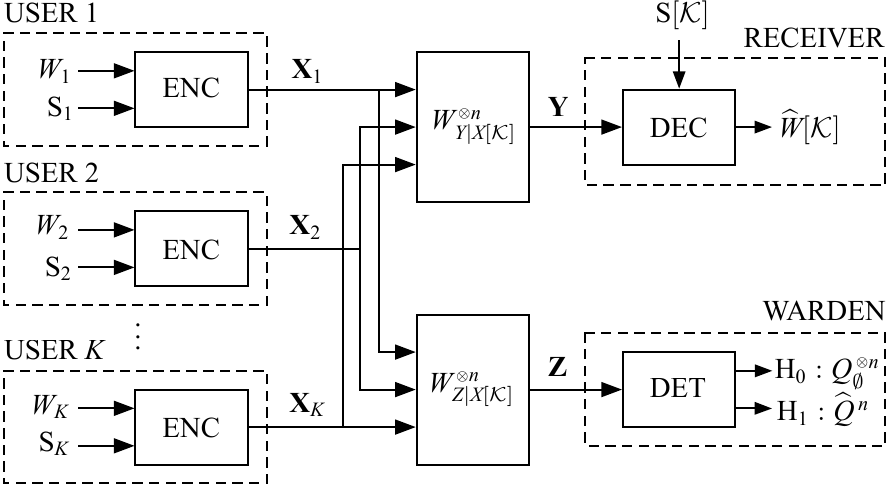}
  \caption{Model of covert communication over a \ac{MAC} with $K$ transmitters.}
  \label{fig:channelmodel}
\end{figure} 

We introduce a $K$-length row vector $X_\calU=(X_1,X_2,\dots,X_K)$, $\calU \subseteq \calK$, with entry $X_k=1$ if $k\in\calU$ and $X_k=0$ otherwise. 
With our assumption that all channel inputs are binary, we represent every column of the matrix $\bxn\bracksq{\calK}$ by a vector $\bracknorm{ X_{\calU}}^T$, where the set $\calU$  consists of the indices of all users transmitting symbol $1$ in this column.
We denote the $k^{\text{th}}$ component of $X_{\calU}$ by $X_{\calU, k}$. 
In accordance with the notation introduced in the previous paragraph, $X_\calU[\calT]$ represents a row vector of length $\card{\calT}$ that contains the entries $\brackcurl{X_{\calU,k}}_{k \in \calT}$. 
Note the difference between $X\bracksq{\calU}$ and $X_{\calU}$; the former is a $\abs{\calU}$-length vector $\brackcurl{X_k}_{k \in \calU}$ whereas the latter is a $K$-length vector with $1$'s in indices that belong to the set $\calU$. 
For conciseness, we define
\begin{align}
	P_\calU(y) \eqdef W_{Y|X{{\bracksq{\calK}}}}(y|x_\calU), \quad	Q_\calU(z) \eqdef W_{Z|X{{\bracksq{\calK}}}}(z|x_\calU), \label{eq:1}
\end{align}
which represent the one-shot output distributions at the legitimate receiver and the warden, respectively, when only the transmitters in $\calU \subseteq \calK$ transmit symbol $1$, while the transmitters in $\calU^c$ transmit a $0$. 
When $\calU$ is a singleton set $\brackcurl{k}$, which corresponds to user $k$ transmitting $1$ and all other users transmitting $0$, we write $P_k$ and $Q_k$ instead of $P_{\brackcurl{k}}$ and $Q_{\brackcurl{k}}$, respectively. 
If $\calU = \emptyset$, which occurs when all users transmit the innocent symbol $0$, we write $\pz$ and $\qz$.
{
We assume that $Q_\calU \ll \qz$ for all non-empty sets $\calU \subseteq \calK$ and that $\qz$ cannot be written as a convex combination of the form $ \qz(z) = \sum_{\calT\subseteq \calK} \left(\prod_{k\in\calT}\mu_{k}\right)\left(\prod_{k\in\calT^c}(1-\mu_{k})\right)Q_{\calT}(z)$ for some $\{\mu_k\}_{k\in\calK}\in[0,1]^K \setminus \brackcurl{0}_{k \in \calK}$. 
In the former case, covert communication involving all $K$ users is impossible; in the latter case, covert communication would directly follow from known channel resolvability results~\cite{Steinberg1998,Yassae2010,Frey2017} and would be possible at a non zero-rate. 
We also assume that there does not exist $\{\rho_k\}_{k\in\calK}\in[0,1]^K $ with $\sum_{k\in\calK}\rho_k=1$ such that $\sum_{k\in\calK}\rho_k    Q_k(z)=Q_\emptyset(z)$ for all $z \in \calZ$.    
As we shall see later in Section~\ref{sec:covertcommprocess}, the square root law of covert communication can be circumvented if such a $\{\rho_k\}_{k\in\calK}$ exists.}

Upon observing $\byn$,  the legitimate receiver estimates the message vector $\widehat{W}{{\bracksq{\calK}}}$. 
We measure reliability {at the receiver} with the average probability of error $\smash{P_e^n \eqdef \P{\widehat{W}{{\bracksq{\calK}}} \neq W{{\bracksq{\calK}}}}}$.
Upon observing $\bzn$, the warden attempts to detect whether all $K$ users transmitted covert messages (Hypothesis $H_1$) or not (Hypothesis $H_0$) by performing a hypothesis test on $\bzn$. 
We denote the Type I (rejecting $H_0$ when true) and Type II (accepting $H_0$ when false) {error probabilities} by $\alpha$ and $\beta$, respectively. 
The warden can achieve any pair $\bracknorm{\alpha, \beta}$ such that $\alpha + \beta = 1$ by ignoring his observation $\bzn$ and basing his decision on the result of a coin toss. 
We define the distribution induced at the warden when communication takes place by
\begin{align}
	\qhatn(\zn) \eqdef \frac{1}{\prod_{k \in \calK} M_{k} L_{k}}  \sum_{m{{\bracksq{\calK}}}} \sum_{\ell{{\bracksq{\calK}}}}  W_{Z|X{{\bracksq{\calK}}}}^{\pn} \bracknorm{\zn | \xn_{\calK}\bracknorm{m{\bracksq{\calK}}, \ell{\bracksq{\calK}}}  }. \label{eq:3}
\end{align}
We measure covertness in terms of the KL divergence $\avgD{\smash{\qhatn}}{\qzn}$, where $\qzn$ is the distribution observed by the warden when none of the $K$ users transmits any covert information. 
{We know from~\cite{lehmann2006testing} that any test conducted by the warden on $\bzn$ satisfies $\alpha + \beta \geq 1 - \V{\smash{\qhatn},\qzn}$. 
Using Pinsker's inequality~\cite{CoverThomas2006}, we write $\alpha + \beta \geq 1 - \sqrt{\avgD{\smash{\qhatn}}{\qzn}}$. }
The primary objective of our covert communication scheme is to guarantee that $\avgD{\smash{\qhatn}}{\qzn}$ is negligible so that any statistical test used by the warden on $\bzn$ is futile.
{Note that we only consider communication schemes for which $\log M_k$ grows to infinity,  for $k \in \calK$, as $n$ grows to infinity.}
\begin{definition}
  The tuple $r{\bracksq{\calK}} \in \mathbb{R}_{{+}}^K$ is an achievable reliable and covert throughput tuple if there exists a sequence of codes as defined above with increasing blocklength $n$ such that for every $k \in \calK$,
  {
  \begin{align}
    & \liminf_{n \to \infty} \frac{\log M_k}{\sqrt{n \avgD{\smash{\qhatn}}{\qzn}}} \geq r_k , \label{eq:4} 
   \end{align}
   }
   and
   \begin{align} 
    \limn P_e^n = 0, \quad \limn \avgD{\qhatn}{\qzn} = 0. \label{eq:6}
  \end{align}
  The covert capacity region of the $K$-user \ac{MAC} consists of the closure of the set of all achievable throughput tuples $r{\bracksq{\calK}}$.
  {
  Also, we define the tuple $s[\mathcal{K}]\in\mathbb{R}_+^K$ as an achievable key throughput tuple associated with the achievable reliable and covert throughput tuple $r[\mathcal{K}]$, if there exist a sequence of codes satisfying~\eqref{eq:4} and~\eqref{eq:6} and if for all $k\in\mathcal{K}$, 
\begin{align}
s_k \geq \limsup_{n \to \infty} \frac{\log L_k}{\sqrt{n \avgD{\smash{\qhatn}}{\qzn}}}.  \label{eq:7}
\end{align}
}
\end{definition}
{Note that in~\eqref{eq:4}, we normalize the number of bits $\log M_k$ by $\sqrt{n \avgD{\smash{\qhatn}}{\qzn}}$ instead of $n$ as traditionally done in information-theoretic problems. 
The normalization by $\sqrt{n}$ is essential to reflect the fact that covert communication corresponds to a zero-rate regime, in which the number of bits scales sub-linearly with the number of channel uses. 
The normalization by $\sqrt{\avgD{\smash{\qhatn}}{\qzn}}$ is also crucial to reflect the fact that  $\avgD{\smash{\qhatn}}{\qzn}$ influences $\brackcurl{ \log M_k}_{ k \in \calK}$. 
While the normalization might seem somewhat ad-hoc, it is justified \emph{a posteriori} in Section~\ref{sec:mainresult} when we prove that $\log M_k/\sqrt{n \avgD{\smash{\qhatn}}{\qzn}}$ is independent of $n$ in the limit of large blocklength. 
Said differently, $\log M_k/\sqrt{n \avgD{\smash{\qhatn}}{\qzn}}$ plays the role of the usual ``rate'' in that it asymptotically does not depend on the blocklength $n$ and already integrates the scaling. 
To avoid confusion, we refer to $r_k$ as \emph{throughput} instead of rate.}

\section{Preliminaries}
\label{sec:covertcommprocess}
Following the approach proposed in~\cite{bloch2016covert}, we introduce a \emph{covert communication process}, which is an \ac{iid} process indistinguishable from the innocent distribution $\qzn$ in the limit. The rationale for introducing this process is to precisely quantify the fraction of channel uses in which the users can transmit symbol $1$ while simultaneously avoiding detection by the warden, without introducing the coding aspect of the problem yet.

For $n \in \mathbb{N}^*$, let $\alpha_n \in \bracknorm{0,1}$. Let $\rhovector \eqdef \{\rho_k\}_{k\in\calK}\in[0,1]^K$ such that\footnote{{The assumption $\sum_k \rho_k = 1$ is only made for convenience and, as we shall see from the converse part of Theorem~\ref{thm:covertcapacity}, without loss of generality.}} $\sum_{k \in \calK} \rho_k = 1$. We define the input distributions $\brackcurl{\Pi_{X_k}}_{k \in \calK}$ on $\calX$ as 
\begin{align}
  \Pi_{X_k}(1) = 1 - \Pi_{X_k}(0) = \rho_k \alpha_n. \label{eq:8}
\end{align}
The output distributions at the legitimate receiver and the warden when the input distribution of each user $k$ is $\Pi_{X_k}$ are defined, respectively, as
\begin{align}
	\pal\!\bracknorm{y} & \eqdef \sum_{x{{\bracksq{\calK}}}} W_{Y|X{{\bracksq{\calK}}}}(y|x{{\bracksq{\calK}}})\bracknorm{ \prod_{k \in \calK} \Pi_{X_k}(x_k)}, \label{eq:9} \\
	\qal\!\bracknorm{z} & \eqdef \sum_{x{{\bracksq{\calK}}}} W_{Z|X{{\bracksq{\calK}}}}(z|x{{\bracksq{\calK}}})\bracknorm{ \prod_{k \in \calK} \Pi_{X_k}(x_k)}. \label{eq:10}
\end{align} 
The $n$-fold product distributions corresponding to~\eqref{eq:8},~\eqref{eq:9}, and~\eqref{eq:10} are
\begin{align}
  \Pi_{X_k}^{\pn} & = \prod_{j=1}^n \Pi_{X_k}, \quad \paln = \prod_{j=1}^n \pal, \quad \qaln  = \prod_{j=1}^n \qal. \label{eq:11}
\end{align}
For a set $\calT \subseteq \calK$, we define 
 \begin{align}
 	G_{\calT}(z) \eqdef \sum_{\calU \subseteq \calT} \bracknorm{-1}^{|\calT| - |\calU|} Q_{\calU}(z). \label{eq:12}
 \end{align}
Then, using {Lemma~\ref{lem:altrepresentation} in Appendix~\ref{sec:altrep}}, we write
\begin{align}
	\qal(z) = \qz(z) + \sum_{\substack{\calT \subseteq \calK: \calT \neq  \emptyset}} \bracknorm{\prod_{k \in \calT} \rho_k \alpha_n } G_{\calT}(z). \label{eq:13}
\end{align}
Note that since $Q_{\calT} \ll \qz$ for all non-empty sets $\calT \subseteq \calK$, it is also true that $\qal \ll \qz$. Furthermore, we define 
  \begin{align}
    & \zeta_n(z) \eqdef \frac{Q_{\alpha_n}(z) - \qz(z)}{\alpha_n},\quad \chi_n{\bracknorm{\rhovector}} \eqdef \sum_z \frac{\zeta_n^2(z)}{\qz(z)}, \label{eq:14} \\
    & \zeta(z)  \eqdef \sum_{k \in \calK} \rho_k (Q_k(z) - \qz(z)),\quad	\chi{\bracknorm{\rhovector}} \eqdef \sum_z \frac{\zeta^2(z)}{\qz(z)}.  \label{eq:15}
  \end{align}
In the following lemma, we bound the KL divergence between $\qal$ and $\qz$. Later, we use the results of this lemma to show that for specific choices of $\alpha_n$, the stochastic process $\qaln$ is indistinguishable from the innocent distribution $\qzn$ in the limit.
\begin{lemma}
  \label{lem:covertprocess}
  Let the sequence $\brackcurl{\alpha_n}_{n \geq 1}$ be such that $\limn \alpha_n = 0$. 
  Then, for $n \in \mathbb{N}^*$ large enough,
  \begin{align}
    \frac{\alpha_n^2}{2} \bracknorm{1 + \sqrt{\alpha_n}} \chi_n{\bracknorm{\rhovector}} \geq \avgD{Q_{\alpha_n}}{\qz} \geq \frac{\alpha_n^2}{2} \bracknorm{1 - \sqrt{\alpha_n}} \chi_n{\bracknorm{\rhovector}}. \label{eq:lem1}
  \end{align}
  In addition, for all $z \in \calZ$, $\limn \zeta_n(z) = \zeta(z)$ and $\limn \chi_n{\bracknorm{\rhovector}} = \chi{\bracknorm{\rhovector}}$. Finally, for random variables $\bracknorm{X{{\bracksq{\calT}}},Z} \in \calX^{\abs{\calT}} \times \calZ$ for some non-empty set $\calT \subseteq \calK$ with joint distribution $W_{Z|X{{\bracksq{\calT}}}}\bracknorm{\prod_{k \in \calT} \Pi_{X_k}}$, we have
  \begin{align}
  	\avgI{X{{\bracksq{\calT}}};Z} = \sum_{k \in \calT} \rho_k \alpha_n \avgD{Q_k}{\qz} + \bigO{\alpha_n^2}. \label{eq:lem2}
  \end{align}
\end{lemma}
The proof of Lemma~\ref{lem:covertprocess} is provided in {Appendix~\ref{sec:covertprocess}}. 
Assume that each transmitter $k \in \calK$ generates a sequence of length $n$ using the process $\Pi_{X_k}^{\pn}$. The weight of these sequences is $\rho_k n \alpha_n$ on average. To be indistinguishable from the innocent distribution in the limit, the covert process $\qaln$ has to satisfy 
\begin{align}
	\limn \avgD{\qaln}{\qzn} = \limn n\avgD{\qal}{\qz} = 0. \label{eq:25}
\end{align}
{Our assumptions in Section~\ref{sec:channelmodel} ensure that $\chi\bracknorm{\rhovector}$ is non-zero. }
Consequently, from the results of Lemma~\ref{lem:covertprocess} and~\eqref{eq:25}, we conclude that if we choose the sequence $\brackcurl{\alpha_n}_{n \in \mathbb{N}^*}$ such that $\limn n \alpha_n^2 = 0$, our covert process {$\qaln$} is indistinguishable from $\qzn$ in the limit. Consequently, we will construct a coding scheme that emulates the covert process $\qaln$ instead of $\qzn$. The prime benefit of using $\qaln$ instead of $\qzn$ is that $\qaln$ allows us to convey covert information through the use of $1$ symbols. In particular, it is possible to choose $\brackcurl{\alpha_n}_{n \in \mathbb{N}^*}$ such that $\limn n \alpha_n = \infty $ so that the number of information bits grows with $n$. 

{The ``square root law'' of covert communication follows from the constraint $\limn n \alpha_n^2 = 0$, which forces the scaling of $n\alpha_n$ to be arbitrarily close to but not exceed $\sqrt{n}$. If $\chi(\rhovector)=0$ for some $\rhovector$,  one would need to push the approximation of $\avgD{\qaln}{\qzn}$ at least to the order $\alpha_n^3$ in Lemma~\ref{lem:covertprocess}. In turn, we would only need to choose a sequence such that $\limn n \alpha_n^3 = 0$, effectively allowing the increase of the scaling of $n \alpha_n$ to be arbitrarily close to but not exceed ${n}^{2/3}$ and beating the square root law. The assumption that $\chi(\rhovector)>0$ made in Section~\ref{sec:channelmodel} therefore excludes the (rare) situations in which the square root law can be beaten}.

\section{Main result} \label{sec:mainresult}
We characterize the covert capacity region of a $K$-user binary-input \ac{MAC} in Theorem~\ref{thm:covertcapacity}, {with the achievability proof in Section~\ref{sec:achievability-proof} and the converse proof in Section~\ref{sec:converse-proof}}. The proofs adapt channel resolvability and converse techniques used in~\cite{bloch2016covert} for point-to-point channels to the \acp{MAC}. The achievability proof is an extension of~\cite{bloch2016covert}, and we provide details in { the appendix}; the converse proof presents more challenges and is fully detailed.
\subsection{Covert capacity region of the $K$-user binary-input \ac{MAC}}
\label{sec:covertcapacity}
\begin{theorem}
  \label{thm:covertcapacity}
  For $\rhovector \eqdef \brackcurl{ \rho_k}_{k \in \calK } \in \bracksq{0,1}^{K}$ such that $\sum_{k \in \calK} \rho_k = 1$, define 
  \begin{align}
  	\chi(\rhovector) \eqdef \sum_z \frac{\bracknorm{\sum_{k \in \calK} \rho_k \bracknorm{ Q_k(z) - \qz(z)}}^2 }{\qz(z)}. \label{eq:cap_0}
  \end{align}
  For the $K$-user binary-input \ac{MAC} described in Section~\ref{sec:channelmodel}, the covert capacity region is
  \begin{align}
    \bigcup_{\{\rho_k\}_{k\in\calK}\in[0,1]^K:\sum_{k \in \calK} \rho_k = 1}\left\{\{r_k\}_{k\in\calK}:\forall k\in\calK,\quad r_k\leq \sqrt{\frac{2}{\chi\bracknorm{\rhovector}}} \rho_k \avgD{P_k}{\pz}\right\}. \label{eq:cap_1}
  \end{align}
  {
  In addition, for any achievable reliable and covert throughput tuple $r[\calK]$ on the boundary of the covert capacity region characterized by $\rhovector$, the set of achievable key throughput tuples is
  \begin{align}
\left\{\{s_k\}_{k\in\calK}:\forall k\in\calK,\quad  s_k \geq \sqrt{ \frac{2}{\chi(\rhovector) }} \rho_k \left[\avgD{Q_k}{\qz} - \avgD{P_k}{\pz}\right]^+\right\}.  \label{eq:cap_2}
  \end{align}
  }
\end{theorem}
{Note that $\chi\bracknorm{\rhovector}$ in~\eqref{eq:cap_0} is positive under the assumption made in Section~\ref{sec:channelmodel}, so that the bounds in~\eqref{eq:cap_1} and~\eqref{eq:cap_2} are well defined and finite. 
A few remarks are now in order.}
\begin{itemize}
    \item Our characterization of the covert capacity region only involves constraints on individual user's throughputs; there are no active constraints on the sum throughput. However, the individual throughputs are not identical to those of the single-user case~\cite{bloch2016covert}, as there exists a non-trivial interplay among the $\rho_k$'s, for $k \in \calK$, through $\chi(\rhovector) $ in~\eqref{eq:cap_1}.  
    \item User $k \in \calK$ can {achieve its maximum covert and reliable throughput} without a key {only} if 
    \begin{align}
      \avgD{P_k}{\pz} \geq \avgD{Q_k}{\qz}, \label{eq:cap_4}
    \end{align} 
    is satisfied; that is, no secret key is required for {user $k$} if the channel from user $k$ to the receiver is \emph{better} than the channel to the warden when all other users are silent. 
    \item { If the \ac{MAC} is symmetric, in the sense that $\forall z \in \calZ$ and $\forall k \in \calK$, $Q_k(z) = Q(z)$, then $\sum_{k \in \calK} \rho_k \allowbreak \bracknorm{Q_k(z) - \qz(z)} = Q(z) - \qz(z)$, so that $\chi{\bracknorm{\rhovector}}$ is independent of  $\rhovector$ and time sharing is optimal.} 
\end{itemize}
\begin{figure}[b!] 
  \centering
  \includegraphics[width=0.6\linewidth]{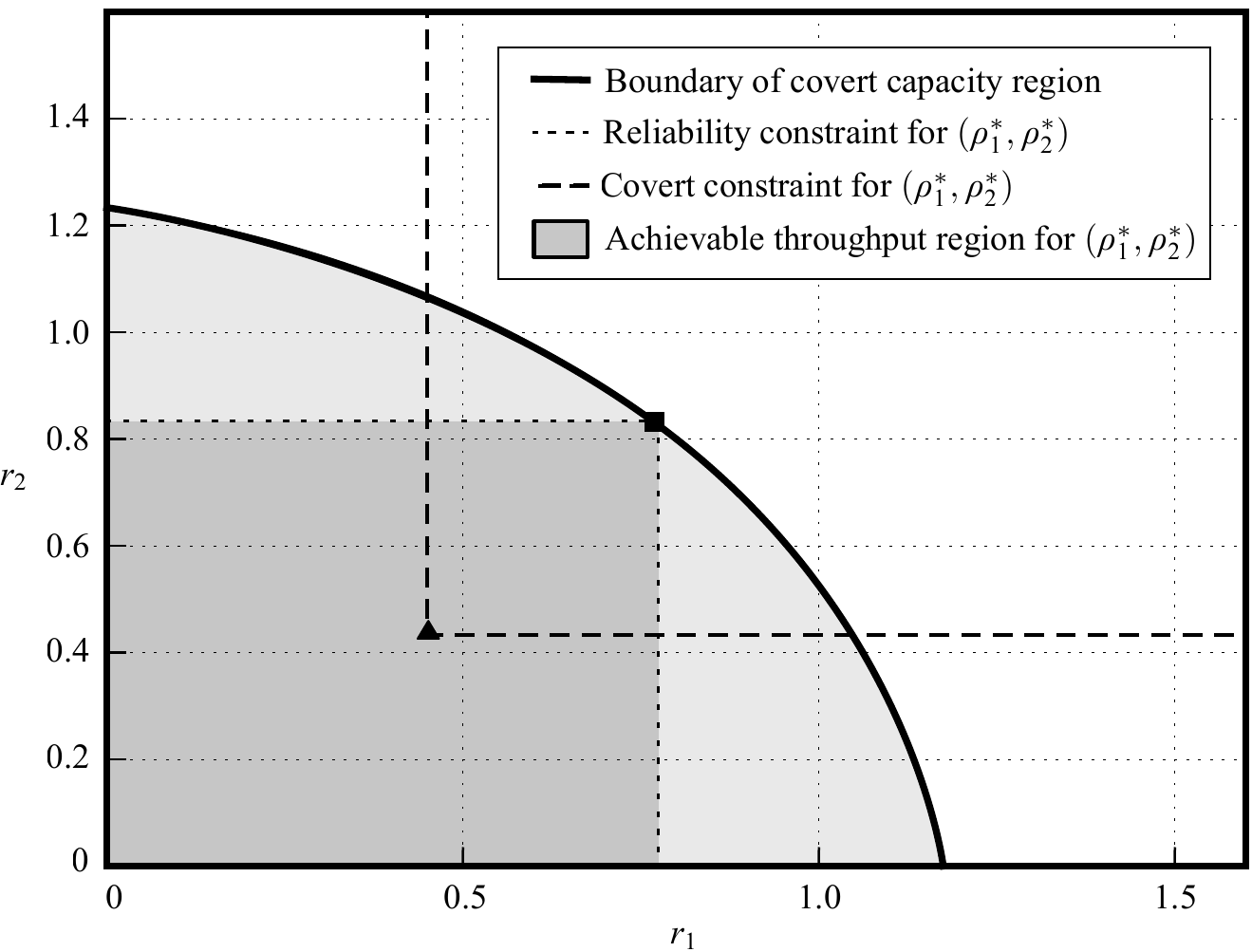}
  \caption{Representative example of the covert capacity region for a 2-user \ac{MAC}. The achievable rate region for a specific choice of $\rhovector = \rhovector^* = \bracknorm{\rho^*_1, \rho^*_2}$ is highlighted.}
  \label{fig:capacityregion}
\end{figure}

Figure~\ref{fig:capacityregion} illustrates the covert capacity region for a $2$-user \ac{MAC} with randomly generated channel matrices, $W_{Y|X_1X_2}$ and $W_{Z|X_1X_2}$, that satisfy~\eqref{eq:cap_4} for $k \in \brackcurl{1,2}$ and the absolute continuity requirements described in Section~\ref{sec:channelmodel} for $\calK = \brackcurl{1,2}$. 
The thick solid curve denotes the boundary of the covert capacity region. 
All points on this boundary can be achieved by varying the values of $\bracknorm{\rho_1, \rho_2}$. 
	For $\rhovector = \rhovector^* \eqdef  \bracknorm{\rho^*_1, \rho^*_2}$, the achievable covert throughput region is highlighted in Figure~\ref{fig:capacityregion}, where the square marker represents the maximum achievable covert throughput pair $\bracknorm{ \sqrt{\frac{2}{\chi{\bracknorm{\rhovector^*}}}} \rho^*_1 \avgD{P_1}{\pz}, \sqrt{\frac{2}{\chi{\bracknorm{\rhovector^*}}}} \rho^*_2 \avgD{P_2}{\pz} }$, while the triangular marker represents the pair $\bracknorm{ \sqrt{\frac{2}{\chi{\bracknorm{\rhovector^*}}}} \rho^*_1 \avgD{Q_1}{\qz}, \sqrt{\frac{2}{\chi{\bracknorm{\rhovector^*}}}} \rho^*_2 \avgD{Q_2}{\qz} }$. A non-empty intersection of the region to the {top-right} of the triangular marker and the region to the {bottom-left} of the square marker implies the existence of keyless covert communication schemes. If the regions do not intersect, a secret key is required to communicate covertly.
	{Note that the achievable region is still the region spanning from (0,0) to the square marker as highlighted in Figure~\ref{fig:capacityregion}. }
{Also note that, for a symmetric 2-user \ac{MAC}, the boundary of the covert capacity region is a straight line, and time sharing is optimal. }

\subsection{Achievability proof}
\label{sec:achievability-proof}

We consider a communication scheme in which every user $k$ employs $L_k$ sub-codebooks, each {consisting} of $M_k$ codewords. The value of the key $S_k \in \intseq{1}{L_k}$ chooses the sub-codebook that user $k$ uses to encode its message $W_k \in \intseq{1}{M_k}$. The decoder, which possesses complete knowledge of the keys $S{{\bracksq{\calK}}}$, attempts to decode the messages sent by the $K$ transmitters. The idea underlying the scheme is to use channel resolvability techniques to ensure that the total number of codewords is sufficiently large to keep the warden confused, while simultaneously ensuring that each sub-codebook is small enough for the receiver to reliably decode the messages.
	\begin{proposition} \label{prop:achievability}
	    Let $\rhovector \eqdef \brackcurl{\rho_k}_{k \in \calK} \in [0,1]^K$ with $\sum_{k \in \calK} \rho_k = 1$. 
	    Let $\brackcurl{\alpha_n}_{n \in \mathbb{N}^*}$ be such that $\alpha_n \in \bracknorm{0,1}$, $\limn n \alpha_n = \infty$, and $\limn n \alpha_n^2 = 0$. 
	    For the channel model described in Section~\ref{sec:channelmodel}, for an arbitrary $\mu \in \bracknorm{0,1}$, there exist covert communication schemes such that $\forall k \in \calK$,
	    \begin{align}
	      & r_k = \bracknorm{1 - \mu}  \sqrt{\frac{2}{\chi{\bracknorm{\rhovector}}}}\rho_k\avgD{P_k}{\pz}, \label{eq:ach_1} \\
	      & s_k = \sqrt{\frac{2}{\chi{\bracknorm{\rhovector}}}} \rho_k \bracksq{\bracknorm{1 + \mu} \avgD{Q_k}{\qz} - \bracknorm{1 - \mu} \avgD{P_k}{\pz}}^+ , \label{eq:ach_2}  \\
	      & \limn P_e^n  = 0, \label{eq:ach_3} \\
	      & \limn \avgD{\qhatn}{\qzn} = 0. \label{eq:ach_4}
	    \end{align}
	\end{proposition}
	\begin{proof}
    To prove Proposition~\ref{prop:achievability}, we rely on random coding arguments for channel reliability and channel resolvability. 
    However, the use of low-weight codewords in our communication scheme requires that we handle concentration inequalities carefully. Since basic concentration inequalities do not apply in the low-weight regime~\cite{bloch2016covert}, we use Bernstein's inequality to establish our random coding arguments. The proof follows otherwise along the lines of~\cite[Theorem 2]{bloch2016covert}. 

    \paragraph{Random codebook generation} At each transmitter $k \in \calK$, generate $M_kL_k$ codewords $\xn_k\bracknorm{m_k,\ell_k} \allowbreak \in \calX^n$, where $\bracknorm{m_k,\ell_k} \in \intseq{1}{M_k} \times \intseq{1}{L_k}$, independently at random according to the distribution $\Pi_{X_k}^{\pn}$.
    For a set $\calT \subset \calK$, define
		\begin{align}
			W^\pn_{Y|X{{\bracksq{\calT}}}}\bracknorm{\yn|\xn{{\bracksq{\calT}}}} \eqdef \sum_{\xn{{\bracksq{\calT^c}}}} W^{\pn}_{Y|X{{\bracksq{\calK}}}}\bracknorm{\yn|\xn{{\bracksq{\calK}}}} \bracknorm{\prod_{k \in \calT^c} \Pi_{X_k}^{\pn}\bracknorm{\xn_k} } \label{eq:26}.
		\end{align}
		Note that $W^\pn_{Y|X{{\bracksq{\calT}}}}$ is a product distribution since each user $k \in \calK$ generates its codeword according to an $n$-fold product distribution $\Pi_{X_k}^{\pn}$. 
		Also, note that if $\calT = \emptyset$, $W_{Y|X{{\bracksq{\calT}}}}^{\pn} = \paln$. 
		Define the set $\agamman{} \eqdef \bigcap_{\substack{\calT \subseteq \calK: \\\calT \neq \emptyset}} \agamman{_\calT}$ with 
		\begin{align}
			\agamman{_{\calT}} & \eqdef \Bigg\{  \bracknorm{\xn{{\bracksq{\calK}}}, \yn} \in \calX^n{{\bracksq{\calK}}} \times \calY^n:  \log \frac{W_{Y|X{{\bracksq{\calK}}}}^{\pn}\bracknorm{\yn|\xn{{\bracksq{\calK}}}}}{W^\pn_{Y|X{{\bracksq{\calT^c}}}}\bracknorm{\yn|\xn{{\bracksq{\calT^c}}}}} \geq  \gamma_{\calT} \Bigg\} , \label{eq:27}
		\end{align}
		where, for every non-empty set ${\calT \subseteq \calK}$,  $\gamma_\calT \eqdef \bracknorm{1 - \mu} n \avgI{X{\bracksq{\calT}}; Y | X {\bracksq{\calT^c}}}$ for an arbitrary $\mu \in (0,1)$. Encoder $k \in \calK$ uses the key $S_k = \ell_k$ to map the message $W_k = m_k$ onto the codeword $\xn_k \bracknorm{m_k,\ell_k}$. The codewords are then transmitted through the memoryless \ac{MAC} to the legitimate receiver. The decoder, who observes $\yn$ and has complete knowledge of the keys $\ell{{\bracksq{\calK}}}$, operates as follows. 
		\begin{itemize}
			\item If there exists a unique $m{{\bracksq{\calK}}} \in \times_{k=1}^K \intseq{1}{M_k}$ such that $ \bracknorm{ \xn_{\calK}\bracknorm{m{\bracksq{\calK}}, \ell{\bracksq{\calK}}}, \yn} \in \agamman{}$, output $\widehat{W}{{\bracksq{\calK}}} = m{{\bracksq{\calK}}}$,
			\item Else, declare a decoding error. 
		\end{itemize}
		\paragraph{Channel reliability analysis} The decoding error probability $P_e^n$ averaged over all random codebooks satisfies the following. 
		\begin{lemma} \label{lem:reliability}
			For any $\mu \in \bracknorm{0,1}$, an $n$ large enough, and
			\begin{align}
				\log M_k = \bracknorm{1 - \mu} \rho_k n \alpha_n \avgD{P_k}{\pz}, \label{eq:28}
			\end{align}
			for every $k \in \calK$, the probability of decoding error averaged over all random codebooks satisfies
			\begin{align}
				\E{P_e^n} \leq \exp \bracknorm{-\xi n \alpha_n}, \label{eq:29}
			\end{align}
			for an appropriate $\xi>0$.
		\end{lemma}
		The proof of Lemma~\ref{lem:reliability} is provided in {Appendix~\ref{sec:lemma_rel}.}
		\paragraph{Channel resolvability analysis} 
		In the following lemma, we show that the KL divergence between the induced distribution and the covert stochastic process averaged over all random codebooks vanishes in the limit. 
		\begin{lemma} \label{lem:resolvability}
			For any $\mu \in \bracknorm{0,1}$, an $n$ large enough, and
			\begin{align}
				\log M_k L_k  = \bracknorm{1 + \mu} \rho_k n  \alpha_n \avgD{Q_k}{\qz}, \label{eq:30}
			\end{align}
			for every $k \in \calK$, the KL divergence between $\qhatn$ and $\qaln$ averaged over all random codebooks satisfies
			\begin{align}
				\E{\avgD{\qhatn}{\qaln}} \leq \exp \bracknorm{-\xi n \alpha_n}, \label{eq:31}
			\end{align}
			for an appropriate $\xi>0$.
		\end{lemma}
		The proof of Lemma~\ref{lem:resolvability} is provided in {Appendix~\ref{sec:lemma_res}.}
		\paragraph{Identification of a specific code} Using Markov's inequality, we obtain
		\begin{align}
			& \P{P_e^n < 4 \E{P_e^n} \cap \avgD{\qhatn}{\qaln} < 4 \E{\avgD{\qhatn}{\qaln}}} \geq \frac{1}{2}. 
		\end{align}
	Then, we conclude that there must exist at least one coding scheme such that for appropriate constants $\xi_1, \xi_2 > 0$ and an $n$ large enough, we have
		\begin{align}
			P_e^n & \leq \exp \bracknorm{-\xi_1 n \alpha_n}, \label{eq:33} \\
			\avgD{\qhatn}{\qaln} & \leq \exp \bracknorm{-\xi_2 n \alpha_n}. \label{eq:34}
		\end{align}
		\begin{lemma} \label{lem:identificationcode}
		For $n$ large enough and an appropriate constant $\xi_3 > 0$, 
			\begin{align}
			\left| \avgD{\qhatn}{\qzn} - \avgD{\qaln}{\qzn} \right| \leq \exp \bracknorm{- \xi_3 n \alpha_n}, \label{eq:38}
			\end{align}
			provided $\avgD{\qhatn}{\qaln}$ satisfies~\eqref{eq:34}. 
		\end{lemma}	
		The proof of Lemma~\ref{lem:identificationcode} is provided in {Appendix~\ref{sec:proof-lemma-identification}.}
		Using~\eqref{eq:lem1},~\eqref{eq:33},~\eqref{eq:38}, and our choice of $\brackcurl{\alpha_n}_{n \in \mathbb{N}^*}$,                                                                                                           we conclude that there exists at least one coding scheme that satisfies~\eqref{eq:ach_3} and~\eqref{eq:ach_4}. 
		Combining~\eqref{eq:lem1} and~\eqref{eq:38} yields
		\begin{align}
			&\frac{n \alpha_n^2}{2} \bracknorm{1 + \sqrt{\alpha_n}} \chi_n{\bracknorm{\rhovector}} + \exp \bracknorm{- \xi_3 n \alpha_n} \geq \avgD{\qhatn}{\qzn} \nex
			& \dupspace \dupspace \dupspace \dupspace \dupspace \dupspace \geq \frac{n \alpha_n^2}{2} \bracknorm{1 - \sqrt{\alpha_n}} \chi_n{\bracknorm{\rhovector}} - \exp \bracknorm{- \xi_3 n \alpha_n}. \label{eq:42} 
		\end{align}
		We normalize $\log M_k$, where $k \in \calK$, by $\sqrt{n \avgD{\qhatn}{\qzn}}$ using~\eqref{eq:28},~\eqref{eq:30}, and~\eqref{eq:42} to obtain
		\begin{align}
			\limn \frac{\log M_k}{\sqrt{n \avgD{\qhatn}{\qzn}}} & = \bracknorm{1 - \mu}\sqrt{ \frac{{2}}{{\chi{\bracknorm{\rhovector}}}}  }\rho_k \avgD{P_k}{\pz}, \label{eq:44_1} \\
			\limn \frac{\log M_kL_k}{\sqrt{n \avgD{\qhatn}{\qzn}}} & = \bracknorm{1+\mu} \sqrt{ \frac{{2}}{{\chi{\bracknorm{\rhovector}}}}}  \rho_k \avgD{Q_k}{\qz}. \label{eq:44_2}
		\end{align}
		{
		Combining~\eqref{eq:44_1} and~\eqref{eq:44_2}, we obtain
		\begin{align}
			\limn \frac{\log L_k}{\sqrt{n \avgD{\qhatn}{\qzn}}} = \sqrt{\frac{2}{\chi{\bracknorm{\rhovector}}}} \rho_k \bracksq{\bracknorm{1 + \mu} \avgD{Q_k}{\qz} - \bracknorm{1 - \mu} \avgD{P_k}{\pz}}^+. \label{eq:44_3}
		\end{align}
		}

	\end{proof}
	{
	Since $\mu$ in~\eqref{eq:ach_1} is arbitrary, we conclude from Proposition~\ref{prop:achievability} that the covert capacity region contains the region defined by
	\begin{align}
		\bigcup_{\{\rho_k\}_{k\in\calK}\in[0,1]^K:\sum_{k \in \calK} \rho_k = 1}\left\{\{r_k\}_{k\in\calK}:\forall k\in\calK,\quad r_k\leq \sqrt{\frac{2}{\chi\bracknorm{\rhovector}}} \rho_k \avgD{P_k}{\pz}\right\}.  \label{eq:region1}
	\end{align}
	In addition, any achievable covert throughput tuple $r\bracksq{\calK}$ that is characterized by a specific $\rhovector$ and lies on the boundary of the region defined in~\eqref{eq:region1} is associated with an achievable key throughput tuple $ \brackcurl{\sqrt{\frac{2}{\chi{\bracknorm{\rhovector}}}} \rho_k \bracksq{ \avgD{Q_k}{\qz} -  \avgD{P_k}{\pz}}^+}_{k \in \calK}$. 
	
	}

\subsection{Converse proof}
\label{sec:converse-proof}

	\begin{proposition} \label{prop:converse}
		For the channel model described in Section~\ref{sec:channelmodel},  consider a sequence of covert communication schemes with increasing blocklength $n \in \mathbb{N}^*$ characterized by $\epsilon_n \eqdef P_e^n$ and $\delta_n \eqdef \avgD{\qhatn}{\qzn}$ such that $\lim_{{n \to \infty}} \epsilon_n = 0$ and $\lim_{{n \to \infty}} \delta_n = 0$. 
	{
		Then, there exists a vector $\rhovector \eqdef \brackcurl{\rho_k}_{k \in \calK} \in \bracksq{0,1}^K$ with $\sum_{k \in \calK} \rho_k = 1$ and an infinite subset  $\calN \subseteq \mathbb{N}^*$, such that for all $k \in \calK$, 
		\begin{align}
			{\liminf_{\substack{n \to \infty \\ n \in \calN}}} \frac{\log M_k}{\sqrt{n \delta_n}} \leq \sqrt{\frac{2}{\chi{\bracknorm{\rhovector}}}} \rho_k \avgD{P_k}{\pz}. \label{eq:con1}
		\end{align}
		For a sequence of codes that achieves the right hand side of~\eqref{eq:con1} for all $k \in \calK$, we have
		\begin{align}
			{\limsup_{\substack{n \to \infty \\ n \in \calN}}} \frac{\log M_kL_k}{\sqrt{n \delta_n}} \geq  \sqrt{\frac{2}{\chi{\bracknorm{\rhovector}}}} \rho_k \avgD{Q_k}{\qz}, \label{eq:con2}
		\end{align} 
		for all $k \in \calK$. 
		}
	\end{proposition}
	\begin{proof}
{
Consider a sequence of covert communication schemes with increasing blocklength $n$ characterized by $\epsilon_n \eqdef P_e^n$ and $\delta_n \eqdef \avgD{\qhatn}{\qzn}$, and $\log M_k$ takes the maximum value such that $\limn \log M_k = \infty$ for all $k \in \calK$.
Each user $k$ transmits an $n-$length codeword ${\bxn}_k = \bracknorm{{X}_{k1}, {X}_{k2}, \ldots, {X}_{kn}} \in \calX^n$, where $n \in \mathbb{N}^*$, to the receiver. }
%
%
	For $j \in \intseq{1}{n}$, we denote the distribution of each symbol $X_{kj}$ on $\calX$ by $\Pi_{X_{kj}}$, where
		\begin{align}
			& \Pi_{X_{kj}}(x) \eqdef \frac{ \sum_{m_k = 1}^{M_k} \sum_{\ell_k = 1}^{L_k} \indic{X_{kj}(m_{k}, \ell_k)=x} }{M_kL_k}. \label{eq:72}
	\end{align}
	We define $\Pi_{X_{kj}}(1) = 1 - \Pi_{X_{kj}}(0) \eqdef \mun{kj}$.  
	Note that $\mun{kj}$ depends on $n$, the transmitter index $k$, and the symbol position $j$. 
{
	For every $n \in \mathbb{N}^*$, we define a permutation $\pi_{k^*}^{(n)}$ of $\intseq{1}{n}$ to define a new code such that 
	\begin{align}
		\bracknorm{k^*, 1} = \arg \max_{\bracknorm{k,j} \in \calK \times \intseq{1}{n} } \mun{kj}.  \label{eq:79_6}
	\end{align}
	Since the channel is memoryless, the performance of the new code that satisfies~\eqref{eq:79_6} matches that of the original code. 
	Hence, without loss of generality, we only study the sequence of codes for which~\eqref{eq:79_6} holds for every $n \in \mathbb{N}^*$. 	
	Note that the sequence $\brackcurl{\brackcurl{\smash{\mun{k1}}}_{k \in \calK}}_{n \in \mathbb{N}^*}$ belongs to $\bracksq{0,1}^K$ which is a closed and bounded set. 
	Hence, we can extract a convergent subsequence $\brackcurl{\brackcurl{\smash{\mun{k1}}}_{k \in \calK}}_{n \in \calN^*}$, where $\calN^* \subseteq \mathbb{N}^* $ is an infinite set, with limit $\brackcurl{\smash{\mu_{k1}^*}}_{k \in \calK}$. 
	Let us now assume that the sequence $\brackcurl{\smash{\mu_{k1}^*}}_{k \in \calK} \in \bracksq{0,1}^K$ is not an all-zero sequence. 
	}
	
	For $j \in \intseq{1}{n}$, we denote the $K$-length vector $\brackcurl{x_{kj}}_{k \in \calK} $ by $x_{\bracknorm{j}}{\bracksq{\calK}}$. The warden makes an observation $\bzn$ of length $n$, whose distribution is denoted by $\qhatn$. For $j \in \intseq{1}{n}$, we denote the distribution of each component $Z_j$ of $\bzn$ by $\qhat_j$, where
	\begin{align}
		\qhat_j(z) & \eqdef \frac{1}{\prod_{k \in \calK} M_k L_k}  \sum\limits_{m{\bracksq{\calK}}} \sum\limits_{\ell{\bracksq{\calK}}} W_{Z|X{\bracksq{\calK}}}(z|\brackcurl{x_{kj}(m_k, \ell_k)}_{k \in \calK} ) \\
		& =  \sum_{x_{(j)}{\bracksq{\calK}}} \bracknorm{\prod_{k \in \calK} \Pi_{X_{kj}}(x_{kj}) } W_{Z|X{\bracksq{\calK}}} \bracknorm{z|x_{(j)}{\bracksq{\calK}} } \\
		& \stackrel{(a)}{=} \sum_{\calT \subseteq \calK}  \bracknorm{\prod_{k \in \calT} \mun{kj}} \bracknorm{\prod_{k \in \calT^c} \bracknorm{1 - \mun{kj}}} Q_{\calT}(z), \label{eq:74}
	\end{align}
	{where $(a)$ follows from the definition of $Q_\calT(z) \eqdef W_{Z|X{{\bracksq{\calK}}}}(z|x_\calT)$ in~\eqref{eq:1}.} 
	{Alternatively, using Lemma~\ref{lem:altrepresentation} in the appendix, we write }
	\begin{align}
		\qhat_j(z) &= \qz(z) + \sum_{\substack{\calT \subseteq \calK: \\ \calT \neq \emptyset}} \bracknorm{\prod_{k \in \calT} \mun{kj} } G_{\calT}(z). \label{eq:74_1}
	\end{align}
	From the definition of $\delta_n$, we have
	\begin{align}
		\delta_n &= \avgD{\qhatn}{\qzn} \\
		&= -\avgH{\bzn} + \E[\qhatn]{\log \frac{1}{\qzn(\bzn)}} \\ 
		& = - \bracknorm{ \sum_{j=1}^n \avgH{Z_j|\bzn^{\, j-1}}} + \E[\qhatn]{\sum_{j=1}^n \log \frac{1}{\qz(Z_j)}}  \\ 
		&\stackrel{(a)}{\geq} \sum_{j=1}^n \bracknorm{-\avgH{Z_j} + \E[\qhat_j]{\log \frac{1}{\qz(Z_j)}}} \\
		& = \sum_{j=1}^n \avgD{\qhat_j}{\qz},\label{eq:75}
	\end{align}
	where $(a)$ follows from the fact that conditioning reduces entropy.
	Since $\limn \delta_n=0$ and KL divergence is non-negative, it follows from~\eqref{eq:75} that 
	\begin{align}
		\limn \avgD{\qhat_j}{\qz} = 0, \label{eq:76}
	\end{align}
	for all $j \in \intseq{1}{n}$. 
	Applying Pinsker's inequality on~\eqref{eq:76}, we obtain $\limn \V{\qhat_j , \qz} =0$, which implies that $\forall z \in \calZ$,
	{
	\begin{align}
		\limn \abs{\qhat_j(z) - \qz(z)} &= 0, \label{eq:79_2}  \\
		\limn \qhat_j(z) &= \qz(z). \label{eq:79_4}
	\end{align} 
	Fixing $j = 1$ and by using~\eqref{eq:74} and~\eqref{eq:79_4}, for $n \in \calN^*$, we obtain
	\begin{align}
		\lim_{\substack{n \to \infty \\ n \in \calN^*}} \bracknorm{\sum_{\calT \subseteq \calK}  \bracknorm{\prod_{k \in \calT} \mun{k1}} \bracknorm{\prod_{k \in \calT^c} \bracknorm{1 - \mun{k1}}} Q_{\calT}(z)} & = \qz(z), \\
		\sum_{\calT \subseteq \calK}  \bracknorm{\prod_{k \in \calT} \mu_{k1}^*} \bracknorm{\prod_{k \in \calT^c} \bracknorm{1 - \mu_{k1}^*}} Q_{\calT}(z) & = \qz(z).  \label{eq:79_5}
	\end{align}
	Since we assumed that the sequence $\brackcurl{\mu_{k1}^*}_{k \in \calK}$ is not an all-zero sequence,~\eqref{eq:79_5} implies that $\qz$ is a convex combination of $\brackcurl{Q_{\calT}}_{\calT \subseteq \calK: \, \calT \neq \emptyset}$. 
	Note that the convex combination in~\eqref{eq:79_5} does not require the transmitters to coordinate, which is the case in our channel model, since the input from each user is independent of the inputs from other users. 
	Since~\eqref{eq:79_5} contradicts the assumption made in Section~\ref{sec:channelmodel}, our assumption about $\brackcurl{\mu_{k1}^*}_{k \in \calK}$ is incorrect, and we have 
	\begin{align}
		 \lim_{\substack{n \to \infty \\ n \in \calN^*}} \mu_{k1}^{(n)} = 0, 
	\end{align}
	for all $k \in \calK$, which implies that 
	\begin{align}
		\lim_{\substack{n \to \infty \\ n \in \calN^*}} \mu_{k^*1}^{(n)} = 0. \label{eq:79_7}
	\end{align}
	Subsequently, from~\eqref{eq:79_6} and~\eqref{eq:79_7}, we obtain
	\begin{align}
		\lim_{\substack{n \to \infty \\ n \in \calN^*}} \mu_{kj}^{(n)} = 0, 
	\end{align}
	for all $\bracknorm{k,j} \in \calK \times \intseq{1}{n}$. 
	Henceforth, we only consider the subsequence of codes with blocklength $n \in \calN^*$. 
	}
{	Next, for $j \in \intseq{1}{n}$, define 
	\begin{align}
	    \Psin{j}(z) \eqdef \qhat_j(z) - \qz(z).     \label{eq:79_8}
	\end{align}
	Note that $\sum_z \Psin{j}(z) = 0$. }
	{Also note that from~\eqref{eq:79_2} and~\eqref{eq:79_8}, we have $\limn \Psin{j}(z)=0$ for all $j \in \intseq{1}{n}$ and $\forall z \in \calZ$. }
	We lower bound $\avgD{\qhat_j}{\qz}$ for $n$ large enough by
	\begin{align}
		\avgD{\qhat_j}{\qz} & = \sum_z \qhat_j(z) \log \frac{\qhat_j(z)}{\qz(z)}  \\
		&= \sum_z \qz(z) \bracknorm{1+\frac{\Psin{j}(z)}{\qz(z)}} \log \bracknorm{1 + \frac{\Psin{j}(z)}{\qz(z)} }  \\ 
		& \stackrel{(a)}{\geq} \sum_z \bracknorm{ \frac{\bracknorm{\Psin{j}(z)}^2}{2 \qz(z)} - \frac{\bracknorm{\Psin{j}(z)}^3}{2\qz^2(z)}} + \sum_{z: \Psin{j}(z)<0} \frac{2 \bracknorm{\Psin{j}(z)}^3}{3\qz^2(z)} \\
		& \geq \sum_z \frac{\bracknorm{\Psin{j}(z)}^2}{2 \qz(z)}  \bracknorm{1 - \frac{\Psin{j}(z)}{\qz(z)} - \frac{4 \left| \Psin{j}(z) \right|}{3 \qz(z)} }, \label{eq:84}
	\end{align}
	where $(a)$ follows from the inequality $\log \bracknorm{1+x} > x - \frac{x^2}{2}$ for $x \geq 0$ and\footnote{{Note that for $n$ large enough, we can ensure that $\Psin{j}(z) \in \bracksq{-\frac{1}{2}, 0}$ if $\Psin{j}(z)< 0$ since $\limn \Psin{j}(z)=0$.}} $\log \bracknorm{1+x} > x - \frac{x^2}{2} + \frac{2x^3}{3}$ for $x \in \bracksq{-\frac{1}{2}, 0}$. 
	For $j \in \intseq{1}{n}$, define $\xijn(z) \eqdef \frac{\Psin{j}(z)}{\qz(z)} + \frac{4 \left| \Psin{j}(z) \right|}{3 \qz(z)}$ and $\xin(z) \eqdef \max_{j \in \intseq{1}{n}} \xijn(z)$.
	{Since $\limn \Psin{j}(z)=0$, we have $\limn \xijn(z) = 0$ for all $j \in \intseq{1}{n}$. 
	From~\eqref{eq:74_1} and~\eqref{eq:79_8}, for $j \in \intseq{1}{n}$, we write
	\begin{align}
	    \abs{\Psin{j}(z)} & = \abs{\qhat_{j}(z) - \qz(z)} \\
	    & \leq \sum_{\calT \subseteq \calK: \calT \neq \emptyset} \bracknorm{\prod_{k \in \calT} \mun{kj} } \abs{G_\calT(z)} \\
	    & \stackrel{(a)}{\leq} \mun{k^*1} \bracknorm{ \sum_{\calT \subseteq \calK: \calT \neq \emptyset} \abs{G_\calT(z)}},  \label{eq:84_1}
	\end{align}
	where $(a)$ follows from~\eqref{eq:79_6} and the fact that $\mun{kj} \in [0,1]$ for all $k \in \calK$ and $j \in \intseq{1}{n}$. 
	Note that the term inside the parentheses in~\eqref{eq:84_1} is positive and bounded. 
	Consequently, for $z \in \calZ$, 
	\begin{align}
	    \max_{j \in \intseq{1}{n}} \abs{\Psin{j}(z)} \leq \mun{k^*1} \bracknorm{ \sum_{\calT \subseteq \calK: \calT \neq \emptyset} \abs{G_\calT(z)}}. 
	\end{align}
	From the definition of $\xijn(z)$, we have
	\begin{align}
	    \xijn(z) & = \frac{\Psin{j}(z)}{\qz(z)} + \frac{4\abs{\Psin{j}(z)}}{3 \qz(z)} \\
	    & \leq \frac{\abs{\Psin{j}(z)}}{\qz(z)} + \frac{4\abs{\Psin{j}(z)}}{3 \qz(z)} \\
	    & = \frac{7\abs{\Psin{j}(z)}}{3\qz(z)}. 
	\end{align}
	Consequently, we have
	\begin{align}
	    \xin(z) & = \max_{j \in \intseq{1}{n}} \xijn(z) \\
	    & \leq \frac{7}{3\qz(z)} \max_{j \in \intseq{1}{n}} \abs{\Psin{j}(z)} \\
	    & \leq \frac{7}{3\qz(z)} \mun{k^*1} \bracknorm{ \sum_{\calT \subseteq \calK: \calT \neq \emptyset} \abs{G_\calT(z)}}. \label{eq:84_2}
	\end{align}
	Note that, by definition, $\xijn(z)$ is non-negative irrespective of the sign of $\Psin{j}(z)$. 
	Then, using~\eqref{eq:79_7} and~\eqref{eq:84_2}, we conclude that $\lim_{\substack{n \to \infty \\ n \in \calN^*}} \xin(z) = 0$. }
	Using~\eqref{eq:84}, we lower bound~\eqref{eq:75} by
	\begin{align}
		\delta_n &\geq \sum_{j = 1}^n  \sum_z \frac{\bracknorm{\Psin{j}(z)}^2}{2 \qz(z)}  \bracknorm{1 - \xijn(z)} \\
		& \geq \sum_z \frac{ \bracknorm{1 - \xin(z)}}{2 \qz(z)} \sum_{j = 1}^n \bracknorm{\Psin{j}(z)}^2. \label{eq:85}
	\end{align}
	For $k \in \calK$, we upper bound $\log M_k$ using standard techniques, 
	\begin{align}
		\log M_k 
		& \stackrel{(a)}{\leq} \avgI{W_k;\byn S_k} + \Hb{\epsilon_n} + \epsilon_n \log M_k \\
		& \leq \avgI{W_k S_k; \byn} + \Hb{\epsilon_n} + \epsilon_n \log M_k \\ 
		& = \avgI{\bxn_k; \byn} + \Hb{\epsilon_n} + \epsilon_n \log M_k  \label{oeq:con11_a} \\
		& = \avgH{\bxn_k} - \avgH{\bxn_k|\byn} + \Hb{\epsilon_n} + \epsilon_n \log M_k \\
		& \stackrel{(b)}{\leq} \avgH{\bxn_k|\bxn\bracksq{\calK \setminus\brackcurl{k}}} - \avgH{\bxn_k|\byn \bxn\bracksq{\calK \setminus \brackcurl{k}}} + \Hb{\epsilon_n} + \epsilon_n \log M_k \\ 
		& = \avgI{\bxn_k; \byn | \bxn\bracksq{\calK \setminus \brackcurl{k}}} + \Hb{\epsilon_n} + \epsilon_n \log M_k \\
		& = \avgH{\byn | \bxn\bracksq{\calK \setminus \brackcurl{k}}} - \avgH{\byn|\bxn{\bracksq{\calK}}} + \Hb{\epsilon_n} + \epsilon_n \log M_k  \\
		& \stackrel{(c)}{\leq} \sum_{j=1}^n \avgH{Y_j|X_{(j)}\bracksq{\calK \setminus \brackcurl{k}}}  - \sum_{j=1}^n \avgH{Y_j|X_{(j)}{\bracksq{\calK}}} + \Hb{\epsilon_n} + \epsilon_n \log M_k \\
		& = \sum_{j=1}^n \avgI{X_{kj}; Y_j |X_{(j)} \bracksq{\calK \setminus \brackcurl{k}} } + \Hb{\epsilon_n}  + \epsilon_n \log M_k, \label{eq:86}
	\end{align}
	where $(a)$ follows from Fano's inequality, $(b)$ follows from the fact that $\bxn_k$ and $\bxn{\bracksq{\calK\! \setminus\! {\brackcurl{k}}}}$ are mutually independent and the fact that conditioning reduces entropy, and $(c)$ follows from the fact that conditioning reduces entropy and  the memoryless property of the channel. 
	We expand the mutual information term in~\eqref{eq:86} as
	\begin{align}
		& \avgI{X_{kj}; Y_j |X_{\bracknorm{j}} \bracksq{\calK \setminus \brackcurl{k}} } \nex
		& =  \sum_{\calT \subseteq \calK}  \bracknorm{\prod_{i \in \calK} \Pi_{X_{ij}} \bracknorm{x_{\calT,i}}} \avgD{P_\calT}{W_{Y_j|X_{\bracknorm{j}}\bracksq{\calK \setminus \brackcurl{k}}=x_{\calT} \bracksq{\calK \setminus \brackcurl{k}}}  }   \\
		& = \sum_{\calT \subseteq \calK}  \bracknorm{\prod_{i \in \calK} \Pi_{X_{ij}} \bracknorm{x_{\calT,i}}} \avgD{P_\calT}{\pz} \nex
		& \dupspace - \sum_y \sum_{\calT \subseteq \calK}  \bracknorm{\prod_{i \in \calK} \Pi_{X_{ij}} \bracknorm{x_{\calT,i}}}P_{\calT}\bracknorm{y}  \log \frac{W_{Y_j|X_{\bracknorm{j}}\bracksq{\calK \setminus \brackcurl{k}}} \bracknorm{y|x_{\calT} \bracksq{\calK \setminus \brackcurl{k}}} }{\pz(y)}. \label{eq:87}
	\end{align} 
	Defining {$\mun{\max} \eqdef \mun{k^*1} $} and $d_1 \eqdef 2^K \max_{\substack{\calT \subseteq \calK: \card{\calT}>1}} \avgD{P_\calT}{\pz} $, we upper bound the first term in~\eqref{eq:87} by
	\begin{align}
		& \sum_{\calT \subseteq \calK}  \bracknorm{\prod_{i \in \calK} \Pi_{X_{ij}} \bracknorm{x_{\calT,i}}} \avgD{P_\calT}{\pz} \nex 
		 &  \stackrel{(a)}{=} \sum_{\substack{\calT \subseteq \calK:\\ \abs{\calT}>1}}  \bracknorm{\prod_{i \in \calT}\mun{ij}} \bracknorm{\prod_{i \in \calT^c} \bracknorm{1 - \mun{ij}} } \avgD{P_\calT}{\pz}  +  \sum_{i \in \calK} \mun{ij} \bracknorm{ \prod_{i' \in \calK \setminus \brackcurl{i}} \bracknorm{1 - \mun{i'j}} } \avgD{P_i}{\pz}  \\
		 &  \stackrel{(b)}{\leq}  d_1 \mun{\max}{ \sum_{i \in \calK} \mun{ij} }  + \sum_{i \in \calK} \mun{ij} \avgD{P_{i}}{\pz} , \label{eq:88}
	\end{align}
	where $(a)$ follows from splitting the sum into two based on the number of $1'$s in $x_\calT$, and $(b)$ follows from the fact that $ \!\bracknorm{\!1 \!- \!\mun{i'j}\!}\! \leq \!1$ for all $\bracknorm{i',j}\! \in\! \calK \times \intseq{1}{n}$.  
	Defining $d_2 \!\eqdef\! 2^K\!\!\! \max\limits_{i \in \calK \setminus \brackcurl{k}} \!\!\avgD{W_{Y_j|X_{\bracknorm{j}}\bracksq{\calK \setminus\brackcurl{k}} = x_{i} \bracksq{\calK \setminus\brackcurl{k}}}  }{\pz}$,  we lower bound the second term in~\eqref{eq:87} by 
		\begin{align}
		& \sum_y \sum_{\calT \subseteq \calK}  \bracknorm{\prod_{i \in \calK} \Pi_{X_{ij}} \bracknorm{x_{\calT,i}}}P_{\calT}\bracknorm{y} \log \frac{W_{Y_j|X_{\bracknorm{j}}\bracksq{\calK \setminus \brackcurl{k}}} \bracknorm{y|x_{\calT} \bracksq{\calK \setminus \brackcurl{k}}} }{\pz(y)} \nex
		& \dupspace =\!\!\sum_{\calT \subseteq \calK \setminus \brackcurl{k}} \!\!\!\bracknorm{\prod_{i \in \calK \setminus \brackcurl{k}} \Pi_{X_{ij}}(x_{\calT,i}) }\!\!\sum_y \sum_{x} \Pi_{X_{kj}}(x) W_{Y_j|X_{\bracknorm{j}}\bracksq{\calK \setminus \brackcurl{k}}X_{kj}} \bracknorm{y|x_{\calT} \bracksq{\calK{\setminus}{\brackcurl{k}}}x} \nex
		& \dupspace  \dupspace \dupspace \times  \log \frac{W_{Y_j|X_{\bracknorm{j}}\bracksq{\calK \setminus \!\brackcurl{k}}} \bracknorm{y|x_{\calT}\!\bracksq{\calK \!\setminus \!\!\brackcurl{k}}} }{\pz(y)}  \\
		& \dupspace \stackrel{(a)}{ =} \sum_{\calT \subseteq {\calK \setminus \brackcurl{k}}} \bracknorm{\prod_{i \in { \calK \setminus \brackcurl{k}}} \Pi_{X_{ij}} \bracknorm{x_{\calT,i}}}\avgD{W_{Y_j|X_{\bracknorm{j}}\bracksq{\calK \setminus\brackcurl{k}} = x_{\calT} \bracksq{\calK \setminus\brackcurl{k}}}  }{\pz} \\
		& \dupspace \geq  \sum_{\substack{ \calT \subseteq {\calK \setminus \brackcurl{k}}:  \card{\calT}=1}} \bracknorm{\prod_{i \in { \calK \setminus \brackcurl{k}}} \Pi_{X_{ij}} \bracknorm{x_{\calT,i}}}\avgD{W_{Y_j|X_{\bracknorm{j}}\bracksq{\calK \setminus \brackcurl{k}} = x_{\calT} \bracksq{\calK \setminus\brackcurl{k}}}  }{\pz} \\
		& \dupspace = \sum_{i \in {\calK \setminus \brackcurl{k}}} \mun{ij} \bracknorm{\prod_{\substack{ i' \in { \calK \setminus \brackcurl{i,k}} }} \bracknorm{1 - \mun{i'j}} }\avgD{W_{Y_j|X_{\bracknorm{j}}\bracksq{\calK \setminus\brackcurl{k}} = x_{i} \bracksq{\calK \setminus\brackcurl{k}}}  }{\pz} \label{eq:89_1}  \\
		& \dupspace \stackrel{(b)}{=} \sum_{i \in {\calK \setminus \brackcurl{k}}} \mun{ij} \bracknorm{ 1 + \sum_{\calT \subseteq \calK \setminus \brackcurl{i,k} } \bracknorm{-1}^{\card{\calT}} \bracknorm{\prod_{i' \in \calT} \mun{i'j} } } \avgD{W_{Y_j|X_{\bracknorm{j}}\bracksq{\calK \setminus\brackcurl{k}} = x_{i} \bracksq{\calK \setminus\brackcurl{k}}}  }{\pz} \\
		& \dupspace \geq \sum_{i \in {\calK \setminus \brackcurl{k}}} \mun{ij} \bracknorm{ 1 - \sum_{\substack{\calT \subseteq \calK \setminus \brackcurl{i,k}: ~ \card{\calT} \text{ is odd } }}  \bracknorm{\prod_{i' \in \calT} \mun{i'j} } } \avgD{W_{Y_j|X_{\bracknorm{j}}\bracksq{\calK \setminus\brackcurl{k}} = x_{i} \bracksq{\calK \setminus\brackcurl{k}}}  }{\pz} \\
		& \dupspace \geq \sum_{i \in {\calK \setminus \brackcurl{k}}} \mun{ij} \bracknorm{ 1 - 2^K \mun{\max} } \avgD{W_{Y_j|X_{\bracknorm{j}}\bracksq{\calK \setminus\brackcurl{k}} = x_{i} \bracksq{\calK \setminus\brackcurl{k}}}  }{\pz} \\
		& \dupspace \geq \sum_{i \in {\calK \setminus \brackcurl{k}}} \mun{ij} \avgD{W_{Y_j|X_{\bracknorm{j}}\bracksq{\calK \setminus\brackcurl{k}} = x_{i} \bracksq{\calK \setminus\brackcurl{k}}}  }{\pz} -  d_2 \mun{\max} { \sum_{{i \in {\calK \setminus \brackcurl{k}}}} \mun{ij} } , \label{eq:89}
	\end{align}
	where $(a)$ follows from $\sum\limits_{x}\! \Pi_{X_{kj}}\!(x) W_{Y_j|X_{\bracknorm{j}}\bracksq{\calK \setminus \! \brackcurl{k}}X_{kj}} \!\!\bracknorm{y|x_{\calT} \!\bracksq{\calK{\setminus}\!{\brackcurl{k}}}\!x}  \!=\! W_{Y_j|X_{\bracknorm{j}}\bracksq{\calK \setminus \! \brackcurl{k}}} \bracknorm{y|x_{\calT} \bracksq{\calK \setminus \brackcurl{k}}} $ and $(b)$ follows from the fact that $\prod_{ i' \in \calK \setminus \brackcurl{i,k}} \bracknorm{1 - \smash{\mun{i'j}}}  =  1 + \sum_{\calT \subseteq \calK \setminus \brackcurl{i,k} } \smash{\bracknorm{-1}^{\card{\calT}}} \bracknorm{\prod_{i' \in \calT} \mun{i'j} }$. 
	Note that we can write $W_{Y_j|X_{(j)}{\bracksq{\calK \setminus \brackcurl{k}}} }(y|x_i{\bracksq{\calK \setminus \brackcurl{k}}}) = \bracknorm{1 - \smash{\mun{kj}}} P_i(y) + \mun{kj} P_{\brackcurl{i,k}}(y)$. 
	We define $d_3 \eqdef \abs{ \sum_y \bracknorm{P_{\brackcurl{i,k}}(y) - P_i(y)}\log \frac{P_i(y)}{\pz(y)}}$. Note that $d_3$ is bounded since $P_i \ll \pz$. Then, we lower bound the KL divergence term in~\eqref{eq:89} by
	\begin{align} 
		\avgD{W_{Y_j|X_{\bracknorm{j}}\bracksq{\calK \setminus\brackcurl{k}} = x_{i} \bracksq{\calK \setminus\brackcurl{k}}}  }{\pz} & =  \sum_y W_{Y_j|X_{\bracknorm{j}}\bracksq{\calK \setminus\brackcurl{k}}} \bracknorm{y|x_{i} \bracksq{\calK \setminus\brackcurl{k}}} \log \frac{P_i(y)}{\pz(y)} \nex
		& \dupspace \dupspace + \avgD{W_{Y_j|X_{\bracknorm{j}}\bracksq{\calK \setminus\brackcurl{k}} = x_{i} \bracksq{\calK \setminus\brackcurl{k}}}  }{P_i}  \\
		& \geq \sum_y P_i(y) \bracknorm{1 + \mun{kj} \frac{P_{\brackcurl{i,k}}(y) - P_i(y)}{P_i(y)} } \log \frac{P_i(y)}{\pz(y)} \\
		& \geq \avgD{P_i}{\pz} - d_3 \mun{\max} . \label{eq:90}
	\end{align} 
	Defining $d_4 \eqdef d_1 + d_2 + d_3 $ and combining~\eqref{eq:87},~\eqref{eq:88},~\eqref{eq:89}, and~\eqref{eq:90}, we obtain
	\begin{align}
		\avgI{X_{kj}; Y_j |X_j \bracksq{\calK\!\setminus \!\brackcurl{k}} } &\leq \mun{kj} \avgD{P_k}{\pz} + d_4 \mun{\max} { \sum_{i \in \calK} \mun{ij} }. \label{eq:92}
	\end{align}
Next, we normalize $\log M_k$, where $k \in \calK$, by $\sqrt{n \delta_n}$. 
{Using~\eqref{eq:85},~\eqref{eq:86}, and~\eqref{eq:92}, for $n$ large enough, we obtain
	\begin{align}
		\frac{\log M_k}{\sqrt{n \delta_n }} & \leq  \frac{ \sum_{j=1}^n \mun{kj} {\avgD{P_k}{\pz}}+ d_4 \mun{\max} \sum_{i \in \calK}  \sum_{j=1}^n  \mun{ij} + \Hb{\epsilon_n}  }{\bracknorm{1 - \epsilon_n} \sqrt{n \sum_z \frac{\bracknorm{1 - \xin(z)}}{2 \qz(z)} \sum_{j = 1}^n \bracknorm{\Psin{j}(z)}^2 } } \\
		& \leq  \frac{ \bracknorm{ \sum_{i \in \calK} \sum_{j=1}^n \mun{ij}} \bracknorm{\avgD{P_k}{\pz} \frac{\sum_{j=1}^n \mun{kj}}{\sum_{i \in \calK}\sum_{j=1}^n \mun{ij} } + d_4\mun{\max}  + \frac{\Hb{\epsilon_n}}{ \sum_{i \in \calK} \sum_{j=1}^n \mun{ij}} }  }{\bracknorm{1 - \epsilon_n} \sqrt{n \sum_z \frac{\bracknorm{1 - \xin(z)}}{2 \qz(z)} \sum_{j = 1}^n \bracknorm{\Psin{j}(z)}^2 } } \\
		& =  \frac{  {\avgD{P_k}{\pz} \frac{\sum_{j=1}^n \mun{kj}}{\sum_{i \in \calK}\sum_{j=1}^n \mun{ij} } + d_4 \mun{\max}  + \frac{\Hb{\epsilon_n}}{ \sum_{i \in \calK} \sum_{j=1}^n \mun{ij}} }  }{\bracknorm{1 - \epsilon_n} \sqrt{n \sum_z \frac{\bracknorm{1 - \xin(z)}}{2 \qz(z)}  { \frac{ \sum_{j = 1}^n \bracknorm{ \Psin{j}(z)}^2}{\bracknorm{ \sum_{i \in \calK} \sum_{j=1}^n \mun{ij}}^2} } } },  \\
		& \stackrel{(a)}{\leq} \frac{  {\avgD{P_k}{\pz} \frac{\sum_{j=1}^n \mun{kj}}{\sum_{i \in \calK}\sum_{j=1}^n \mun{ij} } + d_4 \mun{\max} + \frac{\Hb{\epsilon_n}}{ \sum_{i \in \calK} \sum_{j=1}^n \mun{ij}} }  }{\bracknorm{1 - \epsilon_n} \sqrt{ \sum_z \frac{\bracknorm{1 - \xin(z)}}{2 \qz(z)}  \bracknorm{  \frac{\sum_{j = 1}^n\Psin{j}(z)}{\sum_{i \in \calK} \sum_{j=1}^n \mun{ij}} }^2 } },  \label{eq:94} 
	\end{align}
where  
$(a)$ follows from the fact that $ n \sum_{j = 1}^n\bracknorm{\Psin{j}(z)}^2 \geq \bracknorm{\sum_{j=1}^n \Psin{j}(z)}^2 $ according to the Cauchy-Schwarz inequality. 
Note that since $\bracknorm{1 - \xin(z)}$ is positive for $n$ large enough, our application of Cauchy-Schwarz inequality in~\eqref{eq:94} is valid. 
From the definition of $\Psin{j}(z)$ in~\eqref{eq:79_8}, we have
\begin{align}
    \Psin{j}(z) &= \sum_{i \in \calK} \mun{ij}G_i(z) + \sum_{\substack{\calT \subseteq \calK: \\ \card{\calT} \geq 2 }} \bracknorm{ \prod_{k \in \calT} \mun{kj} } G_\calT(z) \displaybreak[0] \\
    & = \sum_{i \in \calK} \bracknorm{ \mun{ij} G_i(z) + \sum_{\substack{\calT \subseteq \calK: \\ i \in \calT, \card{\calT} \geq 2, \\ \forall k \in \calT, k\geq i }} \mun{ij} \bracknorm{\prod_{k \in \calT \setminus \brackcurl{i}} \mun{kj} } G_\calT(z) }  \\
    & =  \sum_{i \in \calK} \mun{ij} \bracknorm{  G_i(z) + \sum_{\substack{\calT \subseteq \calK: \\ i \in \calT, \card{\calT} \geq 2, \\ \forall k \in \calT, k\geq i }}  \bracknorm{\prod_{k \in \calT \setminus \brackcurl{i}} \mun{kj} } G_\calT(z) }. \label{eq:94_2}
\end{align}
Define $d_5 \eqdef 2^K \max\limits_{z \in \calZ} \max\limits_{\calT \subseteq \calK: \card{\calT}>1} \abs{G_{\calT}(z)} $. 
If $ \sum_{i \in \calK} \mun{ij}  G_i(z) \leq 0 $, we upper bound~\eqref{eq:94_2} by
\begin{align}
    \Psin{j}(z) &\leq \sum_{i \in \calK} \mun{ij} \bracknorm{  G_i(z) + d_5 \mun{\max}  },   \label{eq:94_3}
\end{align}
which is negative for $n$ large enough. 
If $ \sum_{i \in \calK} \mun{ij}  G_i(z) \geq 0 $, we lower bound~\eqref{eq:94_2} by
\begin{align}
        \Psin{j}(z) &\geq \sum_{i \in \calK} \mun{ij} \bracknorm{  G_i(z) - d_5 \mun{\max}  },  \label{eq:94_4}
\end{align}
which is positive for $n$ large enough. 
Consequently, for $n$ large enough, combining~\eqref{eq:94},~\eqref{eq:94_3} and~\eqref{eq:94_4}, we obtain
\begin{align}
    \frac{\log M_k}{\sqrt{n \delta_n }} & {\leq} \frac{  {\avgD{P_k}{\pz} \frac{\sum_{j=1}^n \mun{kj}}{\sum_{i \in \calK}\sum_{j=1}^n \mun{ij} } + d_4 \mun{\max} + \frac{\Hb{\epsilon_n}}{ \sum_{i \in \calK} \sum_{j=1}^n \mun{ij}} }  }{\bracknorm{1 - \epsilon_n} \sqrt{ \sum_z \frac{\bracknorm{1 - \xin(z)}}{2 \qz(z)}  \bracknorm{  \frac{\sum_{j = 1}^n\Psin{j}(z)}{\sum_{i \in \calK} \sum_{j=1}^n \mun{ij}} }^2 } } \\
    & = \frac{  {\avgD{P_k}{\pz} \frac{\sum_{j=1}^n \mun{kj}}{\sum_{i \in \calK}\sum_{j=1}^n \mun{ij} } + d_4 \mun{\max} + \frac{\Hb{\epsilon_n}}{ \sum_{i \in \calK} \sum_{j=1}^n \mun{ij}} }  }{\bracknorm{1 - \epsilon_n} \sqrt{ \sum_z \frac{\bracknorm{1 - \xin(z)}}{2 \qz(z)}  \bracknorm{  \frac{ \sum_{a \in \calK} \bracknorm{G_a(z) + \bigO{\mun{\max}}} \sum_{j=1}^n \mun{a j} }{\sum_{i \in \calK} \sum_{j=1}^n \mun{ij}}}^2 } }. \label{eq:94_5} 
\end{align}
Combining~\eqref{eq:86} and~\eqref{eq:92} with the fact that $\lim_{\substack{n \to \infty \\ n \in \calN^*}} \log M_k = \infty$, we conclude that $\lim_{\substack{n \to \infty \\ n \in \calN^*}}  \sum_{i \in \calK} \sum_{j=1}^n \allowbreak \mun{ij} = \infty $.  
Note that $\frac{\sum_{j=1}^n \mun{a j}}{\sum_{i \in \calK} \sum_{j=1}^n \mun{ij}}$ is bounded between $0$ and $1$ for any $a \in \calK$. 
We extract a convergent subsequence $\brackcurl{\frac{\sum_{j=1}^n \mun{a j}}{\sum_{i \in \calK} \sum_{j=1}^n \mun{ij}}}_{n \in \calN^\dagger}$, where $\calN^\dagger \subseteq \calN^*$ is an infinite set, with limit $\rho_a$. 
Note that $\sum_{a \in \calK} \rho_a =1$.
Since we have assumed in Section~\ref{sec:channelmodel} that there exists no $\brackcurl{\rho_k}_{k \in \calK}$ for which  $\sum_{k\in\calK}\rho_k  Q_k(z)=Q_\emptyset(z)$ for all $z \in \calZ$, the denominator in~\eqref{eq:94_5} is non-zero. 
Henceforth, we only consider the subsequence of codes with blocklength $n \in \calN^\dagger$. }
	Defining $\rhovector \eqdef \brackcurl{\rho_k}_{k \in \calK}$, we obtain from~\eqref{eq:94_5}, 
	\begin{align}
		\liminf_{\substack{n \to \infty \\ n \in \calN^\dagger}} \frac{\log M_k}{\sqrt{n\delta_n}} & \leq \sqrt{2} \rho_k \frac{\avgD{P_k}{\pz}}{\sqrt{\sum_z \frac{\bracknorm{\sum_{i \in \calK} \rho_{i} \bracknorm{ Q_i(z) - \qz(z)}}^2}{\qz(z)}}} \displaybreak[0] \\
		& \stackrel{(a)}{=} \sqrt{\frac{2}{\chi{\bracknorm{\rhovector}}}} \rho_k \avgD{P_k}{\pz}, \label{eq:98}
	\end{align}
	where $(a)$ follows from the definition of $\chi{\bracknorm{\rhovector}}$.
	
	Using standard techniques, we lower bound $\log M_kL_k$, for $k \in \calK$, by
	\begin{align}
		\log M_kL_k & = \avgH{W_k S_k} \\
		& \geq \avgI{W_k S_k; \bzn} \\
		& \stackrel{(a)}{=} \avgI{\bxn_k;\bzn} \\
		& = \avgI{\bxn {\bracksq{\calK}}; \bzn} -\avgI{\bxn\bracksq{\calK \setminus \brackcurl{k} };\bzn|\bxn_k}, \label{eq:99}
	\end{align}
	where $(a)$ follows from the fact that $\bxn_k$ is a function of $W_k$ and $S_k$.
		Defining $d_6 \eqdef 2^K \max_{i \in \calK}\avgD{Q_{i}}{\qz} $, we then lower bound the first term in~\eqref{eq:99} by
		\begin{align}
		\avgI{\bxn {\bracksq{\calK}}; \bzn} & = \sum_{\xn{\bracksq{\calK}}} \sum_{\zn} \bracknorm{ \prod_{i \in \calK} \Pi^n_{X_i} \bracknorm{\xn_{i}} } W_{Z|X{\bracksq{\calK}}}^{\pn} \bracknorm{\zn|\xn{\bracksq{\calK}}} \log \frac{W_{Z|X{\bracksq{\calK}}}^{\pn} \bracknorm{\zn|\xn{\bracksq{\calK}}}}{\qhatn\bracknorm{\zn}} \\
		& = \sum_{\xn{\bracksq{\calK}}} \sum_{\zn} \bracknorm{ \prod_{i \in \calK} \Pi^n_{X_i} \bracknorm{\xn_{i}} } W_{Z|X{\bracksq{\calK}}}^{\pn} \bracknorm{\zn|\xn{\bracksq{\calK}}} \log \frac{W_{Z|X{\bracksq{\calK}}}^{\pn} \bracknorm{\zn|\xn{\bracksq{\calK}}}}{\qzn\bracknorm{\zn}} - \delta_n \\
		& = \sum_{j=1}^n \sum_{x_{(j)}{\bracksq{\calK}}} \sum_{z} \bracknorm{\prod_{i \in \calK} \Pi_{X_{ij}} (x_{ij}) } W_{Z|X{\bracksq{\calK}}}(z|x_{(j)}{\bracksq{\calK}}) \log \frac{W_{Z|X{\bracksq{\calK}}}(z|x_{(j)}{\bracksq{\calK}}) }{\qz(z) } - \delta_n    \\
		& = \sum_{j=1}^n \sum_{\calT \subseteq \calK} \bracknorm{\prod_{i \in \calT} \mun{ij} } \bracknorm{ \prod_{{i \in \calT^c}} \bracknorm{1 - \mun{ij}} } \avgD{Q_{\calT}}{\qz} - \delta_n  \\
		& \geq \sum_{j=1}^n \sum_{i \in \calK} \mun{ij}\bracknorm{ \prod_{\substack{i' \in \calK \setminus \brackcurl{i} }} \bracknorm{1 - \mun{i'j}} } \avgD{Q_{i}}{\qz}  - \delta_n \\
		& \stackrel{(a)}{\geq} \sum_{j=1}^n \sum_{i \in \calK} \mun{ij} \avgD{Q_i}{\qz} - d_6 \mun{\max} \sum_{i \in \calK} \sum_{j=1}^n  \mun{ij}  - \delta_n,  \label{eq:100}
		\end{align}
	where $(a)$ follows from the steps used to obtain~\eqref{eq:89} from~\eqref{eq:89_1}. 
	Note that, by definition, we have
	\begin{align}
		\sum_{\calT \subseteq \calK \setminus \brackcurl{k} } \bracknorm{ \prod_{i \in \calK \setminus \brackcurl{k} } \Pi_{X_{ij}}\bracknorm{x_{\calT,i}} } W_{Z_j|X_{(j)}{\bracksq{\calK \setminus \brackcurl{k}}}X_{kj}} \bracknorm{z|x_\calT{\bracksq{\calK \setminus \brackcurl{k}}}x} = W_{Z_j|X_{kj}}(z|x). \label{eq:100_e}
	\end{align}
	We upper bound the second term in~\eqref{eq:99} by
	\begin{align}
		& \avgI{\bxn\bracksq{\calK \setminus \brackcurl{k} };\bzn|\bxn_k} \nex
		&\dupspace = \avgH{\bzn | \bxn_k} - \avgH{\bzn|\bxn{\bracksq{\calK}}} \\
		& \dupspace\stackrel{(a)}{\leq} \sum_{j=1}^n  \bracknorm{ \avgH{Z_j | X_{kj}} - \avgH{Z_j |X_{(j)}{\bracksq{\calK}} } } \displaybreak[0] \\
		& \dupspace = \sum_{j=1}^n \avgI{X_{(j)}{\bracksq{\calK \setminus \brackcurl{k} }};Z_j | X_{kj} }  \\
		& \dupspace = \sum_{j=1}^n \sum_{\calT \subseteq \calK} \bracknorm{ \prod_{i \in \calK} \Pi_{X_{ij}}\bracknorm{x_{\calT,i}} } \avgD{Q_{\calT}}{W_{Z_j|X_{kj}=x_{\calT,k}}}  \\
		& \dupspace = \sum_{j=1}^n\sum_{\calT \subseteq \calK} \bracknorm{ \prod_{i \in \calK} \Pi_{X_{ij}}\bracknorm{x_{\calT,i}} } \avgD{Q_{\calT}}{\qz} \nex
		& \dupspace \dupspace - \sum_{j=1}^n \sum_{\calT \subseteq \calK} \bracknorm{ \prod_{i \in \calK} \Pi_{X_{ij}}\bracknorm{x_{\calT,i}} } \sum_z Q_{\calT}(z) \log \frac{W_{Z_j|X_{kj}}\bracknorm{z|x_{\calT,k}}}{\qz(z)} \\
		& \dupspace = \sum_{j=1}^n\sum_{\calT \subseteq \calK} \bracknorm{ \prod_{i \in \calK} \Pi_{X_{ij}}\bracknorm{x_{\calT,i}} } \avgD{Q_{\calT}}{\qz} \nex
		& \dupspace \dupspace - \sum_{j=1}^n \sum_x \Pi_{X_{kj}}(x) \sum_z  \sum_{\calT \subseteq \calK \setminus \brackcurl{k} } \bracknorm{ \prod_{i \in \calK \setminus \brackcurl{k} } \Pi_{X_{ij}}\bracknorm{x_{\calT,i}} } \nex
		& \dupspace \dupspace \dupspace \times W_{Z_j|X_{(j)}{\bracksq{\calK \setminus \brackcurl{k}}}X_{kj}} \bracknorm{z|x_\calT{\bracksq{\calK \setminus \brackcurl{k}}}x} \log \frac{W_{Z_j|X_{kj}}\bracknorm{z|x}}{\qz(z)} \\
		& \dupspace \stackrel{(b)}{=} \sum_{j=1}^n\sum_{\calT \subseteq \calK} \bracknorm{ \prod_{i \in \calK} \Pi_{X_{ij}}\bracknorm{x_{\calT,i}} } \avgD{Q_{\calT}}{\qz} - \sum_{j=1}^n \sum_{x} \Pi_{X_{kj}}(x) \avgD{W_{Z_j|X_{kj}=x}}{\qz} \\
		& \dupspace \leq \sum_{j=1}^n\sum_{\calT \subseteq \calK} \bracknorm{ \prod_{i \in \calK} \Pi_{X_{ij}}\bracknorm{x_{\calT,i}} } \avgD{Q_{\calT}}{\qz} - \sum_{j=1}^n \mun{kj} \avgD{W_{Z_j|X_{kj}=1}}{\qz} , \label{eq:100_a}
	\end{align}
	where $(a)$ follows from the fact that conditioning reduces entropy and the memoryless property of the channel, and $(b)$ follows from~\eqref{eq:100_e}.  
	Defining $d_7 \eqdef 2^K \max_{\substack{\calT \subseteq \calK:  \card{\calT} > 1}} \avgD{Q_{\calT}}{\qz}$, we upper bound the first term in~\eqref{eq:100_a} by
	\begin{align}
		& \sum_{j=1}^n\sum_{\calT \subseteq \calK} \bracknorm{ \prod_{i \in \calK} \Pi_{X_{ij}}\bracknorm{x_{\calT,i}} } \avgD{Q_{\calT}}{\qz} \nex
		& = \sum_{j=1}^n \sum_{\calT \subseteq \calK} \bracknorm{\prod_{i \in \calT}\mun{ij}} \bracknorm{\prod_{i \in \calT^c} \bracknorm{1 - \mun{ij}} } \avgD{Q_{\calT}}{\qz} \\
		& \stackrel{(a)}{ \leq} \sum_{j=1}^n \sum_{\calT \subseteq \calK} \bracknorm{\prod_{i \in \calT}\mun{ij}} \avgD{Q_{\calT}}{\qz}  \\
		& = \sum_{j=1}^n \sum_{i \in \calK} \mun{ij} \avgD{Q_i}{\qz} + \sum_{j=1}^n \sum_{\substack{\calT \subseteq \calK: \\ \card{\calT} > 1}} \bracknorm{\prod_{i \in \calT}\mun{ij}} \avgD{Q_{\calT}}{\qz} \\
		& \leq \sum_{j=1}^n \sum_{i \in \calK} \mun{ij} \avgD{Q_i}{\qz} + d_7 \mun{\max} \sum_{i \in \calK}  \sum_{j=1}^n \mun{ij} , \label{eq:100_b}
	\end{align}
	where $(a)$ follows  from the fact that $\bracknorm{\prod_{i \in \calT^c} \bracknorm{1 - \mun{ij}} } \leq 1$ for any $\calT \subseteq \calK$. 
	Then, {from Corollary~\ref{cor:corr-representation}}, we write
	\begin{align}
		W_{Z_j|X_{kj}}(z|1) &= Q_k(z) + \sum_{\substack{\calT \subseteq \calK \setminus \brackcurl{k}: \\ \calT \neq \emptyset}} \bracknorm{\prod_{i \in \calT} \mun{ij}} \bracknorm{ \sum_{\calU \subseteq \calT } \bracknorm{-1}^{\card{\calT} - \card{\calU}}   Q_{\calU \cup \brackcurl{k}}(z)}. \label{eq:100_3}
	\end{align}
Defining $d_8 \eqdef 2^K \max\limits_{\calT \subseteq \calK \setminus \brackcurl{k}: \calT \neq \emptyset} \abs{ \sum_z  \sum\limits_{\calU \subseteq \calT } \bracknorm{-1}^{\card{\calT} - \card{\calU}}   Q_{\calU \cup \brackcurl{k}}(z) \log \frac{Q_k(z)}{\qz(z)}} $ and using~\eqref{eq:100_3}, we bound the second KL divergence term in~\eqref{eq:100_a} by 
	\begin{align}
		& \avgD{W_{Z_j|X_{kj}=1}}{\qz} \nex
		& \dupspace = \avgD{W_{Z_j|X_{kj}=1}}{Q_k} + \sum_z W_{Z_j|X_{kj}}(z|1) \log \frac{Q_k(z)}{\qz(z)} \\
		& \dupspace \geq \sum_z \bracknorm{Q_k(z) + \sum_{\substack{\calT \subseteq \calK \setminus \brackcurl{k}: \\ \calT \neq \emptyset}} \bracknorm{\prod_{i \in \calT} \mun{ij}} \bracknorm{ \sum_{\calU \subseteq \calT } \bracknorm{-1}^{\card{\calT} - \card{\calU}}   Q_{\calU \cup \brackcurl{k}}(z)}} \log \frac{Q_k(z)}{\qz(z)} \\
		& \dupspace { \geq\avgD{Q_k}{\qz} - \sum_{\substack{\calT \subseteq \calK \setminus \brackcurl{k}: \\ \calT \neq \emptyset}}  \bracknorm{\prod_{i \in \calT} \mun{ij}} \abs{ \sum_z \bracknorm{ \sum_{\calU \subseteq \calT } \bracknorm{-1}^{\card{\calT} - \card{\calU}}   Q_{\calU \cup \brackcurl{k}}(z)} \log \frac{Q_k(z)}{\qz(z)}}} \\ 
		& \dupspace \geq \avgD{Q_k}{\qz} - d_8 \mun{\max}.  \label{eq:100_d}
	\end{align}
	Defining $d_9 \eqdef d_7 + d_8$ and combining~\eqref{eq:100_a},~\eqref{eq:100_b}, and~\eqref{eq:100_d}, we obtain
	\begin{align}
		 \!\!\!\!\!\avgI{\bxn\bracksq{\calK \setminus \brackcurl{k} };\bzn|\bxn_k} & \leq \sum_{j=1}^n \sum_{i \in \calK \setminus \brackcurl{k} } \mun{ij} \avgD{Q_i}{\qz} + d_7 \mun{\max} \sum_{i \in \calK} \sum_{j=1}^n  \mun{ij} + d_8 \mun{\max} \sum_{j=1}^n \mun{kj} \\
		& \leq  \sum_{j=1}^n \sum_{i \in \calK \setminus \brackcurl{k} } \mun{ij} \avgD{Q_i}{\qz} + d_{9}\mun{\max} \sum_{i \in \calK} \sum_{j=1}^n  \mun{ij}  . \label{eq:102_1}
	\end{align}
	Defining $d_{10} \eqdef d_6 + d_9$ and combining~\eqref{eq:100} and~\eqref{eq:102_1}, we bound~\eqref{eq:99} by
	\begin{align}
		\log M_kL_k & \! \geq \!\!\bracknorm{ \sum_{j=1}^n \mun{kj}\!\! }\!\! \avgD{Q_k}{\qz} \! - d_{10} \mun{\max}  \sum_{i \in \calK}  \sum_{j=1}^n\smash{\mun{ij}}\! -\! \delta_n. \label{eq:103}
	\end{align} 
		{Normalizing $\log M_kL_k$, where $k \in \calK$, by $\sqrt{n \delta_n}$, we obtain
	\begin{align}
		\frac{\log M_kL_k}{\sqrt{n \delta_n}} & \geq \frac{\bracknorm{ \sum_{j=1}^n \mun{kj} }\avgD{Q_k}{\qz} - d_{10} \mun{\max}  \sum_{i \in \calK}  \sum_{j=1}^n\mun{ij}   - \delta_n}{\sqrt{n \delta_n}} \\
		& = \frac{\bracknorm{\sum_{j=1}^n \mun{kj} } \bracknorm{ \avgD{Q_k}{\qz} - d_{10} \mun{\max} \sum_{i \in \calK}  \frac{ \sum_{j=1}^n\mun{ij}}{\sum_{j=1}^n \mun{kj} } - \frac{\delta_n}{\sum_{j=1}^n \mun{kj} }  }}{\sqrt{n \delta_n}}. \label{eq:104}
	\end{align}}
	Consider a sequence of codes for which~\eqref{eq:98} holds with equality for all $k \in \calK$. Proposition~\ref{prop:achievability} confirms the existence of such schemes. As a result, for an arbitrary $\xi>0$, we have
	\begin{align}
		{\liminf_{\substack{n \to \infty \\ n \in \calN^\dagger}}} \frac{\log M_k}{\sqrt{n \delta_n}} & \geq \bracknorm{1-\xi} \sqrt{\frac{2}{\chi{\bracknorm{\rhovector}}}} \rho_k \avgD{P_k}{\pz}. \label{eq:108}
	\end{align}
	Then, for that sequence of codes, using~\eqref{eq:86} and~\eqref{eq:92}, we obtain
	\begin{align} 
		\!\!\!\!\!\!\!\!\!{\liminf_{\substack{n \to \infty \\ n \in \calN^\dagger}}} \frac{\!\bracknorm{ \sum_{j=1}^n \mun{kj} }\!\avgD{P_k}{\pz}\! +\! d_4 \mun{\max}\!  \sum_{i \in \calK} \!\!\bracknorm{\! \sum_{j=1}^n  \!\mun{ij}\! } \!\!  + \!\Hb{\epsilon_n}}{\bracknorm{1-\epsilon_n} \sqrt{n \delta_n}} & \!\geq\! \bracknorm{1-\xi} \!\sqrt{\!\frac{2}{\chi{\bracknorm{\rhovector}}}} \rho_k \avgD{P_k}{\pz},  \\
		{\liminf_{\substack{n \to \infty \\ n \in \calN^\dagger}}} \frac{\bracknorm{ \sum_{j=1}^n \mun{kj} }\bracknorm{ \avgD{P_k}{\pz} + d_4 \mun{\max}  \sum_{i \in \calK} \frac{\sum_{j=1}^n  \mun{ij} }{\sum_{j=1}^n \mun{kj}} }}{\bracknorm{1-\epsilon_n} \sqrt{n \delta_n}} & \!\geq\! \bracknorm{1-\xi} \!\sqrt{\!\frac{2}{\chi{\bracknorm{\rhovector}}}} \rho_k \avgD{P_k}{\pz}, \\
		{\liminf_{\substack{n \to \infty \\ n \in \calN^\dagger}}} \frac{ \bracknorm{ \sum_{j=1}^n \mun{kj}} \avgD{P_k}{\pz} }{\sqrt{n \delta_n}} & \!\geq\! \bracknorm{1-\xi} \!\sqrt{\!\frac{2}{\chi{\bracknorm{\rhovector}}}} \rho_k \avgD{P_k}{\pz}, \\
		{\liminf_{\substack{n \to \infty \\ n \in \calN^\dagger}}}  \frac{ { \sum_{j=1}^n \mun{kj}} }{\sqrt{n \delta_n}} & \!\geq\! \bracknorm{1-\xi} \!\sqrt{\!\frac{2}{\chi{\bracknorm{\rhovector}}}} \rho_k. \label{eq:109}
	\end{align}	
	{
	However, since $\limsup_{n \to \infty} a_n \geq \liminf_{n \to \infty} a_n$ for any sequence $\brackcurl{a_n}$, we write
	\begin{align}
		\limsup_{\substack{n \to \infty \\ n \in \calN^\dagger}}  \frac{ { \sum_{j=1}^n \mun{kj}} }{\sqrt{n \delta_n}} & \!\geq\! \bracknorm{1-\xi} \!\sqrt{\!\frac{2}{\chi{\bracknorm{\rhovector}}}} \rho_k. \label{eq:109_1} 
	\end{align}
	Combining~\eqref{eq:104} and~\eqref{eq:109_1}, we obtain  
	\begin{align}
		{\limsup_{\substack{n \to \infty \\ n \in \calN^\dagger}}} \frac{\log M_kL_k}{\sqrt{n \delta_n}} & \geq  \bracknorm{1-\xi} \sqrt{\frac{2}{{\chi{\bracknorm{\rhovector}}}}} \rho_k \avgD{Q_k}{\qz}, \label{eq:110}
	\end{align}
	for an arbitrary $\xi > 0$, where $(a)$ follows from the fact that $\lim_{\substack{n \to \infty \\ n \in \calN^\dagger}} \sum_{j=1}^n \mun{kj} = \infty$. 
	Letting $\xi \downarrow 0$ in~\eqref{eq:110}, we obtain~\eqref{eq:con2}. 
	}
\end{proof}
{
Note that for any sequence $\brackcurl{a_n}_{n \in \mathbb{N}^*}$ and any infinite set $\calN \subseteq \mathbb{N}^*$, we have, by definition, 
\begin{align}
	\liminf_{n \to \infty} a_n \leq \liminf_{\substack{n \to \infty \\ n \in \calN}} a_n \leq \limsup_{\substack{n \to \infty \\ n \in \calN}} a_n \leq \limsup_{n \to \infty} a_n. \label{eq:lim1}
\end{align}
From Proposition~\ref{prop:converse} and equation~\eqref{eq:lim1}, we conclude that the covert capacity region is contained in the region defined by
\begin{align}
\bigcup_{\{\rho_k\}_{k\in\calK}\in[0,1]^K:\sum_{k \in \calK} \rho_k = 1}\left\{\{r_k\}_{k\in\calK}:\forall k\in\calK,\quad r_k\leq \sqrt{\frac{2}{\chi\bracknorm{\rhovector}}} \rho_k \avgD{P_k}{\pz}\right\}, \label{eq:region2}
\end{align}
and that, any achievable covert throughput tuple $r\bracksq{\calK}$ characterized by a specific $\rhovector$  and lying on the boundary of the region defined in~\eqref{eq:region2} is associated to an achievable key throughput of at least ${ \sqrt{ \frac{2}{\chi(\rhovector) }} \rho_k \left[\avgD{Q_k}{\qz} - \avgD{P_k}{\pz}\right]^+}$ for each $k \in \calK$. 
}

\section{Conclusion} \label{sec:conclusion}
{We conclude with a discussion of extensions of our results and related problems of interest.
  
We confirm that our proof techniques and results extend to multiple \emph{non-innocent} symbols as in~[5, Theorem 2] and [6, Corollary 3]. However, our relatively concise notation for the subscript of the one-shot output distributions $P$ and $Q$ heavily relies on the fact that users who transmit can only transmit a uniquely defined information symbol. This allows us to index the distributions with the set of transmitting users. In the presence of multiple information symbols, this ease of notation is lost, and one needs to keep track of not only which user is transmitting but also what symbol is transmitted. That being said, the key conceptual results that lead to our characterization hold as in~\cite{Wang2016b,bloch2016covert}. 
For more details, the interested reader can refer to Appendix~\ref{sec:multipleinput}. 

Our resolvability analysis is not directly applicable to \ac{AWGN} channels since we use $\nu_{\min} \eqdef \min_z \qz(z) $ in the denominator of~\eqref{eq:63} and~\eqref{eq:37}, which is zero for \ac{AWGN} channels. However, our achievability results can be extended to \ac{AWGN} channels by using resolvability exponents as in~\cite{hayashi2006general, endo2014reliability, parizi2017exact} to obtain a bound for the KL divergence $\avgD{\smash{\qhatn}}{\qzn}$ that does not rely on the discrete or continuous nature of the channel output alphabet. 
The converse argument can be developed by extending the approach of~\cite{Wang2016b} to deal with multiple users, and one expects the covert capacity region to be
\begin{align*}
  \bigcup_{\{\rho_k\}_{k\in\calK}:\sum_{k}\rho_k=1}\left\{\{r_{k}\}_{k\in\calK}:r_k\leq\rho_k\right\}.
\end{align*}
For more details, the reader can refer to Appendix~\ref{sec:AWGN}. 

A final problem of interest is the characterization of the covert capacity region of a $K$-user \ac{MAC} in which the transmitters share a \emph{common} key. Unlike the situation addressed here, the common key scenario captures the ability of users to \emph{coordinate} their covert transmissions. One can approach the problem by following cooperative channel resolvability techniques studied in~\cite{helal2018multiple,Helhal2018a}.
}

\appendices

\section{Alternative representation of $\qal$ in Eq.~\eqref{eq:13}} \label{sec:altrep}
\begin{lemma}
  \label{lem:altrepresentation}
	For any set $\calS \subseteq \calK $, define $G_\calS(z) \eqdef \sum_{\calT \subseteq \calS} \bracknorm{-1}^{\card{\calS} - \card{\calT}}  Q_{\calT}(z)$. Then,
	\begin{align}
		\qal(z) = \qz(z) + \sum_{\substack{\calS \subseteq \calK: \\ \calS \neq \emptyset}} \bracknorm{\prod_{k \in \calS}  \rho_k \alpha_n} G_{\calS}(z). \label{eq:120}
	\end{align}
\end{lemma}
\begin{proof}
First, we prove the following statement by induction. For any set $\calS$ and $\beta_k \in [0,1]$ for $k \in \calK$,
\begin{align}
	\prod_{k \in \calS} \bracknorm{1 -  \beta_k} & = 1 + \sum_{\substack{\calT \subseteq \calS: \\ \calT \neq \emptyset}} \bracknorm{-1}^{\card{\calT}} \bracknorm{\prod_{k \in \calT}  \beta_k}. \label{eq:111_1}
\end{align}
It is straightforward to show that~\eqref{eq:111_1} is true for $\calS = \brackcurl{1}$. When $\calS = \brackcurl{1,2}$, we have 
\begin{align}
	\prod_{k \in \brackcurl{1,2}} \bracknorm{1 -  \beta_k} & =  1 + \sum_{\substack{\calT \subseteq \brackcurl{1,2}: \\\calT \neq \emptyset}} \bracknorm{-1}^{\card{\calT}} \bracknorm{ \prod_{k \in \calT}  \beta_k } \\
	& = 1 - \beta_1 - \beta_2 + \beta_1 \beta_2. \label{eq:111_2}
\end{align}
We assume that~\eqref{eq:111_1} is true for the set $\calS \eqdef \intseq{1}{K-1}$, where $K \in \mathbb{N}^*$. 
Then, for the set $\calS' \eqdef \calS \cup \brackcurl{K}$, we have
\begin{align}
	\prod_{k \in \calS'} \bracknorm{1 -  \beta_k} & = \bracknorm{1 -  {\beta_K}} \prod_{k \in \calS} \bracknorm{1 -  \beta_k} \displaybreak[0]\\ 
	& = \bracknorm{1 - {\beta_K}} \bracknorm{1 + \sum_{\substack{\calT \subseteq \calS: \\ \calT \neq \emptyset}} \bracknorm{-1}^{\card{\calT}} \bracknorm{\prod_{k \in \calT}  \beta_k}} \displaybreak[0]\\
	& = 1 - {\beta_K} + \sum_{\substack{\calT \subseteq \calS: \\ \calT \neq \emptyset}} \bracknorm{-1}^{\card{\calT}} \bracknorm{\prod_{k \in \calT}  \beta_k}  - {\beta_K} \bracknorm{ \sum_{\substack{\calT \subseteq \calS: \\ \calT \neq \emptyset}} \bracknorm{-1}^{\card{\calT}} \bracknorm{\prod_{k \in \calT}  \beta_k}}  \displaybreak[0]\\
	& \stackrel{(a)}{=} 1 + \sum_{\calT = \brackcurl{K} } \bracknorm{-1}^{\card{\calT}} \bracknorm{\prod_{k \in \calT}  \beta_k } + \sum_{\substack{\calT \subseteq \calS: \\ \card{\calT} = 1 }} \bracknorm{-1}^{\card{\calT}} \bracknorm{\prod_{k \in \calT}  \beta_k} \nex
	& \dupspace + \sum_{\substack{\calT \subseteq \calS : \\ \card{\calT} > 1 }} \bracknorm{-1}^{\card{\calT}} \bracknorm{\prod_{k \in \calT}  \beta_k}  + \sum_{ \substack{\calT \subseteq \calS' : \\ \card{\calT}>1 ,~ K \in \calT }} \bracknorm{-1}^{\card{\calT}} \bracknorm{\prod_{k \in \calT}  \beta_k}   \displaybreak[0]\\
	& \stackrel{(b)}{=} 1 + \sum_{\substack{\calT \subseteq \calS': \\ \card{\calT} = 1 }} \bracknorm{-1}^{\card{\calT}} \bracknorm{\prod_{k \in \calT}  \beta_k} + \sum_{\substack{\calT \subseteq \calS': \\ \card{\calT} > 1,~ K \not \in \calT }} \bracknorm{-1}^{\card{\calT}} \bracknorm{\prod_{k \in \calT}  \beta_k} \nex
	& \dupspace + \sum_{ \substack{\calT \subseteq \calS': \\ \card{\calT}>1,~ K \in \calT }} \bracknorm{-1}^{\card{\calT}} \bracknorm{\prod_{k \in \calT}  \beta_k}  \displaybreak[0]\\ 
	& = 1 + \sum_{\substack{\calT \subseteq \calS' : \\ \card{\calT} = 1 }} \bracknorm{-1}^{\card{\calT}} \bracknorm{\prod_{k \in \calT}  \beta_k} + \sum_{\substack{\calT \subseteq \calS' : \\ \card{\calT} > 1 }} \bracknorm{-1}^{\card{\calT}} \bracknorm{\prod_{k \in \calT}  \beta_k} \displaybreak[0]\\ 
	& = 1 + \sum_{\substack{\calT \subseteq \calS': \\ \calT \neq \emptyset}} \bracknorm{-1}^{\card{\calT}} \bracknorm{\prod_{k \in \calT}  \beta_k}, \label{eq:111_3}
\end{align}
where $(a)$ follows from the fact that $(- {\beta_K} )\bracknorm{ \sum_{\substack{\calT \subseteq \calS: \\ \calT \neq \emptyset}} \bracknorm{-1}^{\card{\calT}} \bracknorm{\prod_{k \in \calT}  \beta_k}} = \sum_{ \substack{\calT \subseteq \calS' : \\ \card{\calT}>1,~ K \in \calT }} \bracknorm{-1}^{\card{\calT}} \allowbreak \bracknorm{\prod_{k \in \calT}  \beta_k}$, and $(b)$  follows from the fact that $\sum_{\substack{\calT \subseteq \calS : \\ \card{\calT} > 1 }} \bracknorm{-1}^{\card{\calT}} \bracknorm{\prod_{k \in \calT}  \beta_k} = \sum_{ \substack{\calT \subseteq \calS': \\ \card{\calT}>1 ,~ K \not \in \calT }} \bracknorm{-1}^{\card{\calT}} \allowbreak \bracknorm{\prod_{k \in \calT}  \beta_k}$. 
From the definition of $\qal$ {in~\eqref{eq:10}}, we have
\begin{align}
	\qal(z) & = \sum_{x{\bracksq{\calK}}} \bracknorm{\prod_{k \in \calK} \Pi_{X_k}({x_k}) } W_{Z|X{\bracksq{\calK}}}(z|x{\bracksq{\calK}}) \\
	& = \sum_{\calT \subseteq \calK} \bracknorm{\prod_{k \in \calT}  \rho_k \alpha_n} \bracknorm{\prod_{k \in \calT^c} \bracknorm{ 1 -  \rho_k \alpha_n}} Q_{\calT}(z) \\
	& = \bracknorm{\prod_{k \in \calK} \bracknorm{1 -  \rho_k \alpha_n} } \qz(z) + \sum_{\substack{\calT \subseteq \calK: \\ \calT \neq \emptyset}} \bracknorm{\prod_{k \in \calT}  \rho_k \alpha_n} \bracknorm{\prod_{k \in \calT^c} \bracknorm{ 1 -  \rho_k \alpha_n}} Q_{\calT}(z) \\
	& \stackrel{(a)}{=} \bracknorm{\prod_{k \in \calK} \bracknorm{1 -  \rho_k \alpha_n} } \qz(z) + \sum_{\substack{\calT \subseteq \calK: \\ \calT \neq \emptyset}} \bracknorm{\prod_{k \in \calT}  \rho_k \alpha_n} \bracknorm{1 + \sum_{\substack{\calU \subseteq \calT^c : \\ \calU \neq  \emptyset}} \bracknorm{-1}^{\card{\calU}} \bracknorm{\prod_{k \in \calU}  \rho_k \alpha_n}} Q_{\calT}(z) \\
	& = \bracknorm{\prod_{k \in \calK} \bracknorm{1 -  \rho_k \alpha_n} } \qz(z) + \sum_{\substack{\calT \subseteq \calK: \\ \calT \neq \emptyset}} \bracknorm{\prod_{k \in \calT}  \rho_k \alpha_n} \bracknorm{\sum_{\calU \subseteq \calT^c} \bracknorm{-1}^{\card{\calU}} \bracknorm{\prod_{k \in \calU}  \rho_k \alpha_n}} Q_{\calT}(z), \label{eq:111}
\end{align}
where $(a)$ follows from~\eqref{eq:111_1}. 
Since $\calT$ and $\calU$ are disjoint sets, it follows from~\eqref{eq:111} that
\begin{align}
	\qal(z) & = \bracknorm{\prod_{k \in \calK} \bracknorm{1 -  \rho_k \alpha_n} } \qz(z) + \sum_{\substack{\calT \subseteq \calK: \\ \calT \neq \emptyset}} \bracknorm{\prod_{k \in \calT}  \rho_k \alpha_n}\bracknorm{\sum_{{\substack{\calS \subseteq \calK: \\ \calT \subseteq \calS}}} \bracknorm{-1}^{\card{\calS} - \card{\calT}} \bracknorm{\prod_{k \in \bracknorm{ \calS \setminus \calT}}  \rho_k \alpha_n} } Q_{\calT}(z)  \displaybreak[0]\\
	& \stackrel{(a)}{=} \qz(z) + \sum_{\substack{\calS \subseteq \calK : \\ \calS \neq \emptyset}} \bracknorm{-1}^{\card{\calS}} \bracknorm{\prod_{k \in \calS}  \rho_k \alpha_n}  \qz(z) + \sum_{{\substack{\calT \subseteq \calK: \\ \calT \neq \emptyset}}} \bracknorm{\sum_{\substack{\calS \subseteq \calK : \\  \calT \subseteq \calS}} \bracknorm{-1}^{\card{\calS} - \card{\calT}} \bracknorm{\prod_{k \in \calS}  \rho_k \alpha_n} } Q_{\calT}(z)  \displaybreak[0]\\
	& = \qz(z) + \sum_{\substack{\calS \subseteq \calK: \\ \calS \neq \emptyset}} \bracknorm{-1}^{\card{\calS}}  \bracknorm{\prod_{k \in \calS}  \rho_k \alpha_n}  \qz(z) + \sum_{\substack{\calS \subseteq \calK: \\ \calS \neq \emptyset}} \bracknorm{\prod_{k \in \calS}  \rho_k \alpha_n} \bracknorm{ \sum_{\substack{\calT \subseteq \calS: \\ \calT \neq \emptyset}} \bracknorm{-1}^{\card{\calS} - \card{\calT}}   Q_{\calT}(z)}  \displaybreak[0]\\ 
	& = \qz(z) + \sum_{\substack{\calS \subseteq \calK: \\ \calS \neq \emptyset}} \bracknorm{\prod_{k \in \calS}  \rho_k \alpha_n} \bracknorm{ \sum_{\calT \subseteq \calS } \bracknorm{-1}^{\card{\calS} - \card{\calT}}   Q_{\calT}(z)} , \label{eq:112}
\end{align}
where $(a)$ follows from~\eqref{eq:111_1}. 
Defining $G_\calS(z) \eqdef \sum_{\calT \subseteq \calS} \bracknorm{-1}^{\card{\calS} - \card{\calT}}  Q_{\calT}(z)$, we obtain
	\begin{align}
		\qal(z) = \qz(z) + \sum_{\substack{\calS \subseteq \calK: \\ \calS \neq \emptyset}} \bracknorm{\prod_{k \in \calS}  \rho_k \alpha_n} G_{\calS}(z). \label{eq:120}
	\end{align}
\end{proof}
\begin{corollary}
  \label{cor:corr-representation}
  For any set $\calS \subseteq \calK $, define $G_\calS(z) \eqdef \sum_{\calT \subseteq \calS} \bracknorm{-1}^{\card{\calS} - \card{\calT}}  Q_{\calT}(z)$. Then,
  \begin{align}
    W_{Z|X_k}(z|1) = Q_k(z) + \sum_{\substack{\calS \subseteq \calK \setminus \brackcurl{k}: \\ \calS \neq \emptyset}} \bracknorm{\prod_{i \in \calS} \rho_i \alpha_n} \bracknorm{ \sum_{\calT \subseteq \calS } \bracknorm{-1}^{\card{\calS} - \card{\calT}}   Q_{\calT \cup \brackcurl{k}}(z)}. \label{eq:119}
  \end{align}
\end{corollary}

\section{Proof of Lemma~\ref{lem:covertprocess}}
\label{sec:covertprocess}
From the definition of $\avgD{Q_{\alpha_n}}{\qz}$, we have 
\begin{align}
  \avgD{Q_{\alpha_n}}{\qz} &= \sum_z Q_{\alpha_n}(z) \log \frac{Q_{\alpha_n}(z)}{\qz(z)} \\
                           & = \sum_z \qz(z) \bracknorm{1 + \frac{\alpha_n \zeta_n(z)}{\qz(z)}} \log  \bracknorm{1 + \frac{\alpha_n \zeta_n(z)}{\qz(z)}}. \label{eq:16}
\end{align}
Since $\log \bracknorm{1+x} < x - \frac{x^2}{2} + \frac{x^3}{3}$, for $x > -1$, we upper bound~\eqref{eq:16} by
\begin{align}
		\avgD{Q_{\alpha_n}}{\qz} & {\leq} \sum_z \qz(z) \bracknorm{1 + \frac{\alpha_n \zeta_n(z)}{\qz(z)}} \bracknorm{\frac{\alpha_n \zeta_n(z)}{\qz(z)} -  \frac{\alpha_n^2\zeta_n^2(z)}{2\qz^2(z)} +  \frac{\alpha_n^3\zeta_n^3(z)}{3\qz^3(z)} } \\
		& \stackrel{(a)}{=} \sum_z \frac{\alpha_n^2}{2} \bracknorm{\frac{\zeta_n^2(z)}{\qz(z)} -  \frac{\alpha_n\zeta_n^3(z)}{3\qz^2(z)} + \frac{2\alpha_n^2\zeta_n^4(z)}{3\qz^3(z)} }, \label{eq:17}
	\end{align}
	where, $(a)$ follows from the fact that $\sum_z \zeta_n(z) = 0$ from the definition of $\zeta_n$. 
	Since $\limn \alpha_n = 0$, $\alpha_n$ is small enough for a sufficiently large $n$ and $\frac{\alpha_n \zeta_n(z)}{\qz(z)} \in \bracksq{-\frac{1}{2},0}$ for any $z \in \calZ$ if $\zeta_n(z)<0$. Then, we lower bound~\eqref{eq:16} by
	\begin{align}
		\avgD{Q_{\alpha_n}}{\qz} & \!\stackrel{(a)}{\geq} \!\sum_z \qz(z) \!\bracknorm{1 + \frac{\alpha_n \zeta_n(z)}{\qz(z)}} \!\!\bracknorm{\frac{\alpha_n \zeta_n(z)}{\qz(z)} - \frac{\alpha_n^2\zeta_n^2(z)}{2\qz^2(z)}}\!\! \nex
		& \dupspace + \!\!\!\! \sum_{z: \zeta_n(z)<0} \!\!\!\! \qz(z)\!\! \bracknorm{1 + \frac{\alpha_n \zeta_n(z)}{\qz(z)}} \!\!\bracknorm{\frac{2\alpha_n^3 \zeta_n^3(z)}{3\qz^3(z)}} \displaybreak[0]\\
		& \stackrel{(b)}{\geq} \sum_z \frac{\alpha_n^2}{2} \bracknorm{\frac{\zeta_n^2(z)}{\qz(z)} - \frac{\alpha_n \zeta_n^3(z)}{\qz^2(z)}} + \sum_{z: \zeta_n(z)<0} \frac{2\alpha_n^3\zeta_n^3(z)}{3\qz^2(z)}, \label{eq:18}
	\end{align}
	where, $(a)$ follows from the inequalities $\log \bracknorm{1+x} > x - \frac{x^2}{2}$ for $x \geq 0$ and $\log \bracknorm{1+x} > x - \frac{x^2}{2} + \frac{2x^3}{3}$ for $x \in \bracksq{-\frac{1}{2}, 0}$, and $(b)$ follows from the fact that $\sum_z \zeta_n(z) = 0$. For $n$ large enough, we loosen the bounds in~\eqref{eq:17} and~\eqref{eq:18} to obtain
	\begin{align}
		\frac{\alpha_n^2}{2} \bracknorm{1 + \sqrt{\alpha_n}} \chi_n{\bracknorm{\rhovector}} \geq \avgD{Q_{\alpha_n}}{\qz} \geq \frac{\alpha_n^2}{2} \bracknorm{1 - \sqrt{\alpha_n}} \chi_n{\bracknorm{\rhovector}}. \label{eq:19}
	\end{align}
	From the definition of $\qal$, we have
	\begin{align}
		\qal(z) = \qz(z) + \alpha_n \bracknorm{\sum_{k \in \calK} \rho_k \bracknorm{Q_k(z) - \qz(z)} } + \bigO{\alpha_n^2}. \label{eq:20}
	\end{align}
	Using the definition of $\zeta_n$ and applying the limit, we obtain
	\begin{align}
		\limn \zeta_n(z) & = \limn \frac{\qal(z) - \qz(z)}{\alpha_n}  \\
		& = \limn \bracknorm{ \sum_{k \in \calK} \rho_k \bracknorm{Q_k(z) - \qz(z)} + \bigO{\alpha_n}} \\
		& \stackrel{(a)}{=} \sum_{k \in \calK} \rho_k \bracknorm{Q_k(z) - \qz(z)} \\
		& = \zeta(z), \label{eq:21}
	\end{align}
	where $(a)$ follows from the fact that $\limn \alpha_n = 0$. 
	From~\eqref{eq:21} and the definition of $\chi_n{\bracknorm{\rhovector}}$, it follows that 
	\begin{align}
		\limn \chi_n{\bracknorm{\rhovector}} = \limn \sum_z \frac{\zeta_n^2(z)}{\qz(z)} = \sum_z \frac{\zeta^2(z)}{\qz(z)} = \chi{\bracknorm{\rhovector}}. \label{eq:22}
	\end{align} 
	Finally, for a non-empty set $\calT \subseteq \calK$, define $\lambda_{n,\calT}(z) \eqdef \frac{ W_{Z|X\bracksq{\calT}} \bracknorm{z|x_{\emptyset}\bracksq{\calT}} - \qz(z)}{\alpha_n} $. Note that $\sum_z \lambda_{n, \calT}(z) = 0$.  Then, for any non-empty set $\calT \subseteq \calK$, we have
	\begin{align}
		\avgI{X{{\bracksq{\calT}}};Z} &= \sum_{x {\bracksq{\calT}}} \sum_z \bracknorm{\prod_{k \in \calT} \Pi_{X_k}\bracknorm{x_k}} W_{Z|X{\bracksq{\calT}}} \bracknorm{z | x{{\bracksq{\calT}}}} \log \bracknorm{\frac{W_{Z|X{\bracksq{\calT}}} \bracknorm{z | x{\bracksq{\calT}}}}{\qal(z)}}  \displaybreak[0]\\
		& = \sum_{x {\bracksq{\calT}}} \sum_z \bracknorm{\prod_{k \in \calT} \Pi_{X_k}\bracknorm{x_k}} W_{Z|X{\bracksq{\calT}}} \bracknorm{z | x{\bracksq{\calT}}} \log \bracknorm{\frac{W_{Z|X{\bracksq{\calT}}} \bracknorm{z | x{\bracksq{\calT}}}}{\qz(z)}} - \avgD{\qal}{\qz} \label{eq:23}  \displaybreak[0]\\
		& \stackrel{(a)}{=} \!\!\! \sum_{\substack{\calU \subseteq \calT: \\ \card{\calU}> 1}}  \bracknorm{ \prod_{k \in \calU} \rho_k \alpha_n } \bracknorm{\prod_{k \in \calU^c} \bracknorm{1 - \rho_k \alpha_n} } \sum_z W_{Z|X{\bracksq{\calT}}} \bracknorm{z|x_{\calU}{\bracksq{\calT}}} \log \bracknorm{ \frac{W_{Z|X{\bracksq{\calT}}} \bracknorm{z | x_{\calU}{\bracksq{\calT}}}}{\qz(z)}} \nex
		& \dupspace + \sum_{k \in \calT} \rho_k \alpha_n \bracknorm{ \prod_{\substack{i \in \calT: \\ i \neq k}} \bracknorm{1 - \rho_i \alpha_n} } \sum_z W_{Z|X{\bracksq{\calT}}}\bracknorm{z| x_{\brackcurl{k}}{\bracksq{\calT}}} \log \bracknorm{\frac{W_{Z|X{\bracksq{\calT}}}\bracknorm{z| x_{\brackcurl{k}}{\bracksq{\calT}}}}{\qz(z)}} \nex
		& \dupspace + \bracknorm{\prod_{k \in \calT} \bracknorm{1 - \rho_k \alpha_n} } \sum_z W_{Z|X\bracksq{\calT}}\bracknorm{z|x_\emptyset \bracksq{\calT}} \log \bracknorm{ \frac{W_{Z|X\bracksq{\calT} } \bracknorm{z|x_\emptyset \bracksq{\calT}} }{\qz(z)}}  - \avgD{\qal}{\qz} \label{eq:23_1} \displaybreak[0] \\
		& = \sum_{k \in \calT} \rho_k \alpha_n \bracknorm{ \prod_{\substack{i \in \calT: \\ i \neq k} } \bracknorm{1 - \rho_i \alpha_n} } \sum_z W_{Z|X{\bracksq{\calT}}}\bracknorm{z| x_{\brackcurl{k}}{\bracksq{\calT}}} \log \bracknorm{\frac{W_{Z|X{\bracksq{\calT}}}\bracknorm{z| x_{\brackcurl{k}}{\bracksq{\calT}}}}{\qz(z)}} \nex
		& \dupspace + \bracknorm{\prod_{k \in \calT} \bracknorm{1 - \rho_k \alpha_n} } \sum_z  \bracknorm{\qz(z) + \alpha_n \lambda_{n,\calT}(z)} \log \bracknorm{ 1 +  \alpha_n \frac{ \lambda_{n,\calT}(z)}{\qz(z)}} \nex
		& \dupspace \dupspace - \avgD{\qal}{\qz} + \bigO{\alpha_n^2} \label{eq:23_2} \displaybreak[0]\\
		& \stackrel{(b)}{=}  \sum_{k \in \calT} \rho_k \alpha_n \sum_z W_{Z|X{\bracksq{\calT}}}\bracknorm{z| x_{\brackcurl{k}}{\bracksq{\calT}}} \log \bracknorm{\!\frac{W_{Z|X{\bracksq{\calT}}}\bracknorm{z| x_{\brackcurl{k}}{\bracksq{\calT}}}}{\qz(z)}\!}-\! \avgD{\qal}{\qz} + \! \bigO{\alpha_n^2} \displaybreak[0]\\
		& =  \sum_{k \in \calT} \rho_k \alpha_n \sum_z \bracknorm{ Q_k(z) + \bigO{\alpha_n}} \log \bracknorm{\frac{ Q_k(z) + \bigO{\alpha_n}}{\qz(z)}}- \avgD{\qal}{\qz} + \bigO{\alpha_n^2} \displaybreak[0]\\
		& =  \sum_{k \in \calT} \rho_k \alpha_n \sum_z  Q_k(z)  \log \bracknorm{ \bracknorm{ \frac{Q_k(z)}{\qz(z)} }\bracknorm{1 + \frac{\bigO{\alpha_n}}{Q_k(z)}}}- \avgD{\qal}{\qz} + \bigO{\alpha_n^2}  \\
		& = \sum_{k \in \calT} \rho_k \alpha_n \bracknorm{\avgD{Q_k}{\qz} +  \sum_z  Q_k(z)  \log \bracknorm{ 1 + \frac{\bigO{\alpha_n}}{Q_k(z)}}}- \avgD{\qal}{\qz} + \bigO{\alpha_n^2} \displaybreak[0]\\
		& = \sum_{k \in \calT} \rho_k \alpha_n \avgD{Q_k}{\qz}- \avgD{\qal}{\qz} + \bigO{\alpha_n^2}, \label{eq:24} 
	\end{align}
	where $(a)$ follows from splitting the first term in~\eqref{eq:23} into three based on the number of users sending symbol $1$, and $(b)$ follows from the fact that the second term in~\eqref{eq:23_2} can be reduced to $\bigO{\alpha_n^2}$ by expanding $\log \bracknorm{ 1 +  \alpha_n \frac{ \lambda_{n,\calT}(z)}{\qz(z)}}$ using Taylor series.
	
\section{Bernstein's inequality} \label{sec:Bernstein}
    \begin{lemma}
    	Let $\brackcurl{ U_i}_{i=1}^n$ be independent zero-mean random variables such that $\abs{U_i} \leq c$ {for a finite $c > 0$} almost surely for all $i \in \intseq{1}{n}$. Then, for any $t>0$, 
    	\begin{align}
    		\P{\sum_{i=1}^n U_i > t} \leq \exp \bracknorm{- \frac{\frac{1}{2}t^2}{\sum_{i=1}^n \E{U_i^2} + \frac{1}{3}ct}}. \label{eq:113}
    	\end{align}
    \end{lemma}
    
\section{Proof of Lemma~\ref{lem:reliability}} \label{sec:lemma_rel}
The $K$ users encode messages $W{{\bracksq{\calK}}} = m{{\bracksq{\calK}}}$ using keys $S{\bracksq{\calK}} = \ell{\bracksq{\calK}}$ into codewords $ \xn_{\calK}\bracknorm{m{\bracksq{\calK}}, \ell{\bracksq{\calK}}} $ and transmit them over a discrete memoryless \ac{MAC}. The following two events lead to a decoding error.
\begin{itemize}
	\item The transmitted codewords do not satisfy $\bracknorm{ \xn_{\calK}\bracknorm{m{\bracksq{\calK}}, \ell{\bracksq{\calK}}} ,\yn} \in \agamman{}$. 
	\item A different message vector $\widetilde{m}{{\bracksq{\calK}}} \neq m{{\bracksq{\calK}}}$ exists such that $\bracknorm{\xn_{\calK}\bracknorm{\widetilde{m}{\bracksq{\calK}},\ell{\bracksq{\calK}}},\yn} \in \agamman{}$.
\end{itemize}
Define the event
\begin{align}
	\calE_{m{\bracksq{\calK}}} \eqdef \brackcurl{\bracknorm{\bxn_{\calK}\bracknorm{m{\bracksq{\calK}},\ell{\bracksq{\calK}}},\byn} \in \agamman{}}. \label{eq:47}
\end{align}
The probability of decoding error at the legitimate receiver averaged over all random codebooks is given by
\begin{align}
	\E{P_e^n} & = \P{\widehat{W}{{\bracksq{\calK}}} \neq W{{\bracksq{\calK}}} } \\
	& = \E{ \frac{1}{\bracknorm{\prod_{k \in \calK} M_{k}}} \sum_{m{\bracksq{\calK}}} \sum_\yn W_{Y|X{\bracksq{\calK}}}^{\pn} \bracknorm{\yn|\bxn_{\calK}\bracknorm{m{\bracksq{\calK}},\ell{\bracksq{\calK}}}} \indic{\calE_{m{{\bracksq{\calK}}}}^c \cup \bigcup_{\widetilde{m}{{\bracksq{\calK}}} \neq m{{\bracksq{\calK}}} } \calE_{\widetilde{m}{{\bracksq{\calK}}}} } } \\
	& \stackrel{(a)}{\leq} \E{ \frac{1}{\bracknorm{\prod_{k \in \calK} M_{k}}} \sum_{m{\bracksq{\calK}}} \sum_\yn W_{Y|X{\bracksq{\calK}}}^{\pn} \bracknorm{\yn| \bxn_{\calK}\bracknorm{m{\bracksq{\calK}},\ell{\bracksq{\calK}}}} \indic{\calE_{m{\bracksq{\calK}}}^c  } } \nex
	& \dupspace + \E{ \frac{1}{\bracknorm{\prod_{k \in \calK} M_{k}}} \sum_{m{\bracksq{\calK}}} \sum_\yn W_{Y|X{\bracksq{\calK}}}^{\pn} \bracknorm{\yn| \bxn_{\calK}\bracknorm{m{\bracksq{\calK}},\ell{\bracksq{\calK}}} } \sum_{\widetilde{m}{{\bracksq{\calK}}} \neq m{{\bracksq{\calK}}} }\indic{ \calE_{\widetilde{m}{\bracksq{\calK}}} } }\label{eq:48},
\end{align}
where $(a)$ follows from the application of the union bound. 
We bound the first term in~\eqref{eq:48} by
\begin{align}
	& \E{ \frac{1}{\bracknorm{\prod_{k \in \calK} M_{k}}} \sum_{m{\bracksq{\calK}}} \sum_\yn W_{Y|X{\bracksq{\calK}}}^{\pn} \bracknorm{\yn| \bxn_{\calK}\bracknorm{m{\bracksq{\calK}},\ell{\bracksq{\calK}}}} \indic{\calE_{m{\bracksq{\calK}}}^c  } } \nex
	& \dupspace = \sum_{\xn{{\bracksq{\calK}}}}  \sum_\yn W_{Y|X{\bracksq{\calK}}}^{\pn} \bracknorm{\yn | \xn{\bracksq{\calK}} } \bracknorm{ \prod_{k \in \calK} \Pi_{X_k}^{\pn} \bracknorm{\xn_k}} \indic{\bracknorm{ \xn{\bracksq{\calK}} , \yn } \in {\agamman{}}^c} \\
	& \dupspace = \P[W_{Y|X{\bracksq{\calK}}}^{\pn} \bracknorm{\prod_{k\in \calK} \Pi_{X_k}^{ \pn} }]{{\agamman{}}^c}  \\
	& \dupspace \stackrel{(a)}{\leq} \sum_{\substack{\calT \subseteq \calK: \\ \calT \neq \emptyset}}\P[W_{Y|X{\bracksq{\calK}}}^{\pn} \bracknorm{\prod_{k\in \calK} \Pi_{X_k}^{ \pn} }]{{\agamman{_\calT}}^c} 	\\
	& \dupspace = \sum_{\substack{\calT \subseteq \calK: \\ \calT \neq \emptyset}}\P[W_{Y|X{\bracksq{\calK}}}^{\pn} \bracknorm{\prod_{k\in \calK} \Pi_{X_k}^{ \pn} }]{\log \frac{W_{Y|X{\bracksq{\calK}}}^{\pn}\bracknorm{\byn|\bxn{\bracksq{\calK}}}}{W^\pn_{Y|X{\bracksq{\calT^c}}}\bracknorm{\byn|\bxn{\bracksq{\calT^c}}} } < \gamma_{\calT} } \\
	& \dupspace =  \sum_{\substack{\calT \subseteq \calK: \\ \calT \neq \emptyset}}\P[W_{Y|X{\bracksq{\calK}}}^{\pn} \bracknorm{\prod_{k\in \calK} \Pi_{X_k}^{ \pn} }]{ \sum_{i=1}^n \log \frac{W_{Y|X{\bracksq{\calK}}} \bracknorm{Y|X {\bracksq{\calK}}}}{W_{Y|X{\bracksq{\calT^c}}} \bracknorm{Y|X{\bracksq{\calT^c}}} } < \gamma_{\calT} }, \label{eq:49}
\end{align}
where $(a)$ follows from the fact that ${\agamman{}}^c = \bigcup_{\substack{\calT \subseteq \calK: \\ \calT \neq \emptyset}} {\agamman{_{\calT}}}^c$ and the application of the union bound. 
We define a zero-mean\footnote{{since $\E{\log \frac{W_{Y|X{\bracksq{\calK}}} \bracknorm{Y|X{\bracksq{\calK}}}}{W_{Y|X{\bracksq{\calT^c}}} \bracknorm{Y|X{\bracksq{\calT^c}}} }} = \avgI{X{\bracksq{\calT}}; Y | X{\bracksq{\calT^c}}}$.}} random variable $U_{\calT} \eqdef \avgI{X{\bracksq{\calT}}; Y | X{\bracksq{\calT^c}}} - \log \frac{W_{Y|X{\bracksq{\calK}}} \bracknorm{Y|X{\bracksq{\calK}}}}{W_{Y|X{\bracksq{\calT^c}}} \bracknorm{Y|X{\bracksq{\calT^c}}} } $. 
Note that $\abs{U_{\calT}}$ is bounded almost surely, and
\begin{align}
	\E{U_{\calT}^2} & = \E{\log^2 \frac{W_{Y|X{\bracksq{\calK}}} \bracknorm{Y|X{\bracksq{\calK}}}}{W_{Y|X{\bracksq{\calT^c}}} \bracknorm{Y|X{\bracksq{\calT^c}}} } } - {\bracknorm{\mathbb{I}\bracknorm{X{\bracksq{\calT}}; Y | X{\bracksq{\calT^c}}}}^2.} \label{eq:50} 
\end{align}
Let us analyze the expectation term on the right hand side of~\eqref{eq:50}. 
\begin{align}
	& \E{\log^2 \frac{W_{Y|X{\bracksq{\calK}}} \bracknorm{Y|X{\bracksq{\calK}}}}{W_{Y|X{\bracksq{\calT^c}}} \bracknorm{Y|X{\bracksq{\calT^c}}} } } \nex
	& \dupspace = \sum_y \sum_{x{\bracksq{\calK}}}  \bracknorm{\prod_{k \in \calK} \Pi_{X_k} \bracknorm{x_k}} W_{Y|X{\bracksq{\calK}}}\bracknorm{y|x{\bracksq{\calK}}} \log^2 \frac{W_{Y|X{\bracksq{\calK}}} \bracknorm{y|x{\bracksq{\calK}}}}{W_{Y|X{\bracksq{\calT^c}}} \bracknorm{y|x{\bracksq{\calT^c}}} } \label{eq:51_1} \\
	& \dupspace { \stackrel{(a)}{=} \sum_y \sum_{x{\bracksq{\calK}} \neq x_\emptyset\bracksq{\calK}}  \bracknorm{\prod_{k \in \calK} \Pi_{X_k} \bracknorm{x_k}} W_{Y|X{\bracksq{\calK}}}\bracknorm{y|x{\bracksq{\calK}}} \log^2 \frac{W_{Y|X{\bracksq{\calK}}} \bracknorm{y|x{\bracksq{\calK}}}}{W_{Y|X{\bracksq{\calT^c}}} \bracknorm{y|x{\bracksq{\calT^c}}} } } \nonumber \\
	& \dupspace \dupspace {+  \sum_y \bracknorm{\prod_{k \in \calK} \bracknorm{1 - \rho_k \alpha_n} } \pz(y) \log^2 \frac{\pz(y)}{W_{Y|X{\bracksq{\calT^c}}} \bracknorm{y|x_{\emptyset}{\bracksq{\calT^c}}} } }\label{eq:51_2} \\
	& \dupspace \stackrel{(b)}{=} \sum_y \bracknorm{\prod_{k \in \calK} \bracknorm{1 - \rho_k \alpha_n} } \pz(y) \log^2 \frac{\pz(y)}{W_{Y|X{\bracksq{\calT^c}}} \bracknorm{y|x_{\emptyset}{\bracksq{\calT^c}}} }  + \bigO{\alpha_n} \\
	& \dupspace \stackrel{(c)}{=} \sum_y \pz(y) \log^2 \frac{W_{Y|X{\bracksq{\calT^c}}} \bracknorm{y|x_{\emptyset}{\bracksq{\calT^c}}} }{\pz(y)}  + \bigO{\alpha_n},  \label{eq:51_3}
\end{align}
{where $(a)$ follows from splitting the first sum on the right hand side of~\eqref{eq:51_1} into two based on whether $x\bracksq{\calK}$ equals $x_\emptyset\bracksq{\calK}$ or not, $(b)$ follows from the fact that the first term in~\eqref{eq:51_2} is on the order of $\alpha_n$ since at least one of the symbols in $x\bracksq{\calK}$ is a $1$, and $(c)$ follows from the expansion of the product term and the fact that $\log^2 \frac{\pz(y)}{W_{Y|X{\bracksq{\calT^c}}} \bracknorm{y|x_{\emptyset}{\bracksq{\calT^c}}} } = \log^2 \frac{W_{Y|X{\bracksq{\calT^c}}} \bracknorm{y|x_{\emptyset}{\bracksq{\calT^c}}} }{\pz(y)}$. 
Expanding the numerator in the $\log^2$ term in~\eqref{eq:51_3}, we obtain
\begin{align}
	W_{Y|X{\bracksq{\calT^c}}} \bracknorm{y|x_{\emptyset}{\bracksq{\calT^c}}} & = \sum_{x\bracksq{\calT}} \bracknorm{\prod_{k\in \calT} \Pi_{X_k} \bracknorm{x\bracksq{\brackcurl{k}}} } W_{Y|X\bracksq{\calT^c}X\bracksq{\calT}}\bracknorm{y|x_\emptyset\bracksq{\calT^c}x\bracksq{\calT}} \\
	& = \sum_{x\bracksq{\calT} \neq x_\emptyset\bracksq{\calT}} \bracknorm{\prod_{k\in \calT} \Pi_{X_k} \bracknorm{x\bracksq{\brackcurl{k}}} } W_{Y|X\bracksq{\calT^c}X\bracksq{\calT}}\bracknorm{y|x_\emptyset\bracksq{\calT^c}x\bracksq{\calT}} \nex
	& \dupspace + \bracknorm{\prod_{k\in \calT} \bracknorm{1 - \rho_k \alpha_n} }  \pz(y) \label{eq:51_4} \\
	& \stackrel{(a)}{=} \pz(y) + \bigO{\alpha_n}, \label{eq:51_5}
\end{align}
where $(a)$ follows from the fact that the first term in~\eqref{eq:51_4} is on the order of $\alpha_n$ since at least one of the symbols in $x\bracksq{\calT}$ is a $1$ and from the expansion of the product term. 
Combining~\eqref{eq:51_3} and~\eqref{eq:51_5}, we obtain
\begin{align}
	\E{\log^2 \frac{W_{Y|X{\bracksq{\calK}}} \bracknorm{Y|X{\bracksq{\calK}}}}{W_{Y|X{\bracksq{\calT^c}}} \bracknorm{Y|X{\bracksq{\calT^c}}} } } & = \sum_y \pz(y) \log^2 \bracknorm{1 + \frac{\bigO{\alpha_n}}{\pz(y)} }  + \bigO{\alpha_n} \\  
	& \stackrel{(a)}{=} \bigO{\alpha_n}, \label{eq:51} 
\end{align}
where $(a)$ follows from using the Taylor series of the $\log$ term.}
Let us now analyze the mutual information term on the right hand side of~\eqref{eq:50}.
\begin{align}
	\mathbb{I}\bracknorm{X{\bracksq{\calT}}; Y | X{\bracksq{\calT^c}}} & = \avgI{X{\bracksq{\calK}};Y} - \avgI{X{\bracksq{\calT^c}}; Y}  \\
	& \stackrel{(a)}{=} \sum_{k \in \calK} \rho_k \alpha_n \avgD{P_k}{\pz} - \sum_{k \in \calT^c} \rho_k \alpha_n \avgD{P_k}{\pz} + \bigO{\alpha_n^2} \\
	& = \sum_{k \in \calT} \rho_k \alpha_n \avgD{P_k}{\pz} + \bigO{\alpha_n^2}, \label{eq:52}
\end{align}
where $(a)$ follows from {Lemma~\ref{lem:covertprocess}}. Using the definition of $\gamma_{\calT}$, for an arbitrary $\mu \in \bracknorm{0,1}$, we upper bound~\eqref{eq:49} using Bernstein's inequality as follows.
\begin{align}
	& \E{ \frac{1}{\bracknorm{\prod_{k \in \calK} M_{k}}} \sum_{m{\bracksq{\calK}}} \sum_\yn W_{Y|X{\bracksq{\calK}}}^{\pn} \bracknorm{\yn| \bxn_{\calK}\bracknorm{m{\bracksq{\calK}},1{\bracksq{\calK}}}} \indic{\calE_{m{\bracksq{\calK}}}^c  } }  \nex
	& \dupspace \leq \sum_{\substack{\calT \subseteq \calK: \\ \calT \neq \emptyset}}\P[W_{Y|X{\bracksq{\calK}}}^{\pn} \bracknorm{\prod_{k\in \calK} \Pi_{X_k}^{ \pn} }]{ \sum_{i=1}^n \log \frac{W_{Y|X{\bracksq{\calK}}} \bracknorm{Y|X{\bracksq{\calK}}}}{W_{Y|X{\bracksq{\calT^c}}} \bracknorm{Y|X{\bracksq{\calT^c}}} } < \gamma_{\calT} }  \displaybreak[0]\\
	& \dupspace = \sum_{\substack{\calT \subseteq \calK: \\ \calT \neq \emptyset}} \P{ \sum_{i=1}^n U_{\calT} > \mu n \avgI{X{\bracksq{\calT}}; Y | X {\bracksq{\calT^c}}} }  \displaybreak[0]\\
	& \dupspace \stackrel{(a)}{\leq} \sum_{\substack{\calT \subseteq \calK: \\ \calT \neq \emptyset}}  \exp \bracknorm{- \frac{\frac{1}{2} \bracknorm{\mu n \avgI{X{\bracksq{\calT}}; Y | X {\bracksq{\calT^c}}}}^2}{n \bigO{\alpha_n} + \frac{1}{3}c \mu n \avgI{X{\bracksq{\calT}}; Y | X {\bracksq{\calT^c}}} }}  \displaybreak[0]\\
	& \dupspace \leq \sum_{\substack{\calT \subseteq \calK: \\ \calT \neq \emptyset}} \exp \bracknorm{-c_1 n \alpha_n} \displaybreak[0]\\
	& \dupspace \stackrel{(b)}{\leq} \exp \bracknorm{-c_2 n \alpha_n}, \label{eq:53}
\end{align}
for appropriate constants $c, c_1, c_2 > 0$, where $(a)$ follows from using Bernstein's inequality, and $(b)$ follows from the fact that, for a finite $K$, there exist $2^{K}-1$ non-empty subsets of $\calK$.
Denoting the $\card{\calT}$-length vector $\bracknorm{1,1,\ldots, 1}$ by $1{\bracksq{\calT}}$ for any non-empty set $\calT \subseteq \calK$, we upper bound the second term in~\eqref{eq:48} by
\begin{align}
	& \E{ \frac{1}{\bracknorm{\prod_{k \in \calK} M_{k}}} \sum_{m{\bracksq{\calK}}} \sum_\yn W_{Y|X{\bracksq{\calK}}}^{\pn} \bracknorm{\yn| \bxn_{\calK}\bracknorm{m{\bracksq{\calK}},\ell{\bracksq{\calK}}} } \sum_{\widetilde{m}{{\bracksq{\calK}}} \neq m{{\bracksq{\calK}}} }\indic{ \calE_{\widetilde{m}{\bracksq{\calK}}} } } \nex
	& \dupspace { \stackrel{(a)}{=} \E{ \frac{1}{\bracknorm{\prod_{k \in \calK} M_{k}}}  \sum_{\substack{\calT \subseteq \calK: \\ \calT \neq \emptyset }}\sum_{m{\bracksq{\calT}}} \sum_{m{\bracksq{\calT^c}}} \sum_{\substack{\widetilde{m}\bracksq{\calT} : \\ \widetilde{m}_k \neq m_k, \forall k \in \calT}} \sum_\yn W_{Y|X{\bracksq{\calK}}}^{\pn} \bracknorm{\yn| \bxn_{\calK}\bracknorm{m{\bracksq{\calK}},\ell{\bracksq{\calK}}} } \indic{ \calE_{\widetilde{m}{\bracksq{\calT}} m\bracksq{\calT^c}} } }} \label{eq:54_1} \displaybreak[0]\\
	& \dupspace { \stackrel{(b)}{=} \sum_\yn \sum_{\substack{\calT \subseteq \calK: \\ \calT \neq \emptyset}} \sum_{ \substack{ \widetilde{m}{\bracksq{\calT}}: \\ \widetilde{m}_k \neq 1, \forall k \in \calT } } \sum_{\xn_{\calK} \bracknorm{1{\bracksq{\calK}}, \ell{\bracksq{\calK}}}}  \sum_{\xn_{\calT} \bracknorm{\widetilde{m}{\bracksq{\calT}}, \ell{\bracksq{\calT}}}} W^{\pn}_{Y|X{\bracksq{\calK}}} \bracknorm{\yn | \xn_{\calK} \bracknorm{1{\bracksq{\calK}}, \ell{\bracksq{\calK}}}} \bracknorm{\prod_{k\in \calK} \Pi^{\pn}_{X_k}\bracknorm{\xn_{k}\bracknorm{1,\ell_k}} }} \nex
	& \dupspace \dupspace { \times \bracknorm{ \prod_{k\in \calT} \Pi^{\pn}_{X_k}\bracknorm{\xn_{k}\bracknorm{\widetilde{m}_k,\ell_k}} }\indic{\bracknorm{\xn_{\calT} \bracknorm{\widetilde{m}{\bracksq{\calT}}, \ell{\bracksq{\calT}}}, \xn_{\calT^c} \bracknorm{1{\bracksq{\calT^c}}, \ell{\bracksq{\calT^c}}} ,\yn } \in \agamman{} }  }\displaybreak[0]\\
	& \dupspace =\! \sum_\yn \! \sum_{\substack{\calT \subseteq \calK: \\ \calT \neq \emptyset}} \! \sum_{ \substack{ \widetilde{m}{\bracksq{\calT}}: \\ \widetilde{m}_k \neq 1, \forall k \in \calT } } \sum_{\xn_{\calT^c} \bracknorm{1{\bracksq{\calT^c}}, \ell{\bracksq{\calT^c}}}}  \sum_{\xn_{\calT} \bracknorm{\widetilde{m}{\bracksq{\calT}}, \ell{\bracksq{\calT}}}} \!\!\!\!\!\!\!\! W^{\pn}_{Y|X{\bracksq{\calT^c}}} \bracknorm{\yn | \xn_{\calT^c} \bracknorm{1{\bracksq{\calT^c}}, \ell{\bracksq{\calT^c}}}} \bracknorm{\prod_{k\in \calT^c} \Pi^{\pn}_{X_k}\bracknorm{\xn_{k}\bracknorm{1,\ell_k}}\!\! } \nex
	& \dupspace \dupspace  \times \bracknorm{ \prod_{k\in \calT} \Pi^{\pn}_{X_k}\bracknorm{\xn_{k}\bracknorm{\widetilde{m}_k,\ell_k}} }\indic{\bracknorm{\xn_{\calT} \bracknorm{\widetilde{m}{\bracksq{\calT}}, \ell{\bracksq{\calT}}}, \xn_{\calT^c} \bracknorm{1{\bracksq{\calT^c}}, \ell{\bracksq{\calT^c}}} ,\yn } \in \agamman{} } \label{eq:54_2} \displaybreak[0]\\
	& \dupspace \stackrel{{(c)}}{\leq} \! \sum_\yn \! \sum_{\substack{\calT \subseteq \calK: \\ \calT \neq \emptyset}} \! \sum_{ \substack{ \widetilde{m}{\bracksq{\calT}}: \\ \widetilde{m}_k \neq 1, \forall k \in \calT } } \! \sum_{\xn_{\calT^c} \bracknorm{1{\bracksq{\calT^c}}, \ell{\bracksq{\calT^c}}}} \!  \sum_{\xn_{\calT} \bracknorm{\widetilde{m}{\bracksq{\calT}}, \ell{\bracksq{\calT}}}} \!\!\!\!\!\!\!\! W^{\pn}_{Y|X{\bracksq{\calT^c}}} \bracknorm{\yn | \xn_{\calT^c} \bracknorm{1{\bracksq{\calT^c}}, \ell{\bracksq{\calT^c}}}} \bracknorm{\prod_{k\in \calT^c} \Pi^{\pn}_{X_k}\bracknorm{\xn_{k}\bracknorm{1,\ell_k}} \!\!} \nex
	& \dupspace  \dupspace \times \bracknorm{ \prod_{k\in \calT} \Pi^{\pn}_{X_k}\bracknorm{\xn_{k}\bracknorm{\widetilde{m}_k,\ell_k}} }\indic{\bracknorm{\xn_{\calT} \bracknorm{\widetilde{m}{\bracksq{\calT}}, \ell{\bracksq{\calT}}}, \xn_{\calT^c} \bracknorm{1{\bracksq{\calT^c}}, \ell{\bracksq{\calT^c}}} ,\yn } \in \agamman{_\calT} } \label{eq:54_3}  \displaybreak[0]\\
	& \dupspace \leq { \sum_{\substack{\calT \subseteq \calK: \\ \calT \neq \emptyset}}e^{- \gamma_\calT} \bracknorm{ \prod_{k \in \calT }M_k} } \bracknorm{ \sum_\yn \sum_{\xn{\bracksq{\calK}}} W^{\pn}_{Y|X[\calK]}\bracknorm{\yn|\xn{\bracksq{\calK}} } \bracknorm{ \prod_{k\in \calK} \Pi^{\pn}_{X_k}\bracknorm{\xn_{k}} }}  \\
	& \dupspace = \sum_{\substack{\calT \subseteq \calK: \\ \calT \neq \emptyset}}e^{- \gamma_\calT} \bracknorm{ \prod_{k \in \calT }M_k}, \label{eq:54}
\end{align} 
{where $(a)$ follows from rewriting the left hand side of~\eqref{eq:54_1} in terms of the positions in which the two vectors $m\bracksq{\calK}$ and $\widetilde{m} \bracksq{\calK}$ do not match, $(b)$ follows from setting $m\bracksq{\calK} = 1\bracksq{\calK}$ without loss of generality, and $(c)$ follows from the fact that $\agamman{}$ in the indicator function of~\eqref{eq:54_2} is a subset of $\agamman{_\calT}$ in the indicator function of~\eqref{eq:54_3} by definition of $\agamman{}$. }
Combining~\eqref{eq:53} and~\eqref{eq:54}, we upper bound~\eqref{eq:48} by
\begin{align}
	\E{P_e^n} \leq \sum_{\substack{\calT \subseteq \calK: \\ \calT \neq \emptyset}}e^{- \gamma_\calT} \bracknorm{ \prod_{k \in \calT }M_k}  + \exp \bracknorm{-c_2 n \alpha_n}. \label{eq:55}
\end{align}
Using the definition of $\gamma_\calT$,~\eqref{eq:52} and~\eqref{eq:55}, we conclude that for an arbitrary $\delta \in \bracknorm{0,1}$ and $n$ large enough, if
	\begin{align}
		\sum_{k \in \calT} \log M_k &=\bracknorm{1-\delta} \bracknorm{1 - \mu} n \alpha_n \sum_{k \in \calT} \rho_k  \avgD{P_k}{\pz} , \label{eq:56}
	\end{align}
	for every non-empty set $\calT \subseteq \calK$, then there exists a constant $\xi > 0$ such that 
	\begin{align}
		\E{P_e^n} \leq \exp \bracknorm{-\xi n \alpha_n}. \label{eq:57}
	\end{align}
	If $\calT$ is a singleton set $\brackcurl{k}$, where $k \in \calK$, it follows from~\eqref{eq:56} that
	\begin{align}
		\log M_k = \bracknorm{1-\delta} \bracknorm{1 - \mu} \rho_k n \alpha_n   \avgD{P_k}{\pz}. \label{eq:58}
	\end{align}
	Observing~\eqref{eq:56} and~\eqref{eq:58}, we conclude that~\eqref{eq:56} is automatically satisfied for every non-empty set $\calT \subseteq \calK$, if $\log M_k$ satisfies~\eqref{eq:58} for every $k \in \calK$.
	
\section{Proof of Lemma~\ref{lem:resolvability}} \label{sec:lemma_res}
	Define the set $\bseteta{} \eqdef \bigcap_{\substack{\calT \subseteq \calK: \\ \calT \neq \emptyset}} \bseteta{_\calT}$ with 
	\begin{align}
		\bseteta{_{\calT}} & \eqdef \brackcurl{ \bracknorm{\xn{\bracksq{\calT}}, \zn} \in \calX^n{\bracksq{\calT}} \times \calZ^n: \log \frac{W^\pn_{Z|X{\bracksq{\calT}}}\bracknorm{\zn|\xn{\bracksq{\calT}}} }{\qaln\bracknorm{\zn}} \leq \eta_{\calT}}, \label{eq:59}
	\end{align}
	where 
	\begin{align}
	\eta_\calT & \eqdef \bracknorm{1+\mu} n \avgI{X{\bracksq{\calT}};Z}, \label{eq:59_1}
	\end{align}
	for every non-empty set $\calT \subseteq \calK$ and an arbitrary $\mu > 0$. For $\bracknorm{m{\bracksq{\calK}},\ell{\bracksq{\calK}}} \in \times_{k=1}^K \intseq{1}{M_k} \times \times_{k=1}^K \intseq{1}{L_k}$, $\mathbb{E}_{\sim \bracknorm{m{\bracksq{\calK}},\ell{\bracksq{\calK}}}}$ denotes the expectation taken over all $ \brackcurl{\bxn_{\calK}{\bracknorm{\widetilde{m}{\bracksq{\calK}},\widetilde{\ell}{\bracksq{\calK}}}}}_{\substack{ \bracknorm{\widetilde{m}{\bracksq{\calK}},\widetilde{\ell}{\bracksq{\calK}}} \in \bracknorm{ \times_{k=1}^K \intseq{1}{M_k} \times \times_{k=1}^K \intseq{1}{L_k} }: \\ \bracknorm{\widetilde{m}{\bracksq{\calK}},\widetilde{\ell}{\bracksq{\calK}}} \neq \bracknorm{m{\bracksq{\calK}},\ell{\bracksq{\calK}}} }} $. Let us analyze the KL divergence between $\qhatn$ and $\qaln$ averaged over all random codebooks.
\begin{align}
	& \E{\avgD{\qhatn}{\qaln}} \nex
	& = \E{ \sum_\zn \qhatn\bracknorm{\zn} \log \frac{\qhatn\bracknorm{\zn}}{\qaln\bracknorm{\zn}} } \\
	& =\E{\! \sum_\zn \frac{ \sum\limits_{m{{\bracksq{\calK}}}} \sum\limits_{\ell{{\bracksq{\calK}}}}  W_{Z|X{{\bracksq{\calK}}}}^{\pn} \bracknorm{\zn | \bxn_\calK \bracknorm{m{{\bracksq{\calK}}}, \ell{{\bracksq{\calK}}}}  }}{\bracknorm{\prod\limits_{k \in \calK} M_{k} L_{k} }}  \log \!\bracknorm{\!  \frac{\sum\limits_{\widetilde{m}{{\bracksq{\calK}}}} \sum\limits_{\widetilde{\ell}{{\bracksq{\calK}}}}  \!W_{Z|X{{\bracksq{\calK}}}}^{\pn} \bracknorm{\zn | \bxn_\calK \bracknorm{\widetilde{m}{{\bracksq{\calK}}}, \widetilde{\ell}{{\bracksq{\calK}}}}}}{\bracknorm{\prod\limits_{k \in \calK} M_{k}L_{k} }\qaln\bracknorm{\zn}}}\!\! }  \\
	& \stackrel{(a)}{\leq}   \sum_\zn \sum_{m{{\bracksq{\calK}}}} \sum_{\ell{{\bracksq{\calK}}}} \sum_{\xn_{\calK}\bracknorm{m{\bracksq{\calK}}, \ell{\bracksq{\calK}}}} \frac{  W^{\pn}_{Z|X{\bracksq{\calK}}}\bracknorm{\zn | \xn_\calK \bracknorm{m{{\bracksq{\calK}}}, \ell{{\bracksq{\calK}}}}  }   \bracknorm{ \prod_{k \in \calK} \Pi^{\pn}_{X_k}\bracknorm{\xn_{k}\bracknorm{m_{k},\ell_{k}}}}}{\bracknorm{ \prod_{k \in \calK} M_{k}L_{k} }} \nex
	& \dupspace \dupspace \times \log \mathbb{E}_{\sim \bracknorm{m{\bracksq{\calK}},\ell{\bracksq{\calK}}}} \bracknorm{\frac{\sum_{\widetilde{m}{{\bracksq{\calK}}}} \sum_{\widetilde{\ell}{{\bracksq{\calK}}}}  W_{Z|X{{\bracksq{\calK}}}}^{\pn} \bracknorm{\zn | \bxn_\calK \bracknorm{\widetilde{m}{{\bracksq{\calK}}}, \widetilde{\ell}{{\bracksq{\calK}}}}}}{\bracknorm{\prod_{k \in \calK} M_{k}L_{k} }\qaln\bracknorm{\zn}}}, \label{eq:60}
\end{align}
where $(a)$ follows from Jensen's inequality. 
Let us analyze the $\log$ term in~\eqref{eq:60}. 
\begin{align}
	& \log \mathbb{E}_{\sim \bracknorm{m{\bracksq{\calK}},\ell{\bracksq{\calK}}}} \bracknorm{\frac{\sum_{\widetilde{m}{{\bracksq{\calK}}}} \sum_{\widetilde{\ell}{{\bracksq{\calK}}}}  W_{Z|X{{\bracksq{\calK}}}}^{\pn} \bracknorm{\zn | \bxn_\calK \bracknorm{\widetilde{m}{{\bracksq{\calK}}}, \widetilde{\ell}{{\bracksq{\calK}}}}}}{\bracknorm{\prod_{k \in \calK} M_{k}L_{k} }\qaln\bracknorm{\zn}}} \nex
	& \dupspace { = \log \bracknorm{\!\frac{\splitfrac{\sum_{\calT \subseteq \calK}\sum_{\widetilde{m}\bracksq{\calT^c}: \widetilde{m}_k \neq m_k, \forall k \in {\calT^c}} \sum_{\widetilde{\ell}\bracksq{\calT^c}: \widetilde{\ell}_k \neq \ell_k, \forall k \in {\calT^c}} \sum_{\xn_{\calT^c}\bracknorm{\widetilde{m}\bracksq{\calT^c},\widetilde{\ell}\bracksq{\calT^c}}}}{ W_{Z|X\bracksq{\calT}X\bracksq{\calT^c}}\bracknorm{\zn|\xn_\calT\bracknorm{m\bracksq{\calT}, \ell\bracksq{\calT}}\xn_{\calT^c}\bracknorm{\widetilde{m}\bracksq{\calT^c}, \widetilde{\ell}\bracksq{\calT^c}}} \bracknorm{\prod_{k\in \calT^c} \Pi_{X_k}^\pn\bracknorm{\xn_k\bracknorm{\widetilde{m}_k, \widetilde{\ell}_k}} } }}{\bracknorm{\prod_{k \in \calK} M_{k}L_{k} }\qaln\bracknorm{\zn}}\!} }\label{eq:61_1} \displaybreak[0]\\
	& \dupspace { \stackrel{(a)}{=} \log \left(\frac{\splitfrac{\sum_{\calT \subseteq \calK: \, \calT \neq \emptyset}\sum_{\widetilde{m}\bracksq{\calT^c}: \widetilde{m}_k \neq m_k, \forall k \in {\calT^c}} \sum_{\widetilde{\ell}\bracksq{\calT^c}: \widetilde{\ell}_k \neq \ell_k, \forall k \in {\calT^c}} \sum_{\xn_{\calT^c}\bracknorm{\widetilde{m}\bracksq{\calT^c},\widetilde{\ell}\bracksq{\calT^c}}}}{ W_{Z|X\bracksq{\calT}X\bracksq{\calT^c}}\bracknorm{\zn|\xn_\calT\bracknorm{m\bracksq{\calT}, \ell\bracksq{\calT}}\xn_{\calT^c}\bracknorm{\widetilde{m}\bracksq{\calT^c}, \widetilde{\ell}\bracksq{\calT^c}}} \bracknorm{\prod_{k\in \calT^c} \Pi_{X_k}^\pn\bracknorm{\xn_k\bracknorm{\widetilde{m}_k, \widetilde{\ell}_k}} } }}{\bracknorm{\prod_{k \in \calK} M_{k}L_{k} }\qaln\bracknorm{\zn}}\right.} \nex
	& \dupspace \dupspace \dupspace \dupspace { \left. + \frac{\splitfrac{\sum_{\widetilde{m}\bracksq{\calK}: \widetilde{m}_k \neq m_k, \forall k \in {\calK}} \sum_{\widetilde{\ell}\bracksq{\calK}: \widetilde{\ell}_k \neq \ell_k, \forall k \in {\calK}} \sum_{\xn_{\calK}\bracknorm{\widetilde{m}\bracksq{\calK},\widetilde{\ell}\bracksq{\calK}}}}{ W_{Z|X\bracksq{\calK}}\bracknorm{\zn|\xn_{\calK}\bracknorm{\widetilde{m}\bracksq{\calK}, \widetilde{\ell}\bracksq{\calK}}} \bracknorm{\prod_{k\in \calK} \Pi_{X_k}^\pn\bracknorm{\xn_k\bracknorm{\widetilde{m}_k, \widetilde{\ell}_k }} } }}{\bracknorm{\prod_{k \in \calK} M_{k}L_{k} }\qaln\bracknorm{\zn}} \right)}  \displaybreak[0]\\
	& \dupspace { \stackrel{(b)}{=} \log  \bracknorm{ \sum_{\substack{\calT \subseteq \calK: \\ \calT \neq \emptyset}} \frac{  W^\pn_{Z|X{\bracksq{\calT}}}\bracknorm{\zn |  \xn_{\calT}\bracknorm{m{\bracksq{\calT}}, \ell{\bracksq{\calT}}}}}{\bracknorm{ \prod_{k \in \calT} M_{k}L_{k}}\qaln\bracknorm{\zn}}  + \frac{\sum_{\widetilde{m}\bracksq{\calK}: \widetilde{m}_k \neq m_k, \forall k \in {\calK}} \sum_{\widetilde{\ell}\bracksq{\calK}: \widetilde{\ell}_k \neq \ell_k, \forall k \in {\calK}} \qaln\bracknorm{\zn}  }{\bracknorm{\prod_{k \in \calK} M_{k}L_{k} }\qaln\bracknorm{\zn}} } }  \displaybreak[0]\\
	& \dupspace \leq \log  \bracknorm{ \sum_{\substack{\calT \subseteq \calK: \\ \calT \neq \emptyset}} \frac{  W^\pn_{Z|X{\bracksq{\calT}}}\bracknorm{\zn |  \xn_{\calT}\bracknorm{m{\bracksq{\calT}}, \ell{\bracksq{\calT}}}}}{\bracknorm{ \prod_{k \in \calT} M_{k}L_{k}}\qaln\bracknorm{\zn}}  + 1  } \\
	& \dupspace = \log  \bracknorm{ \sum_{ \substack{\calT \subseteq \calK: \\ \calT \neq \emptyset} } \frac{  W^\pn_{Z|X{\bracksq{\calT}}}\bracknorm{\zn |  \xn_{\calT}\bracknorm{m{\bracksq{\calT}}, \ell{\bracksq{\calT}}}}}{\bracknorm{ \prod_{k \in \calT} M_{k}L_{k}}\qaln\bracknorm{\zn}}  + 1  } \indic{\bracknorm{ \xn_{\calK}\bracknorm{m{\bracksq{\calK}}, \ell{\bracksq{\calK}}}, \zn } \in \bseteta{}} \nex
	& \dupspace \dupspace + \log  \bracknorm{ \sum_{ \substack{\calT \subseteq \calK: \\ \calT \neq \emptyset} } \frac{  W^\pn_{Z|X{\bracksq{\calT}}}\bracknorm{\zn |  \xn_{\calT}\bracknorm{m{\bracksq{\calT}}, \ell{\bracksq{\calT}}}}}{\bracknorm{ \prod_{k \in \calT} M_{k}L_{k}}\qaln\bracknorm{\zn}}  + 1  } \indic{\bracknorm{ \xn_{\calK}\bracknorm{m{\bracksq{\calK}}, \ell{\bracksq{\calK}}}, \zn } \not \in \bseteta{}}, \label{eq:61}
\end{align}
{where $(a)$ follows from splitting the numerator term in~\eqref{eq:61_1} into two based on whether $\calT$ is empty or not and $(b)$ follows from the fact that $\sum\limits_{\xn_{\calK}\bracknorm{\widetilde{m}\bracksq{\calK},\widetilde{\ell}\bracksq{\calK}}} \!\!\!  W_{Z|X\bracksq{\calK}}\bracknorm{\zn|\xn_{\calK}\bracknorm{\widetilde{m}\bracksq{\calK}, \widetilde{\ell}\bracksq{\calK}}} \bracknorm{\prod\limits_{k\in \calK} \Pi_{X_k}^\pn\bracknorm{\xn_k\bracknorm{\widetilde{m}_k, \widetilde{\ell}_k }} } \allowbreak = \qaln\bracknorm{\zn}$. }
We upper bound the first term in~\eqref{eq:61} by 
\begin{align}
	& \log  \bracknorm{ \sum_{ \substack{\calT \subseteq \calK: \\ \calT \neq \emptyset} } \frac{  W^\pn_{Z|X{\bracksq{\calT}}}\bracknorm{\zn |  \xn_{\calT}\bracknorm{m{\bracksq{\calT}}, \ell{\bracksq{\calT}}}}}{\bracknorm{ \prod_{k \in \calT} M_{k}L_{k}}\qaln\bracknorm{\zn}}  + 1  } \indic{\bracknorm{ \xn_{\calK}\bracknorm{m{\bracksq{\calK}}, \ell{\bracksq{\calK}}}, \zn } \in \bseteta{}} \nex
	& \dupspace \leq \log \bracknorm{\sum_{\substack{\calT \subseteq \calK: \\ \calT \neq \emptyset}} \frac{e^{\eta_\calT}}{\bracknorm{\prod_{k \in \calT }M_{k}L_{k}}} + 1} \\
	& \dupspace \leq \sum_{\substack{\calT \subseteq \calK: \\ \calT \neq \emptyset}} \frac{e^{\eta_\calT}}{\bracknorm{\prod_{k \in \calT }M_{k} L_{k}}}. \label{eq:62}
\end{align}
Defining $\nu_{\min} \eqdef \min_{z} \qz(z)$, we upper bound the second term in~\eqref{eq:61} by
\begin{align}
	& \log  \bracknorm{ \sum_{ \substack{\calT \subseteq \calK: \\ \calT \neq \emptyset} } \frac{  W^\pn_{Z|X{\bracksq{\calT}}}\bracknorm{\zn |  \xn_{\calT}\bracknorm{m{\bracksq{\calT}}, \ell{\bracksq{\calT}}}}}{\bracknorm{ \prod_{k \in \calT} M_{k}L_{k}}\qaln\bracknorm{\zn}}  + 1  } \indic{\bracknorm{ \xn_{\calK}\bracknorm{m{\bracksq{\calK}}, \ell{\bracksq{\calK}}}, \zn } \not \in \bseteta{}} \nex
	& \dupspace \leq \bracknorm{ \log \bracknorm{\frac{1}{\qaln\bracknorm{\zn}} } + \log \bracknorm{ \sum_{ \substack{\calT \subseteq \calK: \\ \calT \neq \emptyset} } \frac{  W^\pn_{Z|X{\bracksq{\calT}}} \bracknorm{\zn |  \xn_{\calT}\bracknorm{m{\bracksq{\calT}}, \ell{\bracksq{\calT}}}}}{\bracknorm{ \prod_{k \in \calT} M_{k}L_{k}}}  + \qaln(\zn)  } } \nex
	& \dupspace \dupspace \dupspace \times \indic{\bracknorm{ \xn_{\calK}\bracknorm{m{\bracksq{\calK}}, \ell{\bracksq{\calK}}}, \zn } \not \in \bseteta{}}  \\ 
	& \dupspace { \stackrel{(a)}{\leq} \bracknorm{ n \log \bracknorm{\frac{1}{\bracknorm{ \prod_{k \in \calK} \bracknorm{1 - \rho_k \alpha_n}} \nu_{\min}}} + \log \bracknorm{ \sum_{ \substack{\calT \subseteq \calK: \\ \calT \neq \emptyset} } 1  + 1  } } \indic{\bracknorm{ \xn_{\calK}\bracknorm{m{\bracksq{\calK}}, \ell{\bracksq{\calK}}}, \zn } \not \in \bseteta{}}} \displaybreak[0]\\ 
	& \dupspace \stackrel{(b)}{\leq} \bracknorm{ n \log \bracknorm{\frac{1}{\bracknorm{ \prod_{k \in \calK} \bracknorm{1 - \rho_k \alpha_n}} \nu_{\min}}} + \log \bracknorm{ 2^K } } \indic{\bracknorm{ \xn_{\calK}\bracknorm{m{\bracksq{\calK}}, \ell{\bracksq{\calK}}}, \zn } \not \in \bseteta{}} \\
	& \dupspace \leq n \log \bracknorm{\frac{2^K}{\bracknorm{ \prod_{k \in \calK}\bracknorm{1 - \rho_k \alpha_n}} \nu_{\min}}} \indic{\bracknorm{ \xn_{\calK}\bracknorm{m{\bracksq{\calK}}, \ell{\bracksq{\calK}}}, \zn } \not \in \bseteta{}}, \label{eq:63} 
\end{align}
{where $(a)$ follows from the fact that we can upper bound both $\frac{  W^\pn_{Z|X{\bracksq{\calT}}} \bracknorm{\zn |  \xn_{\calT}\bracknorm{m{\bracksq{\calT}}, \ell{\bracksq{\calT}}}}}{\bracknorm{ \prod_{k \in \calT} M_{k}L_{k}}}$ and $\qaln(\zn)$ by $1$ and $(b)$ follows from the fact that there only exist $2^K -1$ non-empty subsets of $\calK$.}
Combining~\eqref{eq:62} and~\eqref{eq:63}, we upper bound~\eqref{eq:60} by
\begin{align}
	\E{\avgD{\qhatn}{\qaln}} & \leq \sum_\zn \qaln \bracknorm{\zn} \bracknorm{\sum_{\substack{\calT \subseteq \calK: \\ \calT \neq \emptyset}} \frac{e^{\eta_\calT}}{\bracknorm{\prod_{k \in \calT }M_k L_k}}} +  n \log \bracknorm{\frac{2^K}{\prod_{k \in \calK}\bracknorm{1 - \rho_k \alpha_n} \nu_{\min}}} \P{{\bseteta{}}^c} \displaybreak[0] \\
	& = \sum_{\substack{\calT \subseteq \calK: \\ \calT \neq \emptyset}} \frac{e^{\eta_\calT}}{\bracknorm{\prod_{k \in \calT }M_k L_k}} +  n \log \bracknorm{\frac{2^K}{\prod_{k \in \calK}\bracknorm{1 - \rho_k \alpha_n} \nu_{\min}}} \P{{\bseteta{}}^c}. \label{eq:64} 
\end{align}
{From the definition of $\bseteta{}$}, we obtain
\begin{align}
	\P{{\bseteta{}}^c} \leq \sum_{\substack{\calT \subseteq \calK: \\ \calT \neq \emptyset}} \P{{\bseteta{_{\calT}}}^c}, \label{eq:65}
\end{align}
which follows from the fact that ${\bseteta{}}^c = \bigcup_{\substack{\calT \subseteq \calK: \\ \calT \neq \emptyset}} {\bseteta{_{\calT}}}^c$ and the application of the union bound. We define a zero-mean random variable $V_{\calT} \eqdef \log \frac{W_{Z|X{\bracksq{\calT}}}(Z|X{\bracksq{\calT}}) }{\qal(Z)} - \avgI{X{\bracksq{\calT}}; Z}$ since $\E{\log \frac{W_{Z|X{\bracksq{\calT}}}(Z|X{\bracksq{\calT}}) }{\qal(Z)}} = \avgI{X{\bracksq{\calT}}; Z}$. Note that $\abs{V_\calT}$ is bounded almost surely, and 
\begin{align}
	\E{V^2_\calT} &= \E{\log^2 \frac{W_{Z|X{\bracksq{\calT}}}(Z|X{\bracksq{\calT}}) }{\qal(Z)}} - { \bracknorm{\avgI{X\bracksq{\calT};Z}}^2} \\
	& \stackrel{(a)}{=} \E{\log^2 \frac{W_{Z|X{\bracksq{\calT}}}(Z|X{\bracksq{\calT}}) }{\qal(Z)}} +\bigO{\alpha_n^2}, \label{eq:66_1}
\end{align}
where $(a)$ follows from {Lemma~\ref{lem:covertprocess}.} 
Let us analyze the expectation term in~\eqref{eq:66_1}. 
\begin{align}
	& \E{\log^2 \frac{W_{Z|X{\bracksq{\calT}}}(Z|X{\bracksq{\calT}}) }{\qal(Z)}} \nex 
	& \dupspace = \sum_z \sum_{x{\bracksq{\calT}}} \bracknorm{ \prod_{k \in \calT} \Pi_{X_k}({x_k} ) } W_{Z|X{\bracksq{\calT}}} (z|x{\bracksq{\calT}}) \log^2 \frac{W_{Z|X{\bracksq{\calT}}} (z|x{\bracksq{\calT}})}{\qal(z)} \label{eq:66_3}  \\
	& \dupspace { \stackrel{(a)}{=} \sum_z \bracknorm{ \prod_{k \in \calT} (1 - \rho_k \alpha_n) } W_{Z|X{\bracksq{\calT}}} (z|x_{\emptyset}{\bracksq{\calT}}) \log^2 \frac{W_{Z|X{\bracksq{\calT}}} (z|x_{\emptyset}{\bracksq{\calT}})}{\qal(z)}} \nex
	& \dupspace  \dupspace {+ \sum_z \sum_{x{\bracksq{\calT}} \neq x_\emptyset \bracksq{\calT}} \bracknorm{ \prod_{k \in \calT} \Pi_{X_k}(x_k ) } W_{Z|X{\bracksq{\calT}}} (z|x{\bracksq{\calT}}) \log^2 \frac{W_{Z|X{\bracksq{\calT}}} (z|x{\bracksq{\calT}})}{\qal(z)}} \label{eq:66_6} \\
	& \dupspace { \stackrel{(b)}{=} \sum_z \bracknorm{ \prod_{k \in \calT} (1 - \rho_k \alpha_n) } W_{Z|X{\bracksq{\calT}}} (z|x_{\emptyset}{\bracksq{\calT}}) \log^2 \frac{W_{Z|X{\bracksq{\calT}}} (z|x_{\emptyset}{\bracksq{\calT}})}{\qal(z)} + \bigO{\alpha_n}} \\
	& \dupspace {\stackrel{(c)}{=} \sum_z W_{Z|X{\bracksq{\calT}}} (z|x_{\emptyset}{\bracksq{\calT}}) \log^2 \frac{\qal(z)}{W_{Z|X{\bracksq{\calT}}} (z|x_{\emptyset}{\bracksq{\calT}})} + \bigO{\alpha_n}, }\label{eq:66_4}
\end{align}
{where $(a)$ follows from splitting the term on the right hand side of~\eqref{eq:66_3} into two based on whether $x\bracksq{\calT} = x_\emptyset\bracksq{\calT}$ or not, $(b)$ follows from the fact that at least one of the symbols in $x\bracksq{\calT}$ in the second term in~\eqref{eq:66_6} is the symbol $1$, and $(c)$ follows from the expansion of the product term. 
From the definition of $\qal$, we obtain
\begin{align}
	\qal(z) & = \sum_{x\bracksq{\calT}} W_{Z|X\bracksq{\calT}}(z|x\bracksq{\calT} ) \bracknorm{\prod_{k \in \calT} \Pi_{X_k}(x_k) } \\ 
	& = \sum_{x\bracksq{\calT} \neq x_\emptyset\bracksq{\calT}} W_{Z|X\bracksq{\calT}}(z|x\bracksq{\calT} ) \bracknorm{\prod_{k \in \calT} \Pi_{X_k}(x_k) } + W_{Z|X\bracksq{\calT}}(z|x_\emptyset\bracksq{\calT} ) \bracknorm{\prod_{k \in \calT} \bracknorm{1 - \rho_k \alpha_n} } \displaybreak[0] \\ 
	& =  W_{Z|X\bracksq{\calT}}(z|x_\emptyset\bracksq{\calT} )  + \bigO{\alpha_n}. \label{eq:66_5}
\end{align}
Combining~\eqref{eq:66_4} and~\eqref{eq:66_5}, we obtain
\begin{align}
	\E{\log^2 \frac{W_{Z|X{\bracksq{\calT}}}(Z|X{\bracksq{\calT}}) }{\qal(Z)}} & = \sum_z W_{Z|X{\bracksq{\calT}}} (z|x_{\emptyset}{\bracksq{\calT}}) \log^2 \bracknorm{1 + \frac{\bigO{\alpha_n}}{W_{Z|X{\bracksq{\calT}}} (z|x_{\emptyset}{\bracksq{\calT}})}}  + \bigO{\alpha_n}	\\
	& \stackrel{(a)}{=} \bigO{\alpha_n}, \label{eq:66_2}
\end{align}
where $(a)$ follows from the application of the Taylor series of the $\log$ term. }
Using the definition of $\eta_\calT$ in~\eqref{eq:59_1}, for an arbitrary $\mu>0$, we upper bound~\eqref{eq:65} by 
\begin{align}
	\P{{\bseteta{}}^c} & \leq \sum_{\substack{\calT \subseteq \calK: \\ \calT \neq \emptyset}} \P{ \log \frac{W^\pn_{Z|X{\bracksq{\calT}}}\bracknorm{\bzn|\bxn{\bracksq{\calT}}} }{\qaln\bracknorm{\bzn}} > \eta_{\calT} } \\
	& = \sum_{\substack{\calT \subseteq \calK: \\ \calT \neq \emptyset}} \P{ \sum_{i=1}^n \log \frac{W_{Z|X{\bracksq{\calT}}}(Z|X{\bracksq{\calT}})}{\qal(Z)} > \bracknorm{1+\mu} n \avgI{X{\bracksq{\calT}};Z} }  \\
	& = \sum_{\substack{\calT \subseteq \calK: \\ \calT \neq \emptyset}}  \P{ \sum_{i=1}^n V_\calT > \mu n \avgI{X{\bracksq{\calT}};Z} }  \\
	& \stackrel{(a)}{\leq} \sum_{\substack{\calT \subseteq \calK: \\ \calT \neq \emptyset}}  \exp \bracknorm{- \frac{\frac{1}{2} \bracknorm{\mu n \avgI{X{\bracksq{\calT}};Z}}^2 }{n \bigO{\alpha_n} + \frac{1}{3} c \mu n \avgI{X{\bracksq{\calT}};Z} }}  \\
	& \stackrel{(b)}{ \leq} \sum_{\substack{\calT \subseteq \calK: \\ \calT \neq \emptyset}} \exp \bracknorm{- c_1 n \alpha_n}\\
	& \leq \exp \bracknorm{- c_2 n \alpha_n}, \label{eq:66}
\end{align}
for appropriate constants $c, c_1, c_2>0$, where $(a)$ follows from using Bernstein's inequality, {and $(b)$ follows from the fact that $\mathbb{I}\bracknorm{X{\bracksq{\calT}}; Z} = \sum_{k \in \calT} \rho_k \alpha_n \avgD{Q_k}{\qz} + \bigO{\alpha_n^2}$, for any non-empty set $\calT \subseteq \calK$, from Lemma~\ref{lem:covertprocess}. }
Combining{~\eqref{eq:64} and~\eqref{eq:66},} for an appropriate  constant $c_3>0$, we obtain
\begin{align}
	\E{\avgD{\qhatn}{\qaln}} & \leq \sum_{\substack{\calT \subseteq \calK: \\ \calT \neq \emptyset}} \frac{e^{\eta_\calT}}{\bracknorm{\prod_{k \in \calT }M_k L_k}} + \exp \bracknorm{- c_3 n \alpha_n}. \label{eq:67}
\end{align}
Using the definition of $\eta_\calT$, we conclude from~\eqref{eq:67} that for an arbitrary $\delta \in \bracknorm{0,1}$ and a large $n$, if
	\begin{align}
		\sum_{k \in \calT} \log (M_kL_k) = \bracknorm{1+\delta} \bracknorm{1 + \mu} n \alpha_n \sum_{k \in \calT} \rho_k  \avgD{Q_k}{\qz}, \label{eq:69}
	\end{align}
	for every non-empty set $\calT \subseteq \calK $, then there exists a constant $\xi > 0$, such that
	\begin{align}
		\E{\avgD{\qhatn}{\qaln}} \leq \exp \bracknorm{-\xi n \alpha_n}. \label{eq:70}
	\end{align}
	If $\calT$ is a singleton set $\brackcurl{k}$, where $k \in \calK$, it follows from~\eqref{eq:69} that
	\begin{align}
		\log (M_kL_k) = \bracknorm{1+\delta} \bracknorm{1 + \mu}  \rho_k n \alpha_n   \avgD{Q_k}{\qz}. \label{eq:71}
	\end{align}
	Observing~\eqref{eq:69} and~\eqref{eq:71}, we conclude that~\eqref{eq:69} is automatically satisfied for every non-empty set $\calT \subseteq \calK$, if $\log M_kL_k$ satisfies~\eqref{eq:71} for every $k \in \calK$. 
	
\section{Proof of Lemma~\ref{lem:identificationcode}}
        \label{sec:proof-lemma-identification}
Since $\avgD{\qhatn}{\qaln} \leq \exp \bracknorm{-\xi_2 n \alpha_n}$, it follows from Pinsker's inequality that $\V{\qhatn, \qaln} \allowbreak \leq \exp \bracknorm{- \frac{1}{2} \xi_2 n \alpha_n}$. Furthermore, we write
\begin{align}
	\D{\qhatn}{\qzn} & = \D{\qhatn}{\qaln} + \sum_{\zn} \qhatn(\zn) \log \bracknorm{\frac{\qaln(\zn)}{\qzn(\zn)}} \\
	& = \D{\qhatn}{Q_{\alpha_n}^{\pn}} + \D{\qaln}{\qzn} + \sum_{\zn} \bracknorm{\qhatn(\zn) - \qaln(\zn) }\log \bracknorm{\frac{\qaln(\zn)}{\qzn(\zn)}}. \label{eq:35}
\end{align}
Rearranging the terms in~\eqref{eq:35} and taking the absolute value yields
\begin{align}
	\abs{\D{\qhatn}{\qzn} - \D{\qaln}{\qzn}} & \leq \D{\qhatn}{Q_{\alpha_n}^{\pn}} + \abs{ \sum_{\zn} \bracknorm{\qhatn(\zn) - \qaln(\zn) }\log \bracknorm{\frac{\qaln(\zn)}{\qzn(\zn)}}}. \label{eq:36}
\end{align}
Defining $\nu_{\min} \eqdef \min_{z} \qz(z)$, we bound the second term on the right hand side of~\eqref{eq:36} for $n$ large enough as follows.
\begin{align}
 	& \left|\sum_{\zn} \bracknorm{\qhatn(\zn) - \qaln(\zn) }\log \bracknorm{\frac{\qaln(\zn)}{\qzn(\zn)}}\right| \nex
 	& \dupspace \leq \sum_{\zn} \left|\qhatn(\zn) - \qaln(\zn) \right| \left| \log \bracknorm{\frac{\qaln(\zn)}{\qzn(\zn)}}\right|  \\
 	& \dupspace =\! \sum_{\zn} \! \left|\qhatn(\zn) - \qaln(\zn) \right| \! \bracknorm{\! \log\! \bracknorm{\frac{\qaln(\zn)}{\qzn(\zn)}} \!\!\indic{\qaln(\zn) \!\geq\! \qzn(\zn)} \!+\! \log\! \bracknorm{\frac{\qzn(\zn)}{\qaln(\zn)}} \!\!\indic{\qaln(\zn) \!<\! \qzn(\zn)}\! \!} \\ 
 	& \dupspace \stackrel{(a)}{\leq} \sum_{\zn} \left|\qhatn(\zn) - \qaln(\zn) \right|  \Bigg( n \log \frac{1}{\nu_{\min}} \indic{\qaln(\zn) \geq \qzn(\zn)}  \nex
 	& \dupspace \dupspace \dupspace \dupspace \dupspace \dupspace \dupspace \dupspace \dupspace  + \sum_{i=1}^n \log \frac{\qz(z_i)}{\bracknorm{ \prod_{k \in \calK} \bracknorm{1 - \rho_k \alpha_n}}\qz(z_i)} \indic{\qaln(\zn) < \qzn(\zn)} \Bigg) \\
 	& \dupspace \leq \sum_{\zn} \left|\qhatn(\zn) - \qaln(\zn) \right|   \bracknorm{ n \log \frac{1}{\nu_{\min}}  + n \log \frac{1}{\prod_{k \in \calK} \bracknorm{1 - \rho_k \alpha_n}} } \\
 	& \dupspace = 2 \V{\qhatn, \qaln} \bracknorm{ n \log \frac{1}{ \bracknorm{ \prod_{k \in \calK} \bracknorm{1 - \rho_k \alpha_n} }\nu_{\min}} } \displaybreak[0] \\
 	& \dupspace \leq 2 \exp \bracknorm{- \frac{1}{2} \xi_2 n \alpha_n}  \bracknorm{ n \log \frac{1}{ \bracknorm{ \prod_{k \in \calK} \bracknorm{1 - \rho_k \alpha_n}} \nu_{\min}} }, \label{eq:37}
\end{align}
where $(a)$ follows from the fact that $\qal(z) \geq \prod_{k \in \calK} \bracknorm{1 - \rho_k \alpha_n} {\qz(z)} $. 

\section{Extension to non-binary input alphabets} \label{sec:multipleinput}
    To be more specific, each user is now characterized by a distinct input alphabet  $\calX_k \eqdef \intseq{0}{N_k}$ with one innocent symbol $0$ and $N_k$ information symbols. The input distributions defined in the manuscript need to be suitably modified as follows.
    \begin{align}
      \forall k,\quad \Pi_{X_k}(0)=1-\rho_k\alpha_n, \text{ and for $i\in\intseq{1}{N_k}$, } \, \Pi_{X_k}(i)=\rho_k\beta_{k,i}\alpha_n,
    \end{align}
    where $\sum_{i\in\intseq{1}{N_k}}\beta_{k,i}=1$. We need to introduce a new notation to describe the distributions induced by a fixed choice of input.
    \begin{align}
      \forall \mathbf{x}\in\bigtimes_{i=1}^K\calX_i \quad Q_{\mathbf{x}}(z) = W_{Z|X[\calK]}(z|\mathbf{x}).
    \end{align}
As done before, for a given vector $\mathbf{x}$, the vector $\mathbf{x}[\calT]$ is the subvector of size $\card{\calT}$ comprised of the components of $\mathbf{x}$ with index in $\calT$. In addition $\mathbf{x}_{\calT}=(x_{\calT,1},\dots,x_{\calT,K})$ is a $K$ length vector which contains the symbol $0$ in positions indexed by $\calT^c$. For $\calT\subseteq\calK$, we define the distributions
\begin{align}
  Q_{\calT}(z) \eqdef \sum_{\mathbf{x}_\calT}\left(\prod_{k\in\calT}\beta_{k,x_{\calT,k}}\right) Q_{\mathbf{x}_\calT}(z),
\end{align}
and
\begin{align}
Q_{\alpha_n}(z) \eqdef \sum_{x{{\bracksq{\calK}}}} W_{Z|X{{\bracksq{\calK}}}}(z|x{{\bracksq{\calK}}})\bracknorm{ \prod_{k \in \calK} \Pi_{X_k}(x_k)}.
\end{align}

In the special case that $\calT$ is a singleton, say $\{k\}$, we simply write $Q_{k}$ in place of $Q_{\calT}$, and if the unique non-zero symbol in $\mathbf{x}_\calT$ is $i\in\intseq{1}{N_k}$,  we write $Q_{k,i}$ in place of $Q_{\mathbf{x}_\calT}$. With this convention, note that we have
\begin{align}
  Q_{k}(z) = \sum_{i\in\intseq{1}{N_k}}\beta_{k,i}Q_{k,i}(z).
\end{align}
For $\rhovector\in[0,1]^K$ and $\betavector\in[0,1]^K$, we finally introduce
\begin{align}
  \chi(\rhovector,\betavector) = \sum_z\frac{\left(\sum_{k\in\calK}\rho_k\left(Q_k(z)-Q_\emptyset(z)\right)\right)^2}{Q_\emptyset(z)} =  \sum_z\frac{\left(\sum_{k\in\calK}\rho_k\left(\sum_{i\in\intseq{1}{N_k}}\beta_{k,i}Q_{k,i}(z)-Q_\emptyset(z)\right)\right)^2}{Q_\emptyset(z)}.
\end{align}

Similar notation holds when focusing on the main channel instead of the warden channel, in which case we write $P$ instead of $Q$.

With this notation, one can check that Lemma 1 may be extended to obtain
\begin{align}
      \frac{\alpha_n^2}{2} \bracknorm{1 + \sqrt{\alpha_n}} \chi_n{\bracknorm{\rhovector,\betavector}} \geq \avgD{Q_{\alpha_n}}{\qz} \geq \frac{\alpha_n^2}{2} \bracknorm{1 - \sqrt{\alpha_n}} \chi_n{\bracknorm{\rhovector,\betavector}},
\end{align}
and
\begin{align}
    	\avgI{X{{\bracksq{\calT}}};Z} = \alpha_n \sum_{k \in \calT} \rho_k\sum_{i=1}^{N_k}\beta_{k,i}\avgD{Q_{k,i}}{\qz} + \bigO{\alpha_n^2}.
\end{align}
Notice that while $\chi(\rhovector)$ only depends on the distributions $\brackcurl{Q_k}_{k \in \calK}$, the expansion of the mutual information involves distributions $\brackcurl{Q_{k,i}}$, which effectively forces us to keep track of the $\{\beta_{k,i}\}$.

It is then not too painful to check that the covert capacity region contains the region defined by
	\begin{align}
		\mathop{\bigcup_{\{\rho_k\}_{k\in\calK}\in[0,1]^K:\sum_{k \in \calK} \rho_k = 1}}_{\{\beta_{k,i}\}_{k\in\intseq{1}{K},i\in\intseq{1}{N_k}}:\forall k\{\beta_{k,i}\}_{i=1}^{N_{k}}\in[0,1]^{N_k},\sum_{i=1}^{N_k}\beta_{k,i}=1}\!\!\!\!\!\!\!\!\!\!\!\!\!\!\!\!\!\!\!\!\!\!\left\{\{r_k\}_{k\in\calK}:\forall k\in\calK,\quad r_k\leq \sqrt{\frac{2}{\chi\bracknorm{\rhovector,\betavector}}} \rho_k \sum_{i=1}^{N_k}\beta_{k,i}\avgD{P_{k,i}}{\pz}\right\}.  
	\end{align}

For the converse part, one can define
\begin{align}
  1-\mu_{kj}^{(n)} \eqdef \Pi_{X_{kj}}(0) \eqdef \frac{ \sum_{m_k = 1}^{M_k} \sum_{\ell_k = 1}^{L_k} \indic{X_{kj}(m_{k}, \ell_k)=0} }{M_kL_k},
\end{align}
and for $i\in\intseq{1}{N_k}$
\begin{align}
  \Pi_{X_{kj}}(i)\eqdef \mu_{kj}^{(n)} \beta_{k,i,j}^{(n)}.
\end{align}

in which case the steps leading to the lower bound  of the KL divergence are identical thanks to our redefinition of $Q_{\calT}$ and $Q_k$ done earlier. The steps leading to the upper bound of the mutual information require slightly more care because (88) in the manuscript must be replaced by
\begin{multline}
  \sum_{\calT \subseteq \calK} \sum_{\mathbf{x}_\calT} \bracknorm{\prod_{i \in \calK} \Pi_{X_{ij}} \bracknorm{x_{\calT,i}}} \avgD{P_{\mathbf{x}_\calT}}{\pz} \\
  - \sum_y \sum_{\calT \subseteq \calK} \sum_{\mathbf{x}_\calT}  \bracknorm{\prod_{i \in \calK} \Pi_{X_{ij}} \bracknorm{x_{\calT,i}}}P_{\mathbf{x}_\calT}\bracknorm{y}  \log \frac{W_{Y_j|X_{\bracknorm{j}}\bracksq{\calK \setminus \brackcurl{k}}} \bracknorm{y|x_{\calT} \bracksq{\calK \setminus \brackcurl{k}}} }{\pz(y)},
\end{multline}
to account for multiple information symbols. Checking how this affects the remaining calculations requires additional care, but one can check that we obtain a modified version of (102) in the manuscript of the form
\begin{align}
  \avgI{X_{kj}; Y_j |X_j \bracksq{\calK\!\setminus \!\brackcurl{k}} } \leq \mun{kj} \sum_{\ell=1}^{N_k}\beta_{i,\ell,j}^{(n)} \avgD{P_{k,\ell}}{\pz} + d_4 \mun{\max} { \sum_{i \in \calK} \mun{ij}\sum_{\ell=1}^{N_i}\beta_{i,\ell,j}^{(n)} }.
\end{align}
Following the exact same steps in the manuscript, one obtains the converse matching the achievability region highlighted earlier. The analysis of the least achievable key rates on the boundary follows similarly.

While we could include all the considerations outlined above in the manuscript, we feel that they do not really add much to the paper, and require the introduction of another complex layer of notation. We would welcome the reviewer's opinion on this matter, but we propose for now to add the following simpler statement in the conclusion.

\section{Extension to AWGN channels} \label{sec:AWGN}
     If covertness were to be measured with variational distance and if one were to use on-off-keying, one could follow the approach outlined in~[32] and handle the covert constraint as done in our achievability proof. However, since our results focus on KL divergence to measure covertness, an achievability proof must accommodate the continuous nature of the AWGN channel alphabet and possibly the need to use input distributions that are not discrete (see~[5]). One solution is to use a resolvability exponent approach~[39],~[40],~[41] instead of the typical-sequence approach used in the manuscript to obtain bounds for the KL divergence $\avgD{\smash{\qhatn}}{\qzn}$ that do not depend on the alphabet cardinality. One technical aspect of this approach is that one must perform a careful Taylor series of the resolvability exponent.

    As for the converse part, one can follow the steps used in [5] with the necessary adaptations to handle multiple users. More specifically, Following the single-letterization approach of~[5], we obtain
    \begin{align}
\avgD{\hat{Q}_n}{Q_0^{\otimes n}} & \geq \sum_{j=1}^n\avgD{\hat{Q}_j}{Q_0}\\
  & = \sum_{j=1}^n\left(-h(\hat{Q}_j^n) +\frac{1}{2}\log 2\pi \sigma^2 + \E[\hat{Q}_j]{\frac{Z^2}{2\sigma^2}}\right)
\end{align}
Because $Z_j$ and the channel inputs $\{X_{i,j}\}_{}$ are independent (by definition), we have
\begin{align}
  \Var{Z} = \sigma^2 + \sum_{i\in\calK}\underbrace{\Var{X_{i,j}}}_{\eqdef \theta_{i,j}}.
\end{align}
Since a Gaussian distribution maximizes the differential entropy among all variables with the same variance and since the variance is a lower bound on the second order moment, we obtain
\begin{align}
\avgD{\hat{Q}_j}{Q_0}
&\geq -\frac{1}{2}\log(2\pi e (\sigma^2+\sum_{i\in\calK} \theta_{i,j}))+\frac{1}{2}\log 2\pi \sigma^2  + \frac{\sigma^2+\sum_{i\in\calK} \theta_{i,j}}{2\sigma^2}\\
&=\frac{\sum_{i\in\calK}\theta_{i,j}}{2\sigma^2} - \frac{1}{2}\log\left(1+\frac{\sum_{i\in\calK} \theta_{i,j}}{\sigma^2}\right).
\end{align}
Since we can argue as done in the manuscript that every $\sum_{i\in\calK}\theta_{i,j}$ should vanish, we obtain
\begin{align}
  \avgD{\hat{Q}_n}{Q_0^{\otimes n}}
  &\geq \sum_{j=1}^n \left(\frac{(\sum_{i\in\calK}\theta_{i,j})^2}{4\sigma^4} + o\left((\sum_{i\in\calK}\theta_{i,j})^2\right)\right)\\
  &\geq \frac{1}{n}\left(\frac{(\sum_{i\in\calK}\sum_{j=1}^n\theta_{i,j})^2}{4\sigma^4} + o\left((\sum_{i\in\calK}\sum_{j=1}^n\theta_{i,j})^2\right)\right).
\end{align}
Note that the last step should be argued a bit more carefully but is nevertheless correct. Similarly, we can upper bound $\log M_k$ as
\begin{align}
  \log M_k
  &\leq \frac{1}{1-\epsilon_n}\left( \sum_{j=1}^n\frac{1}{2}\log\left(1+\frac{\theta_{k,j}}{\sigma^2}\right) + H_b(\epsilon_n)\right)\\
  &\leq \frac{1}{1-\epsilon_n}\left( \sum_{j=1}^n\frac{\theta_{k,j}}{2\sigma^2} + H_b(\epsilon_n)\right)
\end{align}
Putting everything together, one would then obtain
\begin{align}
  \frac{\log M_k}{\sqrt{n\avgD{\hat{Q}_j}{Q_0}}}
  &\leq \frac{\sum_{j=1}^n\frac{\theta_{k,j}}{2\sigma^2} + H_b(\epsilon_n)}{(1-\epsilon_n)\sqrt{\frac{(\sum_{i\in\calK}\sum_{j=1}^n\theta_{i,j})^2}{4\sigma^4} + o\left((\sum_{i\in\calK}\sum_{j=1}^n\theta_{i,j})^2\right)}}\\
  &= \frac{1}{(1-\epsilon_n)\sqrt{1+o(1))}}\left(\frac{\sum_{i=1}^n\theta_{k,i}}{\sum_{i\in\calK}\sum_{j=1}^n\theta_{i,j}}+\frac{H_b(\epsilon_n)}{\sum_{i\in\calK}\sum_{j=1}^n\theta_{i,j}}\right)
\end{align}

Reproducing the reasoning in the manuscript to deal with all the terms properly, we would then obtain that the covert capacity must be contained in the region described by
\begin{align}
\bigcup_{\{\rho_k\}_{k\in\calK}:\sum_{k}\rho_k=1}\left\{\{r_{k}\}_{k\in\calK}:r_k\leq\rho_k\right\}.
\end{align}

\bibliographystyle{IEEEtranS}

\end{document}